\documentclass[a4paper,11pt]{article}

\usepackage{xcolor}
\usepackage[latin1]{inputenc}
\usepackage{bbm}
\usepackage{mathrsfs}

\usepackage{qtree}
\usepackage{prooftree}

\usepackage{carlos-draft}
\usepackage{ordinals}
\usepackage{rewriting}
\usepackage{proof-terms}
\usepackage{peqence}

\resize

\title{Proof terms for infinitary rewriting, progress report}
\author{Carlos Lombardi, Alejandro R\'ios, Roel de Vrijer}
\newcommand{\multistepsAfterDevelopments}[1]{}
\newcommand{\multistepsIndependent}[1]{#1}
\newcommand{\usingContractionActivity}[1]{#1}
\newcommand{\usingRedseqOnly}[1]{}
\newcommand{\usingStrEqIdEq}[1]{}
\newcommand{\includeStandardisation}[1]{}
\newcommand{\doNotIncludeStandardisation}[1]{#1}
\newcommand{\denotationInOwnChapter}[1]{#1}
\newcommand{\denotationDistributed}[1]{}
\newcommand{\cfpInsideCompression}[1]{#1}
\newcommand{\cfpInsideStandardisation}[1]{}

\begin{document}

\maketitle

%
%
%
%
\section{Preliminaries}
\label{sec:prelim}
\subsection{Ordinal arithmetics}
\label{sec:ordinal-arithmetics}
 
One of the main foundations for this work is the theory of countable ordinals; (citation needed) and (citation needed) are good references on this subject. 

We want to point out some definitions and results which are critical in order to prove some of the basic properties of infinitary proof terms.

\medskip
\newcommand{\onlyLongVersion}[1]{}
\newcommand{\onlyShortVersion}[1]{#1}
\renewcommand{\ordplus}{+}
In order to deal with infinitary composition, we will need to obtain the sum of a sequence including $\omega$ ordinals. Thus we will resort to the following definition, \confer\ (citation needed).

\begin{definition}[Ordinal infinitary sum]
\label{dfn:ordinal-infAdd}
Let $\iomegaseq{\alpha_i}$ be a sequence of ordinals. We define the sum of $\iomegaseq{\alpha_i}$ as follows:
$$\underset{i < \omega}{\Sigma} \alpha_i \eqdef sup(\set{\alpha_0 \ordplus \alpha_1 \ordplus \ldots \ordplus \alpha_{n-1} \ordplus \alpha_n \setsthat n < \omega})$$
\end{definition}

The sum of $\omega$ ordinals, in the way it was just defined, enjoys the following important property.

\begin{lemma}
\label{rsl:ordinal-lt-infAdd-then-unique-representation}
Let $\iomegaseq{\alpha_i}$ be a sequence of ordinals, and $\beta$ an ordinal such that $\beta < \underset{i < \omega}{\Sigma} \alpha_i$.
Then there exist a unique $k < \omega$ and an ordinal $\gamma$ such that $\beta = \alpha_0 \ordplus \ldots \ordplus \alpha_{k-1} \ordplus \gamma$ and $\gamma < \alpha_k$.
\end{lemma}

\begin{proof}
This is an easy consequence of some properties of ordinals.
Namely, $\beta < \underset{i < \omega}{\Sigma} \alpha_i$ implies that the set $\set{k < \omega \setsthat \beta < \alpha_0 \ordplus \ldots \ordplus \alpha_k}$ is nonempty; we take $n$ as the minimum of this set.
Then 
$\alpha_0 \ordplus \ldots \ordplus \alpha_{n-1} \leq \beta < (\alpha_0 \ordplus \ldots \ordplus \alpha_{n-1}) \ordplus \alpha_n$.%
\onlyLongVersion{
\refLem{ordinal-leq-then-unique-right-add} implies that there exists an unique $\gamma$ verifying 
$(\alpha_0 \ordplus \ldots \ordplus \alpha_{n-1}) \ordplus \gamma = \beta$.
In turn, \reflem{left-add-lt-then-lt} yields $\gamma < \alpha_n$. 
}\onlyShortVersion{
Basic properties of ordinals entail the existence and uniqueness of an ordinal $\gamma$ verifying 
$(\alpha_0 \ordplus \ldots \ordplus \alpha_{n-1}) \ordplus \gamma = \beta$,
and also that $\gamma < \alpha_n$.
}%
Thus we conclude.
\end{proof}

Finally, the property of $\omega$-cofinality of countable ordinals (see e.g. (citation needed)) is critical in some proofs along this work.
We use the following version of the statement of this property.
\begin{proposition}
\label{rsl:cofinality-omega}
Let $\alpha$ be a limit countable ordinal. Then there exists a sequence%
\footnote{can we assert the existence of an \textbf{increasing} sequence $\iomegaseq{\alpha_i}$?} 
of ordinals $\iomegaseq{\alpha_i}$ such that $0 < \alpha_i < \alpha$ for all $i < \omega$, and $\alpha = \underset{i < \omega}{\Sigma} \alpha_i$.
\end{proposition}

\subsection{Positions and terms}

\begin{definition}[Position, depth of a position]
\label{dfn:pos}
A \emph{position} is a finite sequence of $\Nat_{>0}$. 
The empty sequence is denoted by the symbol $\epsilon$. 
The \emph{depth} of a position $p$, notation $\posln{p}$, is defined as its length as a sequence; observe that $\posln{\epsilon} = 0$.
\end{definition}

\begin{definition}[Concatenation of positions]
\label{dfn:pos-concatenation}
Let $p,q$ be positions. 
Then we define $p \cdot q$, the concatenation of $p$ and $q$, as follows: $\epsilon \cdot q \eqdef q$ and $(i p) \cdot q \eqdef i (p \cdot q)$.
Moreover, given $P,Q$ sets of positions, then we define also $P \cdot q \eqdef \set{p \cdot q \setsthat p \in P}$ and $p \cdot Q \eqdef \set{p \cdot q \setsthat q \in Q}$.

We will omit the dot to denote concatenation, \ie\ we will write $pq, pQ, Pq$ instead of $p \cdot q, p \cdot Q, P \cdot q$ wherever no confusion arises.
\end{definition}

\begin{definition}[Signature, function symbol, constant]
\label{dfn:signature}
A \emph{signature} is a finite set of symbols along with a function from this set to $\Nat_{\geq 0}$, called \emph{arity} and noted $\arityfn$.
The usual notation is $\Sigma \eqdef \set{f_i/n_i}_{i \in I}$, where each $f_i$ is a symbol and $n_i = \arity{f_i}$.
We will follow the custom of writing $f \in \Sigma$ as a shorthand notation for $\exists n . n \in \Nat_{\geq 0} \,\land\, f/n \in \Sigma$.

A \emph{constant} is a function symbol $c$ such that $\arity{c} = 0$.
\end{definition}

\begin{definition}[Tree domain]
\label{dfn:tree-domain}
A \emph{tree domain} is any set of positions $P$ satisfying the following conditions ($p, q$ positions; $i,j \in \Nat_{>0})$:
$P \neq \emptyset$;
$P$ is prefix closed, \ie\ $pq \in P$ implies $p \in P$ (particularly, $\epsilon \in P$); 
if $pj \in P$ and $1 \leq i < j$, then $pi \in P$.
\end{definition}

\begin{definition}[Term, positions of a term, symbol at a position, sets of finitary and infinitary terms]
\label{dfn:term}
A \emph{term} over a signature $\Sigma$ and a countable set of variables $\thevar$ is any pair $\pair{P}{F}$, such that $P$ is a tree domain, $F : P \to \Sigma \cup \thevar$, and the following condition holds:
if $p \in P$ and $F(p) = h$, then $pi \in P$ iff $i \leq ar(h)$, where we consider $ar(x) = 0$ if $x \in \thevar$%
\footnote{in some texts, \eg\ \cite{Courcelle83} and \cite{Gallier86}, a term is defined just as a function from positions to symbols; the set of positions is implicitly determined by being the domain of the function. We prefer to explicitly include the set of positions in the definition, I guess that such a decision leads to a clearer definition of terms by describing the tree domain first, and the function afterwards. I guess that we are following the idea expressed in \cite{terese} page 670, ``a (\ldots) term can be described as the set of its positions, together with a function (\ldots)''}%
.

If $t = \pair{P}{F}$ is a term, we will denote $P$ by $\Pos{t}$, and $F$ just by $t$; therefore, we will write $t(p)$ to denote $F(p)$.

A term is \emph{finite} iff its tree domain is, otherwise it is \emph{infinite.}

Given a signature $\Sigma$ and a countable set of variables $\thevar$, the set of \emph{finitary terms} over $\Sigma$, notation $Ter(\Sigma,\thevar)$, is the set of finite terms over $\Sigma$; and the set of \emph{infinitary terms} over $\Sigma$, notation $Ter^\infty(\Sigma,\thevar)$, is the set of finite or infinite terms over $\Sigma$.

We will often drop the set of variables, writing just $Ter(\Sigma)$ or $Ter^\infty(\Sigma)$.
\end{definition}

\medskip
We will name \emph{head symbol} of a term $t$ the symbol $t(\epsilon)$. The name \emph{root symbol} will be used as well.

\begin{notation}[Intuitive notation for terms]
\label{dfn:term-intuitive-notation}
An alternative notation will be often used for terms in $Ter^\infty(\Sigma,\thevar)$: if $x \in \thevar$ and $f/n \in \Sigma$, then we will write  \\
\begin{tabular}{@{$\quad\bullet\quad$}p{.9\textwidth}}
$x$ for $\pair{\set{\epsilon}}{F}$ where $F(\epsilon) = x$, and  \\
$f(t_1, \ldots, t_n)$ for $\pair{P}{F}$, where 
$P = \set{\epsilon} \cup \bigcup_{1 \leq i \leq n} \set{ip \setsthat p \in \Pos{t_i}}$, 
$F(\epsilon) = f$, and
$F(ip) = t_i(p)$.
\end{tabular}

\noindent
We will use $t \in \thevar$ as shorthand notation for $t = \pair{\set{\epsilon}}{F}$, $F(\epsilon) = x$, and $x \in \thevar$.

\noindent
If $f/1 \in \Sigma$, then we will write $f^\omega$ for the term $t = f(f(f( \ldots )))$, \ie\ $\Pos{t} = \set{1^n \setsthat n \in \Nat}$ and $t(p) = f$ for all $p \in \Pos{t}$%
\footnote{This convention could generalise to any $f/n \in \Sigma$, by defining $f^\omega = \pair{P}{F}$ where $P$ is the set of all the sequences that can be built using the numbers $\set{1,2,\ldots,n}$, and $F(p) \eqdef f$ for all $p \in P$. Roughly speaking, $f^\omega$ would be defined as the infinite tree all filled with $f$.}.
\end{notation}

We observe that all terms can be described using \refnotation{term-intuitive-notation}.

\begin{proposition}
\label{rsl:term-then-intuitive-notation}
Let $t \in Ter^\infty(\Sigma,\thevar)$. Then either $t = x$ or $t = f(t_1, \ldots, t_n)$ where $f/n \in \Sigma$ and $t_i \in Ter^\infty(\Sigma,\thevar)$ for all $i \leq n$; \confer\ \refnotation{term-intuitive-notation}.
\end{proposition}

\begin{proof}
\refDfn{tree-domain} implies that $\epsilon \in \Pos{t}$. 

\medskip
Assume $t(\epsilon) = x \in \thevar$. Moreover, assume for contradiction the existence of some $p \in \Pos{t} \sthat p \neq \epsilon$. In that case there should be some $n \in \Nat$ being the minimum of the depths of such positions, \ie\ $n = min(\posln{p} \setsthat p \in \Pos{t} \land p \neq \epsilon)$. 

Observe that $n = 1$ would imply the existence of some $i \in \Nat$ verifying $i \in \Pos{t}$, contradicting \refdfn{term} since we consider $ar(x) = 0$. In turn, $n > 1$ would entail $p = p'i \in \Pos{t}$ for some $p$ verifying $\posln{p} = n$ and $\posln{p'} > 0$, implying $p' \in \Pos{t}$ by \refdfn{tree-domain}, thus contradicting minimality of $n$.
Consequently, $\Pos{t} = \set{\epsilon}$, hence $t = x$.

\medskip
Assume $t(\epsilon) = f \in \Sigma$. 
For each $i \in \Nat$ we define $P_i \eqdef \set{p \ \setsthat ip \in \Pos{t}}$, and $F_i : P_i \to \Sigma \cup \thevar$ such that $F_i(p) \eqdef t(ip)$.

If $i \leq ar(f)$, then $P_i \neq \emptyset$ since $\epsilon \in P_i$. Moreover, $\Pos{t}$ being a tree domain implies immediately that $P_i$ enjoys the remaining conditions in \refdfn{tree-domain}; and also the condition on $F_i$ described in \refdfn{term} stems immediately from the fact that $t$ is a term. 
Therefore, $t_i \eqdef \pair{P_i}{F_i}$ is a term.

On the other hand, $i > ar(f)$ implies that $P_i = \emptyset$, thus $\Pos{t} = \set{\epsilon} \cup \bigcup_{1 \leq i \leq ar(f)}\set{ip \  \setsthat p \in P_i}$.
We conclude by observing that $t = f(t_1, \ldots, t_n)$.
\end{proof}

\begin{definition}[Occurrence]
\label{dfn:occurrence}
Let $t$ be a (either finite or infinite) term over $\Sigma$ and $a \in \Sigma \cup \thevar$. An \emph{occurrence} of $a$ in $t$ is a position $p \in \Pos{t}$ such that $t(p) = a$.
We define $\Occs{a}{t}$ as the set of occurrences of $a$ in $t$.

A symbol $a \in \Sigma \cup \thevar$ \emph{occurs in} a term $t$ iff $\Occs{a}{t} \neq \emptyset$, \ie\ iff there is at least one occurrence of $a$ in $t$; $a$ occurs exactly $n \in \Nat$ times in $t$ iff 
$|\ \Occs{a}{t} | = n$, 
where $|\ S \ |$ denotes the cardinal of any set $S$.
\end{definition}

\begin{definition}[Closed term, linear term]
\label{dfn:closed-linear}
A term $t$ is said to be \emph{closed} iff it includes no occurrences of variables; it is said to be \emph{linear} iff no variable occurs in it more than once.
\end{definition}

\begin{definition}[Subterm at a position]
\label{dfn:subtat}
Let $t = \pair{P}{F}$ be a term, and $p \in P$. We define the \emph{subterm} of $t$ at position $p$, notation $\subtat{t}{p}$, as $\pair{\subtat{P}{p}}{\subtat{F}{p}}$, where $\subtat{P}{p}$ and $\subtat{F}{p}$ are the \emph{projections} of $P$ and $F$ over $p$ respectively; \ie, 
$\subtat{P}{p} \, \eqdef \set {q \setsthat pq \in P}$ and
$\subtat{F}{p} \, : \subtat{P}{p} \ \to \Sigma \cup \thevar$ such that $\subtat{F}{p}(q) \eqdef F(pq)$.
\end{definition}

\refDfn{subtat} allow a straightforward and direct (i.e. non-inductive) proof of a basic result about subterms. Namely
\begin{lemma}
\label{rsl:subtat-composition}
$\subtat{t}{pq} = \subtat{(\subtat{t}{p})}{q}$.
\end{lemma}

\begin{proof}
If we call $\pair{P}{F} \eqdef \subtat{t}{pq}$ and $\pair{P'}{F'} \eqdef \subtat{(\subtat{t}{p})}{q}$, then \refdfn{subtat} yields \\
\minicenter{
$\begin{array}{ccl@{\qquad}ccl}
	P & = & \set{r \setsthat pqr \in \Pos{t}} & P' & = & \set{r \setsthat qr \in \Pos{\subtat{t}{p}}} \\
	F(r) & = & t(pqr) & F'(r) & = & \subtat{t}{p} (qr) = t(pqr)
\end{array}$
} \\[2pt]
We conclude by observing that $pqr \in \Pos{t} \textiff qr \in \Pos{\subtat{t}{p}}$.
\end{proof}

\noindent
Particularly, if $t = f(t_1, \ldots, t_n)$, then $\subtat{t}{ip} = \subtat{t_i}{p}$; \confer\ \refnotation{term-intuitive-notation}.

\begin{definition}[Replacement at a position]
\label{dfn:repl}
Let $\,t$ and $u$ be terms, and $p \in \Pos{t}$. We define the  \emph{replacement} of $\,t$ under position $p$ with $u$, notation $\repl{t}{u}{p}$, as $\pair{P'}{F'}$ such that
$P' \eqdef \set{q \in \Pos{t} \setsthat p \not\leq q} \cup \set{pq \setsthat q \in \Pos{u}}$ and \\
$F'(q) \eqdef \left\{
\begin{array}{r@{ \ \textiff \ }l}
	t(q) & p \not\leq q \\
	u(q') & q = pq'
\end{array}
\right.$.
\end{definition}

We state and prove some basic properties about replacement.
It is worth mentioning that the definition of term we use (\confer\ \refdfn{term}) is different from the definition in \cite{terese} for finitary terms (Dfn 2.1.2, page 26) or \cite{trat} (Dfn. 3.1.2, page 35), so that it is necessary to verify these properties.

\begin{lemma}
\label{rsl:repl-homo}
Let $t = f(t_1, \ldots, t_n)$ and $u$ be terms, and $p \in \Pos{t_i}$. 
Then $\repl{t}{u}{ip} = f(t_1, \ldots, \repl{t_i}{u}{p}, \ldots, t_n)$.
\end{lemma}

\begin{proof}
Let us call $t' = \pair{P'}{F'} \eqdef \repl{f(t_1, \ldots, t_n)}{u}{ip} \ $ and \\ $t'' = \pair{P''}{F''} \eqdef f(t_1, \ldots, \repl{t_i}{u}{p}, \ldots, t_n)$.

By joining \refnotation{term-intuitive-notation} and \refdfn{repl} we obtain 
$P' = \set{\epsilon} \cup \set{jq \setsthat q \in \Pos{t_j} \land j \neq i} \cup \set{iq' \setsthat q' \in \Pos{t_i} \land p \not\leq q'} \cup \set{ipq \setsthat q \in \Pos{u}}$. It is straightforward to verify that $P' = P''$; particularly, notice that $\Pos{\repl{t_i}{u}{p}} = \set{q' \setsthat q' \in \Pos{t_i} \land p \not\leq q'} \cup \set{pq \setsthat q \in \Pos{u}}$.

Let us compare $F'(p)$ and $F''(p)$, for any $p \in P' = P''$.
$F'(\epsilon) = F''(\epsilon) = f$. 
If $j \neq i$ then $ip \not\leq jq$, then $F'(jq) = F''(jq) = t_j(q)$. 
If $p \not\leq q'$, then $F'(iq') = F(iq') = t_i(q')$, and $F''(iq') = \repl{t_i}{u}{p}(q') = t_i(q')$.
Finally, if $q = pq'$, then $F'(iq) = u(q')$ and $F''(iq) = \repl{t_i}{u}{p}(pq') = u(q')$.
Thus we conclude.
\end{proof}

\begin{lemma}
\label{rsl:repl-ctx}
Let $t$ and $u$ be terms and $pq \in \Pos{t}$.
Then $\repl{t}{u}{pq} = \repl{t}{\repl{\subtat{t}{p}}{u}{q}}{p}$.
\end{lemma}

\begin{proof}
By induction on $p$.

If $p = \epsilon$, then both $\repl{t}{u}{pq}$ and $\repl{t}{\repl{\subtat{t}{p}}{u}{q}}{p}$ are equal to $\repl{t}{u}{q}$.

Assume that $p = i p'$, in this case $t = g(t_1, \ldots, t_n)$.
\refLem{repl-homo} implies that $\repl{t}{u}{pq} = \repl{t}{u}{ip'q} = g(t_1, \ldots, \repl{t_i}{u}{p'q}, \ldots t_n)$
and also 
$\repl{t}{\repl{\subtat{t}{p}}{u}{q}}{p} = \repl{t}{\repl{\subtat{t}{ip'}}{u}{q}}{ip'} = 
g(t_1, \ldots, \repl{t_i}{\repl{\subtat{t_i}{p'}}{u}{q}}{p'}, \ldots, t_n)$. We conclude by \ih\ on $p'$, $t_i$ and $u$.
\end{proof}

\begin{lemma}
\label{rsl:repl-disj}
Let $t,s$ be terms and $p,q \in \Pos{t}$ such that $p \disj q$.
Then $\subtat{(\repl{t}{s}{q})}{p} = \subtat{t}{p}$.
\end{lemma}

\begin{proof}
Say $t = \pair{P}{F}$, $\repl{t}{s}{q} = \pair{P'}{F'}$, $\subtat{t}{p} = \pair{P_p}{F_p}$, and $\subtat{(\repl{t}{s}{q})}{p} = \pair{P'_p}{F'_p}$.
We prove $P_p = P'_p$ by double inclusion.

\begin{tabular}{lp{.86\textwidth}}
$\subseteq )$ &
Let $p' \in P_p$, so that $p p' \in P$. Observe that $p \disj q$ implies $p p' \disj q$, so that $q \not\leq p p'$, implying $p p' \in P'$, and therefore $p' \in P'_p$. \\
$\supseteq )$ &
Let $p' \in P'_p$, so that $p p' \in P'$. We have already verified $q \not\leq p p'$, so that the only valid option \wrt\ Dfn.~\ref{dfn:repl} is $p p' \in P$, implying $p' \in P_p$.
\end{tabular}

\smallskip
Let $p' \in P'_p = P_p$, so that $p p' \in P \cap P'$ and $q \not\leq p p'$.
Dfn.~\ref{dfn:subtat} implies $F'_p(p') = F'(p p')$ and $F_p(p') = F(p p')$.
In turn, Dfn.~\ref{dfn:repl} yields $F'(p p') = F(p p')$, since $q \not\leq p p'$.
Consequently $F_p = F'_p$. Thus we conclude.
\end{proof}

\subsection{Contexts}

\begin{definition}[Context, one-hole context]
\label{dfn:ctx}
A \emph{context} over $\Sigma$ is a term (either finite or infinite) over $\Sigma \cup \set{\Box/0}$. 
A \emph{one-hole context} is a context in which the symbol $\Box$ occurs exactly once.
\end{definition}

\begin{definition}[Position of a variable/hole in a linear term/context]
\label{dfn:vpos}
Let $t$ be a term. Then we define ${\tt VOccs}(t) \eqdef \set{p \setsthat t(p) \in \thevar}$. 
Given a term $t$, if $| {\tt VOccs}(t) | = n \in \Nat$, then for any $i$ such that $1 \leq i \leq n$ we define $\VPos{t}{i}$, the $i$-th variable occurrence in $t$, as the $i$-th element of the set ${\tt VOccs}(t)$, considering the order given by $p < q$ iff $\posln{p} < \posln{q}$ or $\posln{p} = \posln{q}$, $p = rip'$, $q = rjq'$, $i < j \ $%
\footnote{orderings among positions will be studied in the analysis of different standard concepts}.

Analogously, if $C$ is a context including a finite number of occurrences of the box, then we define $\BPos{C}{i}$ as the $i$-th element of $\Occs{\Box}{C}$, considering the order just described.
\end{definition}

\begin{definition}[Context replacement]
\label{dfn:ctx-repl}
Let $C$ be a context including exactly $n$ occurrences of the box, and $t_1, \ldots, t_n$ terms. 
We define the replacement of $C$ using $t_1, \ldots, t_n$ 
as $C[t_1, \ldots, t_n] \eqdef \pair{P}{F}$, where \\
$P \eqdef \set{p \in \Pos{C} \setsthat C(p) \neq \Box} \cup \bigcup_{i} \set{\BPos{C}{i} \cdot p \setsthat p \in \Pos{t_i}}$, \\ and 
$F'(p) \eqdef
\left\{
\begin{array}{rcl}	
C(p) & \textiff & C(p) \neq \Box \\
t_i(q) & \textiff & p = \BPos{C}{i} \cdot q
\end{array}
\right.
$
\end{definition}

We remark that, given $\BPos{C}{i} \disj \BPos{C}{j}$ if $i \neq j$, and that $\repl{\repl{t}{u_1}{p}}{u_2}{q} = \repl{\repl{t}{u_2}{q}}{u_1}{p}$ if $p \disj q$, it should be possible to prove that \\
$C[t_1, \ldots, t_n] = 
\repl{
  \repl{
    \repl{C \,}{t_1}{\BPos{C}{1}\,}
    }{t_2}{\BPos{C}{2}} \ldots
   }{t_n}{\BPos{C}{n}}$. We leave the verification of this conjecture as future work.

\medskip
It is easy to verify an expected result about context replacement, namely:
\begin{lemma}
\label{rsl:ctx-repl-composition}
$\subtat{C[t_1, \ldots, t_n]}{\BPos{C}{i} \cdot p} = \subtat{t_i}{p}$
\end{lemma}

\begin{proof}
Immediate from Dfn.~\ref{dfn:ctx-repl}.
\end{proof}

\subsection{Distance between terms}

In this section, the notion of distance between terms to be used in this work, and the corresponding definition of limit of an infinite sequence of (possibly infinite) terms, are introduced.

\begin{definition}[Distance between terms, \confer\ \cite{terese} p. 670]
\label{dfn:distance}
Let $t,u$ be terms. We define the \emph{distance} between $t$ and $u$, notation $\tdist{t}{u}$, as follows: \\
\begin{tabular}{@{$\quad\bullet\quad$}p{.9\textwidth}}
$0$ iff $t = u$, and \\
$2^{-k}$ otherwise, where $k$ is the length of the shortest position at which the two terms differ; \ie\ 
$k = \posln{p} \sthat p$ is minimal for $p \in \Pos{t} \cup \Pos{u}$ and $t(p) \neq u(p)$.
\end{tabular}
\end{definition}

This definition of distance implies that, for any $t$, $u$ terms, obtaining $\tdist{t}{u} < 2^{-k}$ for all $k < \omega$ is a sufficient condition to conclude $t = u$.
In turn, to check $\tdist{t}{u} < 2^{-k}$ it is enough to verify, for any position $p$, that $\posln{p} \leq k$ and $p \in \Pos{t} \cup \Pos{u}$ entails $p \in \Pos{t} \cap \Pos{u}$ and $t(p) = t(u)$.

\begin{definition}[Limit of a sequence of terms]
\label{dfn:limit-terms}
Let $<t_i>_{i < \alpha}$ a sequence of terms where $\alpha$ is a countable limit ordinal. 
We say that the sequence $<t_i>$ has the term $t$ as its limit (notation $\lim_{i \to \alpha} t_i = t$) iff the following limit condition holds: for any $p \in \Nat$ there exists $k_p < \alpha$ such that for all $j$ satisfying $k_p < j < \alpha$, $\tdist{t_j}{t} < 2^{-p}$.
\end{definition}

Since the set of infinitary terms is proven to be equal to the metric completion of $Ter(\Sigma)$ w.r.t. the metric given by \refdfn{distance} (so it is trivially metric-complete \wrt\ that metric), given this definition of limit, if a sequence has limit then it is Cauchy-convergent \wrt\ distance.

\medskip
We observe that the set $\iSigmaTerms$ for a given signature $\Sigma$, along with the distance given in Dfn.~\ref{dfn:distance}, form an \emph{ultrametric} space%
\footnote{references here?}. 
Formally:

\begin{lemma}
\label{rsl:tdist-is-ultrametric}
Let $t, u, w$ be terms. Then $\tdist{t}{w} \leq max(\tdist{t}{u}, \tdist{u}{w})$.
\end{lemma}

\begin{proof}
If $t = u = w$, then all distances are $0$.
Oteherwise, we analyse $k$ where $max(\tdist{t}{u}, \tdist{u}{w}) = 2^{-k}$.
If $k = 0$ we conclude immediately isince the distance between any pair of terms cannot be more than one.
Assume $k = k' + 1$. Then $\tdist{t}{u} < 2^{-k'}$, implying that for any position $p$ such that $\posln{p} \leq k'$, it is easy to verify that $p \in \Pos{t}$ iff $p \in \Pos{u}$, and moreover, $p \in \Pos{t}$ implies $t(p) = u(p)$.
On the other hand, the same properties hold for $u$ \wrt\ $w$, since $\tdist{u}{w} < 2^{-k'}$.
Hence $\tdist{t}{w} \leq 2^{-k}$, thus we conclude.
\end{proof}

The distance between a term and the result of a replacement on that term is limited by the depth of the position corresponding to the replacement. Namely:

\begin{lemma}
\label{rsl:repl-dist}
Let $t,s$ be terms and $p \in \Pos{t}$. Then $\tdist{t}{\repl{t}{s}{p}} \leq 2^{-\posln{p}}$.
\end{lemma}

\begin{proof}
We proceed by induction on $p$.
If $p = \epsilon$ then we conclude immediately since $\tdist{t}{u} \leq 2^0 = 1$ for any term $u$.
Otherwise, \ie\ if $p = i p'$, observe that $i p' \in \Pos{t}$ implies $t = f(t_1, \ldots, t_i, \ldots, t_m)$.
Then $\repl{t}{s}{p} = f(t_1, \ldots, \repl{t_i}{s}{p'}, \ldots, t_m)$, \confer\ Lem.~\ref{rsl:repl-homo}, implying $\tdist{t}{\repl{t}{s}{p}} = \frac{1}{2} * \tdist{t_i}{\repl{t_i}{s}{p'}}$.
In turn, \ih\ yields $\tdist{t_i}{\repl{t_i}{s}{p'}} \leq 2^{-\posln{p'}}$.
Therefore, easy exponent arithmetics recalling $\posln{p} = \posln{p'} + 1$ suffices to conclude.
\end{proof}

\subsection{Substitutions}

\begin{definition}[Substitution]
\label{dfn:substitution}
Given a set of variables $\thevar$ and a signature $\Sigma$, a \emph{substitution} is a function $\sigma : \thevar \to \sinfterms$ where $\sigma(x) = x$ except for a finite subset of $\thevar$%
\footnote{Even when removing the finite support condition is not needed so far, I wonder whether something is broken if we consider arbitrary substitutions, allowing those with infinite support as well. -- Carlos May 25th, 2013.}
.

Any substitution is extended into a function, bearing the same name $\sigma$, where $\sigma: \sinfterms \to \sinfterms$, defined as follows:
$\sigma t \eqdef \pair{P}{F}$ where \\
$P = \set{p \in \Pos{t} \setsthat t(p) \notin \thevar} \cup \set{pq \setsthat t(p) = x \in \thevar \land q \in \Pos{\sigma x}}$ and \\
$F(p) = 
\left\{
\begin{array}{rcl}
	t(p) & \textiff & p \in \Pos{t} \land t(p) \notin \thevar \\
	\sigma x (q') & \textiff & p = q q' \land t(q) = x \in \thevar
\end{array}
\right.$
\end{definition}

\subsubsection{Uniqueness of the extension to terms}
For finitary terms, the extension of the domain of a substitution from variables to terms can be defined by relying to the concept of \emph{$\Sigma$-algebra}; \confer\ \cite{trat} Chapter 3.

Given a signature $\Sigma$, we can define a $\Sigma$-algebra whose carrier set is $Ter(\Sigma,\thevar)$, which we will denote by $Ter(\Sigma,\thevar)$ as well. For any $f/n \in \Sigma$, the corresponding function is defined simply as follows: \\
\minicenter{$f^{Ter(\Sigma,\thevar)}(t_1, \ldots, t_n) \eqdef f(t_1, \ldots, t_n)$} \\ 
(\confer\ \refprop{term-then-intuitive-notation}).
Moreover, this $\Sigma$-algebra is \emph{generated} by $\thevar$, \confer \cite{trat} dfn. 3.2.2.

A similar $\Sigma$-algebra can be defined having \sinfterms\ as carrier set. On the other hand, \sinfterms\ considered as a $\Sigma$-algebra is not generated by $\thevar$; notice that the $\Sigma$-subalgebra generated by $\thevar$ for \sinfterms\ is exactly $Ter(\Sigma,\thevar)$.

\medskip
The following result relates substitutions with the $\Sigma$-algebra \sinfterms\ in an expected way. 
In the sequel, we will distinguish between the two functions introduced in \refdfn{substitution}.
We will use $\sigma$ for the function whose domain is the set of variables, and $\wid{\sigma}$ for the function whose domain is the set of terms.

\begin{lemma}
Let $\wid{\sigma}$ be a substitution on terms. 
Then $\wid{\sigma}$ is an endomorphism on \sinfterms\ which extends the corresponding $\sigma$ defined on variables.
\end{lemma}

\begin{proof}
It is enough to show that $\wid{\sigma}(f(t_1, \ldots, t_n)) = f(\wid{\sigma}(t_1), \ldots, \wid{\sigma}(t_n)))$; \confer\ \refprop{term-then-intuitive-notation}; let us call these terms $t' = \pair{P'}{F'}$ and $t'' = \pair{P''}{F''}$ respectively.

By applying notation~\ref{dfn:term-intuitive-notation} and \refdfn{substitution}, we obtain \\[2pt]
$\begin{array}{rcl@{}l}
P' & = & 
	\set{\epsilon} \cup 
   \ \bigcup_i \ (&
     \set{ip \setsthat p \in \Pos{t_i} \land t_i(p) \notin \thevar} \ \cup 
     \\ & & & 
     \set{ipq \setsthat t_i(p) = x \in \thevar \land q \in \Pos{\sigma x}})
     \\
F'(\epsilon) & = & f \\ 
F'(ip) & = & t_i(p) & \textif p \in \Pos{t_i} \land t_i(p) \notin \thevar \\
F'(ipq) & = & \sigma x(q) & \textif t_i(p) = x \in \thevar \land q \in \Pos{\sigma x}
\end{array}
$ \\[2pt]
An analogous analysis for $P''$ and $F''$ is enough to conclude.
\end{proof}

Nonetheless, we cannot use the result on uniqueness of homomorphisms on generated $\Sigma$-algebras given the values for the generator set (\confer\ \cite{trat} lemma 3.3.1) to assert that $\wid{\sigma}$ is the only endomorphism on \sinfterms\ which extends $\sigma$. 
The reason is that \sinfterms\ is not generated by $\thevar$.

Fortunately, an analogous uniqueness result can be proved for endomorphisms on \sinfterms.

\begin{proposition}
\label{rsl:endomorphism-uniqueness-for-infinitary-terms}
Let $\Sigma$ be a signature, and $\phi, \psi$ two endomorphisms on the $\Sigma$-algebra \sinfterms\ which coincide on $\thevar$.
Then $\phi = \psi$.
\end{proposition}

\begin{proof}
We will prove the following statement, which entails the desired result (\ie\ that for any term $t$, $\psi(t) = \phi(t)$): 
for any $k < \omega$, given a term $t$ and a position $p$ such that $\posln{p} \leq k$ and $p \in \Pos{\psi(t)} \cup \Pos{\phi(t)}$, then $\psi(t)(p) = \phi(t)(p)$.
\Confer\ comment following Dfn.~\ref{dfn:distance}.

We proceed by induction on $k$. There is one case which does not need to resort to the inductive argument:
if $t \in \thevar$, then $\psi(t) = \phi(t)$ since hypotheses assert that these functions coincide on $\thevar$.

Thus assume $t = f(t_1, \ldots, t_m)$; \confer\ Prop.~\ref{rsl:term-then-intuitive-notation}.
In this case hypotheses entail $\psi(t) = f(\psi(t_1), \ldots, \psi(t_m))$ and $\phi(t) = f(\phi(t_1), \ldots, \phi(t_m))$.
If $k = 0$, then $\posln{p} \leq k$ implies $p = \epsilon$, hence it is enough to observe that $\psi(t)(\epsilon) = \phi(t)(\epsilon) = f$.
Assume $k = k' + 1$. If $\posln{p} \leq k'$ then applying \ih\ on $k'$ \wrt\ $t$ and $q$ suffices to conclude. If $\posln{p} = k$, then $p = iq$ (recall $k > 0$) where $\posln{q} = k'$ and $q \in \Pos{\psi(t_i)} \cup \Pos{\phi(t_i)}$.
Therefore we can apply \ih\ on $k'$ \wrt\ $t_i$ and $q$, obtaining $\psi(t_i)(q) = \phi(t_i)(q)$. Thus we conclude by observing $\psi(t)(p) = \psi(t_i)(q)$ and analogously for $\phi$.
\end{proof}

%
%

Consequently, we can assert that $\wid{\sigma}$ is the only endomorphism on \sinfterms\ which extends $\sigma$, as desired.

\subsection{Term rewriting systems}
\label{sec:trs}

\begin{definition}[Reduction rule, term rewriting system]
\label{dfn:trs}
Assuming a set of variables $\thevar$ and given a signature $\Sigma$, a \emph{reduction rule} (just \emph{rule} if no confusion arises) over $\Sigma$ is a pair of terms $\pair{l}{r}$ satisfying the following conditions: $l$ is a finite term, $l \notin \thevar$, and each variable occurring in $r$ occurs also in $l$. Notation for a reduction rule: $l \to r$, also $\mu: l \to r$ if assigning explicit names to rules is desirable.
The terms $l$ and $r$, respectively, are the \emph{left-hand side} and \emph{right-hand} side, \emph{lhs} and \emph{rhs} for short, of the rule $l \to r$.

A \emph{term rewriting system} (shorthand TRS) is a pair $T = \pair{\Sigma}{R}$, where $\Sigma$ is a signature and $R$ is a set of rules over $\Sigma$.

If the right-hand sides of all the rules are finite terms, then $T$ can be considered as a TRS over either $Ter(\Sigma)$ or $Ter^\infty(\Sigma)$; otherwise, only the infinitary interpretation is valid.
In either case, a TRS over $Ter^\infty(\Sigma)$ is known as a \emph{infinitary TRS}, or iTRS for short.
\end{definition}

We define that a \TRS\ is \emph{left-linear} iff for any $l$ left-hand side of a rule, and for any $x$ variable, $x$ occurs in $l$ at most once.
This work will study reductions in left-linear \iTRSs\ only.

Additionaly, we will say that a reduction rule $\mu : l \to r$ is \emph{collapsing} iff $r \in \thevar$.

\subsection{Reduction, redex occurrence}
\label{sec:reduction}

\begin{definition}[Reduction step, source, target, active position, depth]
\label{dfn:step}
Let $T = \pair{\Sigma}{R}$ be a TRS, $t \in Ter^\infty(\Sigma)$, $p \in \Pos{t}$, $\mu : l \to r \in R$ and $\sigma$ a substitution, such that $\subtat{t}{p} = \sigma l$.
Then the 4-tuple $a = \langle t, p, \mu, \sigma \rangle$ is a \emph{reduction step}.
We define $src(a) \eqdef t$, $tgt(a) \eqdef \repl{t}{\sigma r}{p}$, $\RPos{a} \eqdef p$, and $\sdepth{a} \eqdef \posln{p}$. They are, respectively, the \emph{source}, \emph{target}, \emph{redex position} and \emph{depth} of $a$.
\end{definition}
\noindent
If the source term of a reduction step is clear from the context, it can be omitted when describing the step. On the other hand, if the substitution is unimportant \wrt\ the subject being discussed, it can be omitted as well.
Therefore, we will sometimes refer to a reduction step $\langle t, p, \mu, \sigma \rangle$ as $\langle p, \mu, \sigma \rangle$, or even just $\pair{p}{\mu}$.

\medskip
Notice that, given a term $t$, the reduction steps having $t$ as source term are in an obvious bijection with the occurrences of redexes (\ie\ of subterms having the form $\sigma l$ for some rule $\mu: l \to r$) inside $t$.
Namely, the reduction step $\langle t, p, \mu, \sigma \rangle$ correspond to the occurrence, at position $p$, of a redex with rule $\mu$ and substitution $\sigma$.
Therefore, we will take the (maybe rather unusual) convention of considering reduction steps from $t$ and \textbf{redex occurrences in $t$} as synonyms.

We also want to remark that the definition of a reduction step is given in terms of the \emph{position} of the corresponding redex occurrence, opposed to the \emph{context} which surrounds it (\confer\ \cite	{terese} dfn. 2.2.4).
The choice of position is motivated by the fact that in infinitary rewriting reasonings, induction on terms (and therefore in contexts which are terms for an extended signature) is not valid, whereas induction on positions is allowed.

Finally, notice that if $t$, $p$ and $\mu$ are known in advance, then the specification of $\sigma$ is redundant. Nonetheless, I prefer to include the substitution in the definition because it will permit to describe with precision a redex occurrence whose existence is asserted.
Notice also that the inclusion of the rule is redundant for orthogonal TRSs; it is included in the characterisation of reduction steps because proof terms are intended to describe reductions in any, maybe non-orthogonal, left-linear TRS.

\medskip
A \textbf{normal form} is a term having no redex occurrences, or equivalently, a term being the source of no reduction step.

\medskip
Some examples of reduction steps follow: consider the TRS whose rules are $\mu: f(x) \to g(x)$ and $\nu: h(i(x),y) \to j(y,x)$, and the term 
$t = g \,\big( h(i(f(a)),f(i(b))) \,\big)$.
Then there are three reductions steps from $t$, namely: \\
$\langle t, 1, \nu, \set{x \eqdef f(a), y \eqdef f(i(b))} \rangle$, 
$\langle t, 111, \mu, \set{x \eqdef a} \rangle$, and
$\langle t, 12, \mu, \set{x \eqdef i(b)} \rangle$.

\bigskip
Next we will give a precise formal definition for the concept of \emph{\redseq}.
Producing a precise definition is needed, particularly since proof terms are meant as a tool to study precisely \redseqs.
Formal definitions of infinitary \redseqs\ are given and discussed throughout the literature on the subject, \confer\ \eg\ \cite{orthogonal-itrs-90}, \cite{orthogonal-itrs-95}, \cite{terese}, \cite{inf-normalization}.
%

\medskip
A \emph{\redseq} will be defined as a sequence of reduction steps, having any (finite or infinite) ordinal as length. 
This approach, and also the idea of concatenating reduction sequences, is in line with the description given in \cite{terese}, Sec. 2. We quote from page 38
\begin{quote}
Concatenating reduction steps we have (possibly infinite) \emph{reduction sequences} $t_0 \to t_1 \to t_2 \ldots$, or \textit{reductions} for short.
\end{quote}
Notice that in the definition which follows, focus is set on \emph{steps} rather than \emph{terms}.

Not all sequences of steps are \redseqs; some conditions must hold.
Obviously, if $\stepa$ and $\stepb$ are consecutive steps in a sequence, then $tgt(\stepa)$ must coincide with $src(\stepb)$.
This coherence condition must hold also for steps having \emph{limit} positions in the sequence. \Eg\ in a sequence $\stepa_0 ; \stepa_1 ; \ldots ; \stepa_n; \ldots, \stepa_\omega \ldots$, there must be some relation between the step $\stepa_\omega$ and the sequence of the steps previous to it.
This relation is commonly formalised in the literature by asking the sequence of targets of the previous steps, \ie\ the sequence $tgt(\stepa_0); tgt(\stepa_1); \ldots; tgt(\stepa_n); \ldots$ to have a limit, and that limit to coincide with $src(\stepa_\omega)$. This requirement is related with the characterisation of \emph{weakly convergent} infinitary rewriting.  

In order to obtain a notion of \redseq\ enjoying some desired properties, a further condition is imposed. Namely, the \emph{depth} of successive steps is required to tend to $\omega$ at each limit in the sequence, \ie\ up to the $\omega$-th step, up to the $\omega * 2$-th step, and so on. \Redseqs\ for which this requirement, and also the coherence requirements described before, hold, are known as \emph{strongly convergent} in the literature.

These considerations motivate the following definitions.

\begin{definition}[Reduction sequence, convergence]
\label{dfn:sred}
A (well-formed) \emph{\redseq} is: either $\redid{t}$, the \emph{empty \redseq\ for} the term $t$, or else a non-empty sequence of reduction steps $\reda \eqdef \langle \redel{\reda}{\alpha} \rangle_{\alpha < \beta}$, where $\beta > 0$ and $\reda$ verifies all the following conditions:
\begin{enumerate}
\item \label{it:dfn-sred-successor-coherence}
For all $\alpha$ such that $\alpha + 1 < \beta$, $src(\redel{\reda}{\alpha+1}) = tgt(\redel{\reda}{\alpha})$.
\item \label{it:dfn-sred-limit}
For all limit ordinals $\beta_0 < \beta$:
	\begin{enumerate}
	\item \label{it:dfn-sred-limit-existence}
	The sequence $\langle tgt(\redel{\reda}{\alpha}) \rangle_{\alpha < \beta_0}$ has a limit.
	\item \label{it:dfn-sred-limit-coherence}
	$\lim_{\alpha \to \beta_0} tgt(\redel{\reda}{\alpha}) = src(\redel{\reda}{\beta_0})$.
	\item \label{it:dfn-sred-depth}
	For all $n < \omega$, there exists $\beta' < \beta_0$ such that $\sdepth{\redel{\reda}{\alpha}} > n$ if $\beta' < \alpha < \beta_0$.	
	\end{enumerate}	
\end{enumerate}
We say that a reduction sequence $\reda$ is \emph{convergent} iff either $\reda = \redid{t}$ for some term $t$, or else $\reda = \langle \redel{\reda}{\alpha} \rangle_{\alpha < \beta}$, and either $\beta$ is a successor ordinal, or else $\beta$ is a limit ordinal and conditions (\ref{it:dfn-sred-limit-existence}) and (\ref{it:dfn-sred-depth}) hold for $\beta$ as well.
\end{definition}

\begin{definition}[Source of a \redseq]
\label{dfn:redseq-src}
Let $\reda$ be a \redseq. We define the \emph{source} term of $\reda$, notation $src(\reda)$, as follows:
if $\reda = \redid{t}$, then $src(\reda) \eqdef t$, 
if $\reda = \langle \redel{\reda}{\alpha} \rangle_{\alpha < \beta}$, then $src(\reda) \eqdef src(\redel{\reda}{0})$. 
\end{definition}

\begin{definition}[Target of a \redseq]
\label{dfn:redseq-tgt}
Let $\reda$ be a convergent \redseq. We define the the \emph{target} term of $\reda$, notation $tgt(\reda)$, as follows:
if $\reda = \redid{t}$, then $tgt(\reda) \eqdef t$; 
if $\reda = \langle \redel{\reda}{\alpha} \rangle_{\alpha < \beta}$, then $\beta = \beta' + 1$ implies $tgt(\reda) \eqdef tgt(\redel{\reda}{\beta'})$, and $\beta$ being a limit ordinal implies $tgt(\reda) \eqdef \lim_{\alpha \to \beta} tgt(\redel{\reda}{\alpha})$.
\end{definition}

\begin{definition}[Length of a \redseq]
\label{dfn:redseq-length}
Let $\reda$ be a \redseq. We define the \emph{length} of $\reda$, notation $\redln{\reda}$, as follows:
if $\reda = \redid{t}$, then $\redln{\reda} \eqdef 0$, 
if $\reda = \langle \redel{\reda}{\alpha} \rangle_{\alpha < \beta}$, then $\redln{\reda} \eqdef \beta$. 
\end{definition}

\begin{definition}[Minimum activity depth of a \redseq]
\label{dfn:redseq-mind}
Let $\reda$ be a \redseq. We define the \emph{minimum activity depth} of $\reda$, notation $\mind{\reda}$, as follows:
if $\reda = \redid{t}$, then $\mind{\reda} \eqdef \omega$, 
if $\reda = \langle \redel{\reda}{\alpha} \rangle_{\alpha < \beta}$, then $\mind{\reda} \eqdef min \set{\sdepth{\redel{\reda}{\alpha}} \setsthat \alpha < \beta}$. 
\end{definition}

\begin{definition}[Section of a \redseq]
\label{dfn:redseq-section}
Let $\reda$ be a \redseq\ and $\alpha, \beta$ ordinals verifying $\alpha < \redln{\reda}$, $\beta \leq \redln{\reda}$ and $\alpha \leq \beta$.
We define the \emph{section} of $\reda$ from $\alpha$ to $\beta$, notation $\redsublt{\reda}{\alpha}{\beta}$, as follows: 
if $\alpha = \beta < \redln{\reda}$, then $\redsublt{\reda}{\alpha}{\beta} \eqdef \redid{src(\redel{\reda}{\alpha})}$, otherwise, \ie\ if $\alpha < \beta$, then $\redsublt{\reda}{\alpha}{\beta} \eqdef \langle \redel{\reda}{\alpha+\gamma} \rangle_{\alpha+\gamma < \beta}$.
\end{definition}

\begin{remark}
\label{rmk:tgt-mentioned-then-defined}
Any mention of $tgt(\reda)$ implies that the target of the \redseq\ $\reda$ is defined, \ie\ that $\reda$ is a \emph{convergent} \redseq.
\end{remark}

\medskip
It is worth remarking that the requirement about depths of successive steps, \ie\ condition~(\ref{it:dfn-sred-depth}) in Dfn.~\ref{dfn:sred}, is not enough to guarantee the well-formedness of \redseqs. 
Let us discuss briefly this issue. Some examples will be given using the rules $f(x) \to g(x)$, $h(x) \to j(x)$, and $g(x) \to f(x)$, and denoting concatenation of sequences by semicolons.

Depth requirement alone does not guarantee coherence at limit positions, as discussed before defining \redseqs. \Eg, the sequence of steps $f\om \infred g\om ; h\om \infred j\om$, which length is $\omega * 2$, does not produce a well-formed reduction sequence, even when depths tend to infinity at each limit ordinal in the sequence of steps; a target (namely $g\om$) can be determined for the prefix of the first $\omega$ steps, but it does not coincide with the source of the $\omega$-th step, \ie\ $h\om$.

Moreover, the depth condition alone does not even guarantee the existence of a limit for each limit ordinal prefix. \Eg\ consider the sequence of steps, having length $\omega^2$, informally described as follows:
$f\om \infred g\om ; g\om \infred f\om ; g(f\om) \infred g\om ; f(g\om) \infred f\om ; g^2(f\om) \infred g\om ; f^2(g\om) \infred f\om ; \ldots g^n(f\om) \infred g\om ; f^n(g\om) \infred f\om ; \ldots \ $.
This sequence of steps obeys the depth condition at each limit ordinal, including $\omega^2$ itself, but even though, a limit cannot be determined for it.
Therefore, the requirement about the existence of a limit, \ie\ condition~(\ref{it:dfn-sred-limit-existence}), cannot be removed by the mere fact of including the depth requirement.

It could possibly be proved, by means of a careful transfinite induction on limit ordinals, that for any sequence of steps, and each limit ordinal $\beta$ up to the length of that sequence, the depth requirement on each limit ordinal $\leq \beta$, plus coherence (\ie\ condition~(\ref{it:dfn-sred-limit-coherence})) at all limit ordinals $< \beta$, imply the existence of a limit in the sequence of targets at ordinal $\beta$.
Since this issue is not in the focus of the present work, we leave it as subject of further investigation.

\medskip
Notice that the way in which the concept of \redseq\ is formalised here differs from the approach taken in \cite{terese}, Sec. 8.2, 
which cannot be adapted for infinitary rewriting (perhaps unless coinduction is involved, \confer\ \cite{endrullis-wir-2013}) since the construction of \redseqs\ is based there on simple induction, and therefore can only describe finite sequences.
As the correspondence between proof terms and \redseqs\ given in \cite{terese} is based on the mentioned characterisation, this observation suggests that the adequacy between proof terms and \redseqs\ for the infinitary case should be treated in a way different than what is described in \cite{terese}.

\medskip
Given a term $t$, we will refer to the reduction sequences having $t$ as source term as the reduction sequences \textbf{from} $t$.
Moreover, if $s$ is the only normal form verifying $t \infred s$, then we will say that $s$ is the (infinitary) \textbf{normal form of $t$}.

\bigskip
We can define reduction steps and sequences which model applications of rules to contexts rather than terms.

\begin{remark}
\label{rmk:trs-ctx}
For any TRS $T = \pair{\Sigma}{R}$ we can think of an associated TRS $T^\Box \eqdef \pair{\Sigma \cup \set{\Box/0}}{R}$, which makes it possible to describe reductions on contexts.
In the sequel we will include references to reduction steps and reduction sequences whose source and target are contexts; they must be understood as defined in $T^\Box$.
\end{remark}

\subsection{Patterns, pattern depth}
\label{sec:patterns}
Given a reduction rule $\mu: l \to r$ and a reduction step $\stepa = \langle t,p,\mu,\sigma \rangle$, the role of the \emph{function symbol} occurrences in $l$ differs from that of the \emph{variable} occurrences: 
the former must be present explicitly in $src(a)$ having the same structure as in $l$; 
while the latter are included in the domain of $\sigma$.

We will sometimes need to refer to the positions of all the occurrences of function symbols in (the lhs of) a rule, and also in (the source term of) a reduction step.
\Eg\ if $\mu = f(g(x,h(y)) \to f(y)$, then the occurrences of function symbols in (the lhs of) $\mu$ are at positions $\epsilon$, $1$ and $12$.
The corresponding formal definitions follow.

\begin{definition}[Pattern, pattern positions, pattern depth]
\label{dfn:patt}
Let $t$ be a term. The \emph{pattern} of $t$, notation $\patt{t}$, is the context which results of changing all the variable occurrences in $t$ with boxes; \confer\ \cite{terese} dfn. 2.7.3, pg. 49.
The set of \emph{pattern positions} of $t$, notation $\PPos{t}$ is defined as $\set{p \setsthat p \in \Pos{t} \textand t(p) \notin \thevar}$.
The \emph{pattern depth} of $t$, notation $\Pdepth{t}$, is defined as $max(\set{|p| \setsthat p \in \PPos{t}})$; if $x \in \thevar$ then $\Pdepth{x}$ is undefined.

Let $\mu:l \to r$ be a reduction rule. The set of pattern positions and the pattern depth of $\mu$ are defined as follows: 
$\PPos{\mu} \eqdef \PPos{l}$, $\Pdepth{\mu} \eqdef \Pdepth{l}$.

Let $a = \langle t,p,\mu,\sigma \rangle$ be a reduction step. The set of pattern positions of $a$ is defined as follows: 
$\PPos{a} \eqdef p \cdot \PPos{\mu}$.
\end{definition}

For example, if $\mu : h(i(x),g(i(y)),c) \to h(x,x,y)$, $t = g(h(i(g(a)),g(i(b)),c))$ and $a = \langle t,1,\mu,\set{x \eqdef g(a), y \eqdef b}\rangle$, then
$\PPos{\mu} = \set{\epsilon,1,2,21,3}$, $\Pdepth{\mu} = 2$, and $\PPos{a} = 1 \cdot \PPos{\mu} = \set{1,11,12,121,13}$.

\includeStandardisation{
\bigskip\noindent
Sometines we will need to consider the maximum rule pattern depth of a TRS.

\begin{definition}[Maximum rule pattern depth]
\label{dfn:maximum-rule-pd}
Let $T = \langle \Sigma, R \rangle$ be a (finitary or infinitary) TRS. 
We define the maximum rule pattern depth of $T$ as follows:
$\maxPdepth{T} \eqdef max(\set{\Pdepth{\mu} \setsthat \mu \in R})$.
\end{definition}
}
\subsection{Some properties about infinitary rewriting}
We include in this Section the statement and proof of some properties on infinitary rewriting which are needed in following Sections.
In turn, these properties require some definitions to be given.

\medskip
We say that a term $t$ is \emph{infinitary weakly normalising}, shorthand notation $WN^\infty$, iff there exists at least one \redseq\ $\reda$ such that $t \infredx{\reda} u$ and $u$ is a normal form. We say that a term $t$ is \emph{strongly normalising}, shorthand notation $SN^\infty$, iff there is no divergent \redseq\ whose source term is $t$. 
A term $t$ has the \emph{unique normal-form property}, shorthand notation $UN^\infty$, iff whenever $t \infred u_1$, $t \infred u_2$ and both $u_1$ and $u_2$ are normal forms, then $u_1 = u_2$.
A \TRS\ is $WN^\infty$ ($SN^\infty$, $UN^\infty$) iff all its terms are.
\Confer\ \cite{inf-normalization} for a study of normalisation for infinitary rewriting.

\medskip
A \TRS\ is \emph{left-linear} iff for any rule $\mu$ and for any $x \in \thevar$, $x$ occurs in the left-hand side of $\mu$ at most once.
A \TRS\ \trst\ is \emph{orthogonal} iff it is left-linear and there is no term $t$ such that $t = \sigma_1 l_1$ and $\subtat{t}{p} = \sigma_2 l_2$, where $l_1$ and $l_2$ are left-hand sides of rules in \trst, and $p \in \PPos{l_1}$.

Some examples of left-hand sides of rules leading to non-orthogonal \TRSs\ follow.
No \TRS\ including a rule whose left-hand side is $f(g(x))$ and another having as left-hand side either $g(x)$ or $g(h(x))$, is orthogonal: $t = f(g(h(a)))$ is a counterexample for the corresponding condition.
Also, no \TRS\ including rules whose left-hand sides are $h(f(x),y)$ and $h(x,g(y))$ is orthogonal, a counterexample is $t = h(f(a),g(b))$. In this case the position $p$ mentioned in the definition is $\epsilon$ for the given counterexample.
Finally, no \TRS\ including a rule whose left-hand side is $f(f(x))$ is orthogonal, a counterexample is $t = f(f(f(a)))$. In this case the same rule corresponds to $l_1$ and $l_2$.

Properties of first-order infinitary orthogonal \TRSs\ are studied \eg\ in \cite{orthogonal-itrs-95}.

\medskip
A \TRS\ $\trst$ is \emph{disjoint} iff 
the set of all the function symbols occurring in the left-hand sides of the rules of $\trst$
is disjoint from
the set of all the function symbols occurring in the right-hand sides of the rules of $\trst$.

\medskip
The results to be given in this Section are particularly needed for the study of the class of proof terms corresponding to coinitial sets of redexes, which involves the definition of TRSs which are `companions' to the TRS under study. \Confer\ the concept of \emph{2-rewriting system}, notation 8.2.12 in \cite{terese}.

The `companion' TRSs enjoy some desirable properties.
First of all, they are all orthogonal, and therefore they enjoy the property $UN^\infty$; \confer\ \cite{inf-normalization} Section 5.
Some of them are Recursive Program Schemes (\confer\ \cite{terese} dfn. 3.4.7), \ie, they are orthogonal and all their rules have the form $f(\ldots, x_i, \ldots) \to t$, so that we can distinguish the subset $\mathcal{F} \eqdef \set{f \setsthat f(\ldots, x_i, \ldots) \to t \in R}$ within their signature.
Furthermore, the following additional restriction is imposed.

Notice that for Recursive Program Schemes, the disjointness condition amounts to assert that no symbol in $\mathcal{F}$ appears in the right-hand side of any rule.

\bigskip
Sections of \redseqs, \confer\ Dfn.~\ref{dfn:redseq-section}, enjoy some basic properties.

\begin{lemma}
\label{rsl:redseq-proper-section-convergent}
Let $\reda$ be a \redseq, and $\alpha < \redln{\reda}$. Then $\redupto{\reda}{\alpha}$ is convergent.
\end{lemma}

\begin{proof}
It is immediate to verify that $\redupto{\reda}{\alpha}$ is a well-formed \redseq.
If $\alpha = 0$, \ie\ $\redupto{\reda}{\alpha} = \redid{src(\reda)}$, or if $\alpha$ is a successor ordinal, then it is immediately convergent.
If $\alpha$ is a limit ordinal, the fact that $\reda$ is well-formed implies that conditions~(\ref{it:dfn-sred-limit-existence}) and (\ref{it:dfn-sred-depth}) hold for $\alpha < \redln{\reda}$, hence $\redupto{\reda}{\alpha}$ is convergent.
\end{proof}

\begin{lemma}
\label{rsl:redseq-section-src-tgt-coherence}
Let $\reda$ be a \redseq\ and $\alpha < \redln{\reda}$. Then $src(\redel{\reda}{\alpha}) = tgt(\redupto{\reda}{\alpha})$.
\end{lemma}

\begin{proof}
Notice that Lem.~\ref{rsl:redseq-proper-section-convergent} implies that $\redupto{\reda}{\alpha}$ is convergent, so that its limit is defined.
If $\alpha = 0$, \ie\ $\redupto{\reda}{\alpha} = \redid{src(\redel{\reda}{0})}$, then we conclude immediately.
Otherwise, $\alpha = \alpha' + 1$ implies $src(\redel{\reda}{\alpha}) = tgt(\redel{\reda}{\alpha'})$, and $\alpha$ limit implies $src(\redel{\reda}{\alpha}) = \lim_{\alpha' \to \alpha} tgt(\redel{\reda}{\alpha'})$, \confer\ conditions~(\ref{it:dfn-sred-successor-coherence}) and (\ref{it:dfn-sred-limit-coherence}) resp. in Dfn.~\ref{dfn:sred}. In either case, this coincides with $tgt(\redupto{\reda}{\alpha})$, \confer\ Dfn.~\ref{dfn:redseq-tgt}. Thus we conclude.
\end{proof}

\medskip
We prove some expected properties of targets of convergent \redseqs.

\begin{lemma}
\label{rsl:redseq-mind-big-src-tgt}
Let $\reda$ be a convergent \redseq\ and $n < \omega$ such that $\mind{\reda} > n$.
Then $\tdist{src(\reda)}{tgt(\reda)} < 2^{-n}$.
\end{lemma}

\begin{proof}
We proceed by induction on $\redln{\reda}$. 
If $\redln{\reda} = 0$, \ie\ $\reda = \redid{t}$ for some term $t$, then $tgt(\reda) = src(\reda) = t$, so that we conclude immediately.

Assume that $\redln{\reda}$ is a successor ordinal, so that $\reda = \reda'; \stepa \ $ where $\redln{\reda'} < \redln{\reda}$.
Then \ih\ can be applied to obtain 
$\tdist{src(\reda')}{tgt(\reda')} = \tdist{src(\reda)}{src(\stepa)} < 2^{-n}$.
In turn, $tgt(\reda) = tgt(\stepa) = \repl{src(\stepa)}{s}{p}$ for some term $s$, where $p = \RPos{\stepa}$, so that hypotheses imply $\mind{\stepa} > n$.
Then Lem.~\ref{rsl:repl-dist} implies $\tdist{src(\stepa)}{tgt(\reda)} \leq 2^{-\posln{p}} < 2^{-n}$.
Hence Lem.~\ref{rsl:tdist-is-ultrametric} allows to conclude.

Assume that $\alpha \eqdef \redln{\reda}$ is a limit ordinal.
In this case $tgt(\reda) = \lim_{\alpha' \to \alpha} tgt(\redel{\reda}{\alpha'})$.
Let $\alpha_n < \alpha$ such that $\tdist{tgt(\redel{\reda}{\alpha'})}{tgt(\reda)} < 2^{-n}$ if $\alpha_n < \alpha' < \alpha$.
Then particularly 
$\tdist{tgt(\redel{\reda}{\alpha_n + 1})}{tgt(\reda)} 
    = \tdist{tgt(\redupto{\reda}{\alpha_n + 2})}{tgt(\reda)} 
		< 2^{-n}$; recall $\alpha_n < \alpha$ limit implies $\alpha_n + 2 < \alpha$.
In turn, \ih\ can be applied on $\redupto{\reda}{\alpha_n + 2}$ to obtain $\tdist{src(\redupto{\reda}{\alpha_n + 2})}{tgt(\redupto{\reda}{\alpha_n + 2})} < 2^{-n}$.
Hence we conclude by Lem.~\ref{rsl:tdist-is-ultrametric}.
\end{proof}

\begin{lemma}
\label{rsl:redseq-disj}
Let $t \infredx{\reda} u$ and $p \in \Pos{t}$ such that $\RPos{\redel{\reda}{\alpha}} \disj p$ for all $\alpha < \redln{\reda}$.
Then $\subtat{t}{p} = \subtat{u}{p}$.
\end{lemma}

\begin{proof}
We proceed by induction on $\redln{\reda}$.
If $\redln{\reda} = 0$, \ie\ $\reda = \redid{t}$, then we conclude immediately since $u = t$.

Assume that $\redln{\reda}$ is a successor ordinal, so that $t \infredx{\reda'} u' \sstepx{\stepa} u$. In this case, \ih\ applies to $\reda'$, yielding $\subtat{t}{p} = \subtat{u'}{p}$.
In turn, $u = \repl{u'}{s}{q}$ for some term $s$, where $q = \RPos{\stepa} \disj p$. Then Lem.~\ref{rsl:repl-disj} implies $\subtat{u'}{p} = \subtat{u}{p}$. Thus we conclude.

Assume that $\alpha \eqdef \redln{\reda}$ is a limit ordinal. 
In this case $u = tgt(\reda) = \lim_{\alpha' \to \alpha} tgt(\redel{\reda}{\alpha'})$.
Let $n < \omega$, and $\alpha_n < \alpha$ such that 
$\tdist{tgt(\redel{\reda}{\alpha'})}{u} < 2^{-(n + \posln{p})}$, implying 
$\tdist{\subtat{tgt(\redel{\reda}{\alpha'})}{p}}{\subtat{u}{p}} < 2^{-n}$, if $\alpha_n < \alpha' < \alpha$. Particularly, $\tdist{\subtat{tgt(\redel{\reda}{\alpha_n+1})}{p}}{\subtat{u}{p}} < 2^{-n}$.
Recall that $\alpha_n < \alpha$ limit implies $\alpha_n + i < \alpha$ if $i < \omega$. Then \ih\ can be applied to $\redupto{\reda}{\alpha_n + 2}$, yielding 
$\subtat{src(\redupto{\reda}{\alpha_n + 2})}{p} = tgt(\subtat{\redupto{\reda}{\alpha_n + 2})}{p}$, so that $\subtat{t}{p} = \subtat{tgt(\redel{\reda}{\alpha_n + 1})}{p}$. Hence $\tdist{\subtat{t}{p}}{\subtat{u}{p}} < 2^{-n}$ for all $n < \omega$. 
Consequently, we conclude.
\end{proof}

\medskip
The just introduced properties allow to define the \emph{projection} of a \redseq\ not including head steps over an index.
We verify that the definition yields a well-formed \redseq; in the infinitary setting, this verification involves a fair amount of work.
The following definition involves the use of a sequence of non-contiguous ordinals which we will call $A$. We use $\seqln{A}$ and $\seqel{A}{\alpha}$ to  denote the order type of $A$ and its $\alpha$-th element respectively, where $\alpha < \seqln{A}$.
In turn, this sequence is built from a set of ordinals $S$ as follows. If $S = \emptyset$, then $A$ is the empty sequence, so that $\seqln{A} = 0$. Otherwise, we define $\seqel{A}{0}$ as the minimal element of $S$. 
Let $\alpha > 0$ such that $\seqel{A}{\alpha'}$ is defined for all $\alpha' < \alpha$. If $\alpha = \alpha' + 1$ then we consider the set $\set{\beta \in S \setsthat \beta > \seqel{A}{\alpha'}}$, and if $\alpha$ is a limit ordinal then we consider $\set{\beta \in S \setsthat \beta \geq sup(\set{\seqel{A}{\alpha'} \setsthat \alpha' < \alpha})}$. 
In either case, if the considered set is empty then we state that $\seqel{A}{\alpha_1}$ as undefined for all $\alpha_1 \geq \alpha$, so that $\seqln{A} = \alpha$. Otherwise, we define $\seqel{A}{\alpha}$ as the minimum of the considered set.

\newcommand{\redai}{\proj{\reda}{i}}
\begin{definition}
\label{dfn:proj-redseq}
Let $\reda$ a \redseq\ such that $\mind{\reda} > 0$, and $i$ such that $1 \leq i \leq m$ where $src(\reda) = f(t_1, \ldots, t_m)$.
We define the \emph{projection of $\reda$ over $i$}, notation $\redai$, as the \redseq\ whose specification follows.

Let $A$ be sequence built from the set $\set{ \alpha \setsthat \alpha < \redln{\reda} \,\land\, i \leq \RPos{\redel{\reda}{\alpha}} }$, \wrt\ the usual order of ordinals.
If $A$ is empty, then $\redai \eqdef \redid{t_i}$.
Otherwise $\redln{\redai} \eqdef \seqln{A}$, and $\redel{(\redai)}{\alpha} \eqdef \langle s_i, p, \mu \rangle$ where $\redel{\reda}{\seqel{A}{\alpha}} = \langle f(s_1, \ldots, s_i, \ldots, s_m), ip, \mu \rangle$.
Observe that Lem.~\ref{rsl:redseq-proper-section-convergent} implies $\redupto{\reda}{\seqel{A}{\alpha}}$ to be convergent, and in turn Lem.~\ref{rsl:redseq-mind-big-src-tgt} implies $tgt(\redupto{\reda}{\seqel{A}{\alpha}})(\epsilon) = src(\reda)(\epsilon) = f$; therefore, $tgt(\redupto{\reda}{\seqel{A}{\alpha}}) = src(\redel{\reda}{\seqel{A}{\alpha}}) = f(s_1, \ldots, s_i, \ldots, s_m)$. \Confer\ also Lem.~\ref{rsl:redseq-section-src-tgt-coherence}.
\end{definition}

\begin{lemma}
\label{rsl:proj-redseq-well-defined}
Let $\reda$ be a \redseq\ such that $\mind{\reda} > 0$, and $i$ such that $1 \leq i \leq m$ where $src(\reda)(\epsilon) = f/m$.
Then $\redai$ is a well-formed \redseq\ and $src(\redai) = \subtat{src(\reda)}{i}$.
Moreover, if $\reda$ is convergent, then $\redai$ is convergent as well, and $tgt(\redai) = \subtat{tgt(\reda)}{i}$.
\end{lemma}

\begin{proof}
Let $A$ be the sequence of positions of steps in $\reda$ at or below position $i$. We proceed by induction on $\redln{\redai} = \seqln{A}$.

Assume $A$ is empty, so that $\redai = \redid{\subtat{src(\reda)}{i}}$.
Then just Dfn.~\ref{dfn:sred} implies immediately that $\redai$ is a well-formed and convergent \redseq, and Dfn.~\ref{dfn:redseq-src} that $src(\redai) = \subtat{src(\reda)}{i}$.
If $\reda$ is convergent, then observe that $A$ being empty implies $\RPos{\redel{\reda}{\alpha}} \disj i$ for all $\alpha < \redln{\reda}$; recall $\mind{\reda} > 0$. Then Lem~\ref{rsl:redseq-disj} implies $\subtat{tgt(\reda)}{i} = \subtat{src(\reda)}{i} = tgt(\redai)$. Thus we conclude.

\medskip
Assume that $\seqln{A} = \alpha + 1$, \ie, $\seqln{A}$ is a successor ordinal.
Observe that $\redupto{(\redai)}{\alpha} = \proj{\redupto{\reda}{\seqel{A}{\alpha}}}{i}$, and that Lem.~\ref{rsl:redseq-proper-section-convergent} implies that $\redupto{\reda}{\seqel{A}{\alpha}}$ is convergent.
Then \ih\ on $\redupto{\reda}{\seqel{A}{\alpha}}$ yields that $\redupto{(\redai)}{\alpha}$ is a well-formed and convergent \redseq, that $src(\redai) = src(\redupto{(\redai)}{\alpha}) = \subtat{src(\reda)}{i}$, and that $tgt(\redupto{(\redai)}{\alpha}) = \subtat{tgt(\redupto{\reda}{\seqel{A}{\alpha}})}{i} = \subtat{src(\redel{\reda}{\seqel{A}{\alpha}})}{i}$, \confer\ Lem.~\ref{rsl:redseq-section-src-tgt-coherence}.
On the other hand, $src(\redel{(\redai)}{\alpha}) = \subtat{src(\redel{\reda}{\seqel{A}{\alpha}})}{i}$. 

We verify that the conditions in Dfn.~\ref{dfn:sred} hold for $\redai$. The analysis depends on $\alpha$. \\[5pt]
\begin{tabular}{@{$\ \ \bullet\quad$}p{.9\textwidth}}
If $\alpha = 0$, then $\redupto{\redai}{\alpha} = \redid{\subtat{src(\reda)}{i}}$. In this case, conditions~(\ref{it:dfn-sred-successor-coherence}) and (\ref{it:dfn-sred-limit}) hold immediately. 
\\[2pt]
If $\alpha = \alpha' + 1$, then $\redupto{(\redai)}{\alpha}$ being a well-formed \redseq\ implies that condition~(\ref{it:dfn-sred-successor-coherence}) holds for all $\alpha_0$ such that $\alpha_0 + 1 < \alpha$; \ie\ for all needed indexes but $\alpha'$. In turn, 
$tgt(\redel{(\redai)}{\alpha'}) = tgt(\redupto{(\redai)}{\alpha}) = \subtat{src(\redel{\reda}{\seqel{A}{\alpha}})}{i} = src(\redel{(\redai)}{\alpha}) = src(\redel{(\redai)}{\alpha'+ 1})$. On the other hand, $\redupto{(\redai)}{\alpha}$ being well-formed implies also that condition~(\ref{it:dfn-sred-limit}) holds for $\redai$; indeed, $\alpha_0 < (\alpha' + 1) + 1$ and $\alpha_0$ limit implies $\alpha_0 < \alpha' + 1$.
\\[2pt]
If $\alpha$ is a limit ordinal, then $\redupto{(\redai)}{\alpha}$ being a well-formed \redseq\ implies that condition~(\ref{it:dfn-sred-successor-coherence}) holds for $\redai$; notice $\alpha_0 + 1 < \alpha + 1$ implies $\alpha_0 < \alpha$, so that $\alpha$ limit implies in turn $\alpha_0 + 1 < \alpha$.
Furthermore, $\redupto{(\redai)}{\alpha}$ being convergent implies that conditions~(\ref{it:dfn-sred-limit-existence}) and (\ref{it:dfn-sred-depth}) hold for all $\alpha_0$ limit ordinals verifying $\alpha_0 < \alpha + 1$, particularly for $\alpha$; and also that condition~(\ref{it:dfn-sred-limit-coherence}) holds for all limit $\alpha_0 < \alpha$. 
In turn, $\lim_{\alpha' \to \alpha} tgt(\redel{(\redai)}{\alpha'}) = tgt(\redupto{(\redai)}{\alpha}) = \subtat{src(\redel{\reda}{\seqel{A}{\alpha}})}{i} = src(\redel{(\redai)}{\alpha})$, so that condition~(\ref{it:dfn-sred-limit-coherence}) to hold also for $\alpha$
\end{tabular} \\[5pt]
Hence, in either case, we have verified that $\redai$ is a well-formed \redseq.
In turn, $\redln{\redai} = \seqln{A}$ being a successor ordinal implies immediately that $\redai$ is convergent. 

If $\reda$ is convergent, then we must verify $tgt(\redai) = \subtat{tgt(\reda)}{i}$.
Let $\redel{(\redai)}{\alpha} = \langle t_i, p, \mu \rangle$ where $\redel{\reda}{\seqel{A}{\alpha}} = \langle f(t_1, \ldots, t_i, \ldots, t_m), ip, \mu \rangle$.
Then $tgt(\redai) = tgt(\redel{(\redai)}{\alpha}) = \repl{t_i}{s}{p}$ for some term $s$, and $tgt(\redel{\reda}{\seqel{A}{\alpha}}) = \repl{f(t_1, \ldots, t_i, \ldots, t_m)}{s}{ip} = $ \\ $f(t_1, \ldots, \repl{t_i}{s}{p}, \ldots, t_m)$, \confer\ Lem.~\ref{rsl:repl-homo}, therefore $tgt(\redai) = \subtat{tgt(\redel{\reda}{\seqel{A}{\alpha}})}{i}$.
If $\redln{\reda} = \seqel{A}{\alpha} + 1$, then $tgt(\reda) = tgt(\redel{\reda}{\seqel{A}{\alpha}})$.
Otherwise, for all $\alpha'$ verifying $\seqel{A}{\alpha} < \alpha' < \redln{\reda}$, it is immediate that $\RPos{\redel{\reda}{\alpha'}} \disj i$. Then Lem.~\ref{rsl:redseq-disj} implies $\subtat{tgt(\redsublt{\reda}{\seqel{A}{\alpha} + 1}{\redln{\reda}})}{i} = \subtat{src(\redsublt{\reda}{\seqel{A}{\alpha} + 1}{\redln{\reda}})}{i}$.
In either case, $\subtat{tgt(\reda)}{i} = \subtat{tgt(\redel{\reda}{\seqel{A}{\alpha}})}{i} = tgt(\redai)$.
Thus we conclude.

\medskip
Assume that $\alpha \eqdef \seqln{A}$ is a limit ordinal.

Let $\alpha'$ such that $\alpha' + 1 < \alpha$, then $\alpha$ limit implies $\alpha' + 2 < \alpha$.
Therefore \ih\ can be applied to obtain that $\redupto{(\redai)}{\alpha' + 2}$ is a well-formed \redseq, implying that $src(\redel{(\redai)}{\alpha' + 1}) = tgt(\redel{(\redai)}{\alpha'})$. 
Consequently, $\redai$ verifies condition~(\ref{it:dfn-sred-successor-coherence}) in Dfn.~\ref{dfn:sred}.

Let $\alpha_0$ be a limit ordinal verifying $\alpha_0 < \alpha$.
Observe that $\seqel{A}{\alpha_0} < \redln{\reda}$, then Lem.~\ref{rsl:redseq-proper-section-convergent} implies that $\redupto{\reda}{\seqel{A}{\alpha_0}}$ is convergent.
We apply \ih\ to obtain that $\redupto{(\redai)}{\alpha_0}$ is a well-formed and convergent \redseq. Therefore conditions~(\ref{it:dfn-sred-limit-existence}) and (\ref{it:dfn-sred-depth}) hold for $\redai$ \wrt\ $\alpha_0$.
Moreover $\lim_{\alpha' \to \alpha_0} tgt(\redel{(\redai)}{\alpha'}) = tgt(\redupto{(\redai)}{\alpha_0}) = src(\redel{(\redai)}{\alpha_0}$, \confer\ Dfn.~\ref{dfn:redseq-tgt} and Lem.~\ref{rsl:redseq-section-src-tgt-coherence} resp.. Hence $\redai$ enjoys condition~(\ref{it:dfn-sred-limit-coherence}) \wrt\ $\alpha_0$ as well.

Consequently, $\redai$ is a well-formed \redseq.
Observe that $src(\redai) = src(\redel{(\redai)}{0}) = src(\proj{\redupto{\reda}{\seqel{A}{1}}}{i})$. Since obviously $1 < \alpha$, we can use \ih\ to obtain $src(\redai) = \subtat{src(\redupto{\reda}{\seqel{A}{1}})}{i} = \subtat{src(\reda)}{i}$.

Assume that $\reda$ is convergent. 
Let $B \eqdef \set{\beta' \setsthat \beta' < \redln{\reda} \,\land\, \seqel{A}{\alpha'} < \beta' \textforall \alpha' < \alpha}$.
We define $\beta$ as follows: $\beta \eqdef \redln{\reda}$ if $B$ is empty, and $\beta \eqdef min(B)$ otherwise.
Assume for contradiction that $\beta = \beta' + 1$ for some $\beta'$.
If $B$ is empty, so that $\redln{\reda} = \beta' + 1$, then $\beta' \notin B$ implies the existence of some $\alpha' < \alpha$ such that $\beta' \leq \seqel{A}{\alpha'}$ and then $\beta' < \seqel{A}{\alpha' + 1}$, contradicting $\seqel{A}{\alpha'+1} < \redln{\reda}$.
Otherwise $\beta = min(B)$, implying that $\beta' \leq \seqel{A}{\alpha'}$ for some $\alpha' < \alpha$. But this would imply $\beta \leq \seqel{A}{\alpha' + 1}$, contradicting $\beta \in B$.
Consequently, $\beta$ is a limit ordinal.

We verify conditions~(\ref{it:dfn-sred-limit-existence}) and (\ref{it:dfn-sred-depth}) for $\redai$ \wrt\ $\alpha$. \\[5pt]
\begin{tabular}{@{$\ \ \bullet\quad$}p{.9\textwidth}}
To verify condition~(\ref{it:dfn-sred-limit-existence}), it is enough to show that $\lim_{\alpha' \to \alpha} tgt(\redel{(\redai)}{\alpha'}) = \subtat{u}{i}$, where $u = \lim_{\beta' \to \beta} tgt(\redel{\reda}{\beta'}) = tgt(\redupto{\reda}{\beta})$.
Let $n < \omega$, and $\beta_n < \beta$ such that $\tdist{tgt(\redel{\reda}{\beta'})}{u} < 2^{-(n+1)}$, implying $\tdist{\subtat{tgt(\redel{\reda}{\beta'})}{i}}{\subtat{u}{i}} < 2^{-n}$, if $\beta_n < \beta' < \beta$.
Then $\beta_n < \beta$ implies that $\beta_n \leq \seqel{A}{\alpha_n}$ for some $\alpha_n < \alpha$, then $\alpha_n < \alpha' < \alpha$ implies $
\tdist{tgt(\redel{(\redai)}{\alpha'})}{\subtat{u}{i}} < 2^{-n}$, recalling that $tgt(\redel{(\redai)}{\alpha'}) = \subtat{tgt(\redel{\reda}{\seqel{A}{\alpha}})}{i}$.
Consequently, $\lim_{\alpha' \to \alpha} tgt(\redel{(\redai)}{\alpha'}) = \subtat{tgt(\redupto{\reda}{\beta})}{i}$, and then $\redai$ verifies condition~(\ref{it:dfn-sred-limit-existence}) \wrt\ $\alpha$. 
\\[2pt]
Let $n < \omega$, let $\beta_n < \beta$ such that $\sdepth{\redel{\reda}{\beta'}} > n + 1$ if $\beta_n < \beta' < \beta$. By an argument similar to that used for condition~(\ref{it:dfn-sred-limit-existence}), we obtain the existence of some $\alpha_n < \alpha$ such that $\sdepth{\redel{\reda}{\seqel{A}{\alpha'}}} > n + 1$, implying $\sdepth{\redel{(\redai)}{\alpha'}} > n$, if $\alpha_n < \alpha' < \alpha$.
Consequently, $\redai$ verifies condition~(\ref{it:dfn-sred-depth}) for $\alpha$.
\end{tabular} \\[5pt]
Hence, $\redai$ is a convergent \redseq.
In turn, Dfn.~\ref{dfn:redseq-tgt} yields $tgt(\redai) = \lim_{\alpha' \to \alpha} tgt(\redel{(\redai)}{\alpha'})$, then we have already verified that $tgt(\redai) = \subtat{tgt(\redupto{\reda}{\beta})}{i}$.
If $\beta = \redln{\reda}$, then immediately $tgt(\redai) = \subtat{tgt(\reda)}{i}$.
Otherwise, it is immediate to observe that $\RPos{\redel{\reda}{\beta'}} \disj i$ if $\beta \leq \beta' < \redln{\reda}$. 
Hence $tgt(\redai) 
  = \subtat{tgt(\redupto{\reda}{\beta})}{i}
	= \subtat{src(\redsublt{\reda}{\beta}{\redln{\reda}})}{i}
	= \subtat{tgt(\redsublt{\reda}{\beta}{\redln{\reda}})}{i}
	= \subtat{tgt(\reda)}{i}$; 
by already obtained result, Lem.~\ref{rsl:redseq-section-src-tgt-coherence} (recall $src(\redsublt{\reda}{\beta}{\redln{\reda}}) = src(\redel{\reda}{\beta})$, Lem.~\ref{rsl:redseq-disj}, and simple analysis of Dfn.~\ref{dfn:redseq-tgt} resp..
Thus we conclude.
\end{proof}

\medskip
The following result extends the idea of a projection of a \redseq\ from arguments of function symbols to arguments of contexts.

\begin{lemma}
\label{rsl:redseq-respects-src-tgt}
Let $C$ a context having exactly $m$ holes, and $C[t_1, \ldots, t_m] \infredx{\reda} u$, such that for all $\alpha < \redln{\reda}$, there exists some $i$ verifying $1 \leq i \leq m$ and $\BPos{C}{i} \leq \RPos{\redel{\reda}{\alpha}}$.
Then $u = C[u_1, \ldots, u_m]$ and for all $i$ such that $1 \leq i \leq m$, there is a \redseq\ $\reda_i$ verifying $t_i \infredx{\reda_i} u_i$.
\end{lemma}

\begin{proof}
Straightforward induction on $max \set{\posln{\BPos{C}{i}} \setsthat 1 \leq i \leq m}$, resorting on Lem.~\ref{rsl:proj-redseq-well-defined} for the inductive case.
\end{proof}

\medskip
Two properties about normalisation follow.

\begin{lemma}
\label{rsl:orthogonal-leading-head-steps-to-nf}
Let $\trst$ an orthogonal \TRS, and $t,s,u$ terms such that $t \infredx{\redc} u$, $t \sredx{\reda} s$, $u$ is a normal form, and $\sdepth{\redel{\reda}{i}} = 0$ for all $i < \redln{\reda}$.
Then $s \infredx{\redc'} u$ for some \redseq\ $\redc'$.
\end{lemma}

\begin{proof}
We proceed by induction on $\redln{\reda}$, observe that $\reda$ is finite, and then only numeral induction is needed.
If $\redln{\reda} = 0$, \ie\ $\reda$ is the empty reduction for $t$, then $s = t$ so that we conclude by taking $\redc' \eqdef \redc$.

Assume $\redln{\reda} = n + 1$, so that $t \sstepx{\stepa} s_0 \sredx{\reda'} s$ where $\stepa = \langle t, \epsilon, \mu \rangle$ for some rule $\mu: l[x_1, \ldots, x_m] \to h$, and $\redln{\reda'} = n$.

We will resort to a result presented and proved in \eg\ \cite{orthogonal-itrs-95} and \cite{terese}, where it is called \emph{Strip Lemma}%
\footnote{in \cite{orthogonal-itrs-90}, a preliminary version of \cite{orthogonal-itrs-95}, the same property is called \emph{Parallel Moves Lemma}}%
. This result implies that whenever $t \infredx{\redb} t'$ and $t \sstepx{\stepb} s_0$, then $t' \infredx{\stepb_r} s'$ and $s_0 \infredx{\redb_r} s'$, 
where $\stepb_r$ is the \emph{residual} of $\stepb$ \emph{after} $\redb$%
\footnote{the statement in \cite{terese}, and also in \cite{orthogonal-itrs-90}, describe also the nature of $\redb_r$. We will not give the details here since they are not needed for this proof.}
. The result of the lemma can be described graphically as follows: \\[5pt]
\newcommand{\arInfR}{
\ar@{*{}*@{-}*{\hspace{-3mm}\scriptscriptstyle >\hspace{-1mm}>\hspace{-1mm}>}}[r]
}
\newcommand{\arInfD}{
\ar@{*{}*@{-}*{\textnormal{
\rotatebox{270}{$\hspace{-4mm}\scriptscriptstyle >\hspace{-1mm}>\hspace{-1mm}>$}
}}}[d]
}
$
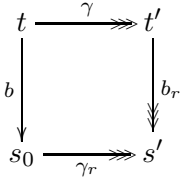
\xymatrix@C=12mm@R=12mm{
t \arInfR^{\redb} \ar[d]_{\stepb} & t' \arInfD^{\stepb_r} \\
s_0 \arInfR_{\redb_r} & s'
}$
%

While we will not include here the formal definition of residual, we mention a feature valid for orthogonal \TRSs\ which is crucial for this proof.
Assume $\stepb = \langle t, \epsilon, \mu \rangle$ such that $\mu: l \to h$, and $\stepc = \langle t, p, \nu \rangle$ where $p \neq \epsilon$ and $t \sstepx{\stepc} v$. Then $t = l[t_1, \ldots, t_m]$, $q \leq p$ for some $q$ such that $l(q) \in \thevar$, and therefore $v = l[t'_1, \ldots, t'_m]$.
In this case, there is exactly one residual of $\stepb$ after $\stepc$, namely $\langle v, \epsilon, \mu \rangle$.
This property carries on for the residual of $\stepb$ after a reduction $\redd$ where $\mind{\redd} > 0$, even if $\redln{\redd}$ is a limit ordinal.
Graphically: \\[5pt]
$
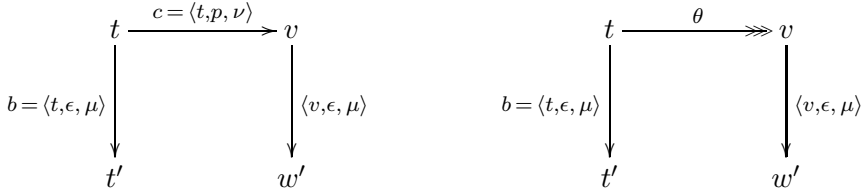
\xymatrix@C=18mm@R=15mm{
t \ar[r]^{\stepc \,=\,  \langle t, p, \, \nu \rangle} \ar[d]_{\stepb \,=\, \langle t, \epsilon, \, \mu \rangle} &
v \ar[d]^{\langle v, \epsilon, \, \mu \rangle} 
\\
t' & w'
}
\qquad \qquad
\xymatrix@C=18mm@R=15mm{
t \arInfR^{\redd} \ar[d]_{\stepb \,=\, \langle t, \epsilon, \, \mu \rangle} &
v \ar[d]^{\langle v, \epsilon, \, \mu \rangle} 
\\
t' & w'
}
$

\medskip
We return to the proof. Observe that $t = l[v_1, \ldots, v_m]$ since $\langle t, \epsilon, \mu \rangle$ is a redex. 
Then a simple transfinite induction yields that $\redc$ not including any root step would imply $u = l[v'_1, \ldots, v'_m]$, contradicting that $u$ is a normal form.
Let $\alpha$ be the minimum index corresponding to a root step in $\redc$.
Then the described property of residuals implies that $\stepa$ has exacly one residual after $\redupto{\redc}{\alpha}$, which is $\stepa' \eqdef \langle t_\alpha, \epsilon, \mu \rangle$ where $t_\alpha$ is the target term of $\redupto{\redc}{\alpha}$.
Moreover, $\redel{\redc}{\alpha}$ being a root step implies that the rule used in that step is also $\mu$, \ie\ $\redel{\redc}{\alpha} = \langle t_\alpha, \epsilon, \mu \rangle = \stepa'$. Therefore we can build the following graphic: \\[5pt]
$
\xymatrix@C=28mm@R=15mm{
t \arInfR^{\redupto{\redc}{\alpha}} \ar[d]_{\stepa \,=\, \langle t, \, \epsilon, \, \mu \rangle} &
t_{\alpha} \ar[r]^{\redel{\redc}{\alpha} \,=\, \langle t_\alpha, \, \epsilon, \, \mu \rangle}\ar[d]_{\stepa' \,=\, \langle t_\alpha, \, \epsilon, \, \mu \rangle} &
t_{\alpha + 1} \arInfR^{\redsublt{\redc}{\alpha+1}{\redln{\redc}}} \ar@{=}[d] &
u
\\
s_0 \arInfR_{\redc_1} &
t_{\alpha+1} \ar@{=}[r] & 
t_{\alpha+1} \arInfR_{\redsublt{\redc}{\alpha+1}{\redln{\redc}}} &
u
}
$ 

Hence \ih\ on $s_0 \sredx{\reda'} s$ suffices to conclude.
\end{proof}

\begin{proposition}
\label{rsl:disjoint-then-isn}
Let \trst\ be a disjoint \TRS\ which does not include collapsing rules. Then \trst\ has the property $SN^\infty$.%
\footnote{I guess that this property can be generalised to any \TRS\ in which the sets of \textbf{head} symbols of lhss and rhss are disjoint, with exactly the same proof. I don't know whether change the statement of the proposition, which is used through this text only for disjoint \TRSs.}
\end{proposition}

\begin{proof}
First we prove the following auxiliary result: for any \redseq\ \reda, limit ordinal $\beta$ such that $\beta \leq \redln{\reda}$, and $n < \omega$, 
\begin{equation}
\begin{array}{rl}
\textif & \exists \beta_1 < \beta \sthat \forall i \ (\beta_1 < i < \beta \textnormal{ implies } \sdepth{\redel{\reda}{i}} \geq n) \\ 
\textthen & \exists \beta' < \beta \sthat \forall i \ (\beta' < i' < \beta \textnormal{ implies } \sdepth{\redel{\reda}{i'}} > n)
\end{array}
\label{eq:disjoint-sn-01}
\end{equation}
Assume for any $\reda$, $\beta$ and $n$ that the premise holds.
The term $src(\redel{\reda}{\beta}) = tgt(\redsublt{\reda}{\beta_1}{\beta})$ can include only a finite number of redexes at depth $n$.
Additionally, the hypothesis yields that any reduction step included in $\redsublt{\reda}{\beta_1}{\beta}$, say $\redel{\reda}{j}$, satisfies $\sdepth{\redel{\reda}{j}} \geq n$, and moreover leaves at its redex position (\confer\ \refdfn{step})
a symbol not being the head symbol of a left-hand side, since $T$ is disjoint and it does not include collapsing rules.
Therefore, 
no redex occurrence can be created at depth $n$, implying that any reduction step at depth exactly $n$ included in $\redsublt{\reda}{\beta_1}{\beta}$ must correspond to a redex occurrence already included in $src(\redel{\reda}{\beta_1})$ and being at the same position.
Consequently, if we call $k$ the number of steps at depth exactly $n$ included in $\redsublt{\reda}{\beta_1}{\beta}$, we obtain $k < \omega$.
Thus we conclude the proof of the auxiliary result by taking $\beta'$ to be the ordinal such that $\redel{\reda}{\beta'}$ is the last of such steps if $k > 0$, and $\beta' \eqdef \beta_1$ if $k = 0$.

\medskip
Now we prove, for any \redseq\ \reda\ in $T$, that \reda\ is convergent; \ie\ that for any $n < \omega$ and $\beta$ limit ordinal such that $\beta \leq \redln{\reda}$, 
\begin{equation}
\exists \beta' < \beta \sthat \forall i \ (\beta' < i < \beta \textnormal{ implies } \sdepth{\redel{\reda}{i}} > n)
\label{eq:disjoint-sn-02}
\end{equation}
We conclude the proof of the proposition by proving \refeqn{disjoint-sn-02} by induction on $n$.
If $n = 0$, then the premise of \refeqn{disjoint-sn-01} holds taking $\beta_1 = 0$, then we conclude by \refeqn{disjoint-sn-01}.
If $n > 0$, then the premise of \refeqn{disjoint-sn-01} holds for some $\beta_1$ by \ih\ of \refeqn{disjoint-sn-02} considering $n - 1$ instead of $n$, then we conclude again by \refeqn{disjoint-sn-01}.
\end{proof}

\newpage
\section{Proof terms}
\label{sec:pterm}
The intent of the definition of proof terms is to provide a tool to formally denote, or witness, \redseqs\ in infinitary rewriting. 
Proof terms are, indeed, terms, in a signature extending that of the \iTRS\ whose \redseqs\ are to be described.
This \TRS\ will be referred to as the \textbf{object} \TRS\; we will also use the terminology `object terms' and `object \redseqs' analogously.
As already noticed, the scope of this work is limited to \emph{left-linear} \iTRSs.

The proof terms for infinitary rewriting we introduce in this Section generalise the definition given in \cite{terese} for finitary first-order rewriting, \confer\ their Dfn. 8.2.18.
The idea of using terms to denote reduction sequences has been proposed also for simply-typed lambda-calculus in \cite{Hilken96}, and for higher-order rewriting in \cite{bruggink2008}.

For each proof term we define: its \emph{source} and \emph{target} which are object terms, if it is \emph{convergent}, and its \emph{minimum depth}. All these concepts refer to the \redseqs\ which are denoted by the proof term.

\bigskip
In this section, a formal definition of the set of infinitary proof terms for a given \iTRS\ will be given. 
Then a simplified transfinite induction principle on the set of valid proof terms is given. The form of induction we introduce allows for simpler proofs for many properties to be verified in the rest of this work. Also, we will verify that proof terms enjoy some basic properties.

The definition of the set of proof terms is extensive, because it is given in two different stages, and also some auxiliary notions need to be defined simultaneously.
Therefore, we give firstly an informal introduction to the idea of proof term, and how it is used to describe the reduction space of a \TRS.

\bigskip
For each reduction rule in the object TRS, a \emph{rule symbol} is introduced in the signature for proof terms. The arity of a rule symbol coincides with the number of different variables occurring in the left-hand side of the rule it represents.
\Eg, the signature of proof terms for a first-order \TRS\ \trst\ including the rules $f(x) \to g(x)$, $h(j(x),j(y)) \to f(x)$ and $g(x) \to k(x)$ adds the rule symbols $\mu/1$, $\nu/2$ and $\rho/1$, corresponding respectively to each of the described rules. 
We describe some valid proof terms along with the \trst-reductions they denote $\mu(a) : f(a) \to g(a)$, $g(\nu(a,b)) : g(h(j(a),j(b))) \to g(f(a))$, $h(\mu(a), \mu(b)) : h(f(a),f(b)) \sred h(g(a),g(b))$, $\nu(\mu(a),b) : h(j(f(a)),j(b)) \sred f(g(a))$.

In the infinitary setting, infinite proof terms denote \redseqs\ involving infinite terms, and/or having infinite length. 
We give some examples of infinite proof terms corresponding to the the \TRS\ \trst\ introduced in the previous paragraph: 
$\mu(j\om) : f(j\om) \to g(j\om)$, $\mu\om : f\om \infred g\om$%
\footnote{In the following, a formal way to compute the source and target corresponding to any proof term will be developed.}
.

\medskip
Proof terms, as described up to this point, can be used to denote arbitrarily complex \emph{developments}, \ie\ \redseqs\ in which all the contracted redexes are present in its source term.
On the other hand, dealing with the contraction of redexes which are \emph{created} by previous steps in a \redseq\ require the idea of \emph{concatenation}, or \emph{composition}, to be taken into account in the definition of proof terms.
This proposal takes from \cite{terese} the idea of describing concatenation by means of a binary symbol which is added to the signature of proof terms. This symbol is called ``the dot'', because of its graphical representation, \ie\ $\comp$.

Some examples of proof terms including occurrences of the dot follow:
$\mu(a) \comp \rho(a) : f(a) \to g(a) \to k(a)$, 
$\nu(\mu(b),c) \comp \mu(\rho(b)) \comp \rho(k(b)) : 
		h(j(f(b)), j(c)) \sred f(g(b)) \sred g(k(b)) \sred k(k(b))$, 
$\mu\om \comp \rho\om : f\om \infred g\om \infred k\om$.
As the concatenation symbol have no special ``status'' in the signature, it can be freely combined with rule as well as object symbols, \eg\
$j(\mu(a) \comp \rho(a)) : j(f(a)) \sred j(k(a))$ denotes a two-step contraction being ``local'' to the argument of the $j$ symbol, while $\nu(\mu(a) \comp \rho(a), b) : h(j(f(a)),j(b)) \sred f(k(a))$ denote a parallelism between an outer step and the concatenation of two inner steps.

We observe that not any term in the extended signature correspond to a valid proof term. 
Each occurrence of the dot imposes a coherence condition: (the \redseqs\ corresponding to) its operands must be composable. 
\Eg\ neither $\mu(a) \comp \nu(a,a)$ nor $\mu(a) \comp \rho(b)$ are valid proof terms, because the step $f(a) \to g(a)$ is not left-composable, neither with $h(j(a),j(a)) \to f(a)$, nor with $g(b) \to h(b)$.
Therefore, some rules must be provided in order to specify the subset of valid proof terms out of the set of all terms corresponding to the extended signature. As suggested by the just given example, these rules will be related with the occurrences of the dot.

\medskip
We want to stress that the denotational capabilities of proof terms allow for a great variety in the description of reductions.
Particularly, parallel/nested steps can be explicitly described, and thus differentiated from its sequential counterparts.
\Eg, the proof terms $\mu(f(a)) \comp g(\mu(a))$ and $\mu(\mu(a))$ are different, so that the model of reductions given by proof terms allow to recognise $f(f(a)) \to g(f(a)) \to g(g(a))$ and $f(f(a)) \mulstep g(g(a))$ as different objects in the reduction space of the same \TRS.
Furthermore, as we have already observed, proof terms allow to combine in different ways the concatenation symbol with the other symbols in the extended signature.
This capability brings new ways to differentiate subtly different reductions, by describing them using different proof terms.
These considerations motivate the following assertion: proof terms denote different forms of \emph{contraction activity}, a concept broader than that of \emph{\redseq}.

We claim that proof terms as a way to describe contraction activity allow for a very detailed study of the reduction/derivation space of a calculus.

\subsection{Multisteps}
\label{sec:mstep}
\label{sec:mstep-defs}
Since the restrictions on the set of valid proof terms pertain to the dot occurrences, ``dotless'' proof terms form the foundation from which the definition of proof terms is built. 

We will give the name \emph{multistep} to any proof terms without dot occurrences.
As we have discussed in the informal introduction, multisteps correspond to sets of coinitial redexes. 
We have also seen that sequencing is explicitly denoted in proof terms by means of the concatenation symbol, \ie\ the dot.
Therefore, multisteps are intended to denote the contraction activity consisting in the \textbf{simultaneous} contraction of a set of redexes, \ie\ a multistep, \confer\ \cite{terese}, Dfn. 4.5.11.. Hence the name we have given to the proof terms to be defined next.

In the sequel, we define the set of \emph{\imsteps}, along with some basic features of a multistep, namely: how to determine its \emph{source} and \emph{target} terms, whether it is \emph{convergent} or not, and its \emph{minimum activity depth}.
These concepts are needed to properly define the restrictions to be imposed to occurrences of the dot in the general definition of the set of proof terms.

\begin{definition}[Signature for multisteps]
\label{dfn:sigma-r}
Let $T = \pair{\Sigma}{R}$ be a (either finitary or infinitary) TRS. We define the signature for the \imsteps\ over $T$ as follows:
$\Sigma^R := 
\Sigma \cup 
\set{ \mu/n \setsthat \mu : l \to r \in R \land |FV(l)| = n }$ 
.
\end{definition}

\begin{definition}[\Imsteps] 
\label{dfn:imstep}
The set of \imsteps\ for an iTRS $T \pair{\Sigma}{R}$ is exactly the set of the closed (\confer\ Dfn.\ref{dfn:closed-linear}) terms%
\footnote{By restricting \imsteps, and later proof terms (\confer\ Sec.~\ref{sec:pterm}) to be closed terms, we follow the idea expressed in \cite{terese}, Remark~8.2.21 (pg. 324): ``Since here we are interested in \peqence, we may simply assume that reductions/proof terms are closed.''.
Moreover, this decision simplifies our treatment of \peqence\ given in Sec.~\ref{sec:peqence}.
}
in $Ter^\infty(\Sigma^R)$.
\end{definition}

To define the source and target terms of a multistep, we define `companion' ad-hoc iTRSs; \confer\ \refsec{trs}.

\begin{definition}[$src_T$, $tgt_T$]
\label{dfn:srct-tgtt}
Let $T = \pair{\Sigma}{R}$ be a (either finitary or infinitary) TRS. We define the TRSs $src_T$ and $tgt_T$ as follows.
The signature of both $src_T$ and $tgt_T$ is $\Sigma^R$.
The rules of $src_T$ are $\set{\mu(x_1, \ldots, x_n) \to l[x_1, \ldots, x_n] \setsthat \mu : l \to r \in R}$.
The rules of $tgt_T$ are $\set{\mu(x_1, \ldots, x_n) \to r[x_1, \ldots, x_n] \setsthat \mu : l \to r \in R}$.
\end{definition}

We remark that for any object TRS $T$, both $src_T$ and $tgt_T$ are orthogonal and disjoint; moreover, $src_T$ does not include collapsing rules, since the lhs of a reduction rule cannot be a variable (\confer\ \refdfn{trs}). 
Therefore, both $src_T$ and $tgt_T$ enjoy the property $UN^\infty$ (\confer\ the comment about $UN^\infty$ in \refsec{trs}) and $src_T$ enjoys also $SN^\infty$ (\confer\ \refprop{disjoint-then-isn}).
Consequently, any \imstep\ has exactly one $src_T$-normal form, and at most one $tgt_T$-normal form.
This observations entail the soundness of the following definition.

\begin{definition}[Source and target of an \imstep]
\label{dfn:src-tgt-imstep}
Let $\psi$ be an \imstep. 
Then we define $src(\psi)$ to be the $src_T$-normal form of $\psi$.
Moreover, if $\psi$ is weakly normalising in $tgt_T$, we define $tgt(\psi)$ to be the corresponding normal form; otherwise, $tgt(\psi)$ is undefined.
\end{definition}

For the kind of 
\multistepsAfterDevelopments{object \redseqs}
\multistepsIndependent{contraction activity}
we intend to denote with \imsteps, it is correct to identify convergence with existence of target. Formally:

\begin{definition}[Convergent \imsteps]
\label{dfn:imstep-convergence}
An \imstep\ $\psi$ is convergent iff $\tgtt(\psi)$ is defined.
\end{definition}

\begin{definition}[Minimum activity depth of an \imstep]
\label{dfn:dmin-imstep}
Let $\psi$ be an \imstep. 
We define the minimum activity depth of $\psi$, notation $\mind{\psi}$, as follows. \\
If $\psi$ does not include occurrences of rule symbols, \ie\ if it is a term in $Ter^\infty(\Sigma)$, then $\mind{\psi} \eqdef \omega$. \\
Otherwise $\mind{\psi}$ is the minimum $n$ such that exists at least one position $p$ verifying $\psi(p) = \mu$ where $\mu$ is a rule symbol, and $n = \posln{p}$. This case admits an equivalent inductive definition based on \refnotation{term-intuitive-notation}: \\
\hspace*{1cm}
$
\begin{array}{rcl}
	\mind{f(\psi_1 \ldots \psi_n)} & := & 1 + min(\mind{\psi_1} \ldots \mind{\psi_n}) \\
	\mind{\mu(\psi_1 \ldots \psi_n)} & := & 0 
\end{array}
$
\end{definition}

\multistepsAfterDevelopments{As examples of the definitions we have just given, we show \imsteps\ corresponding to each \orthoredexset\ in the examples of \refsec{dev-max-dev}, along with the computation of their source and target terms.
We recall the rules of the object iTRS considered: }
\multistepsIndependent{In the following, we will give some examples of \imsteps. We will consider the following object rules: }
$\mu: f(i(x),y) \to h(y)$, $\nu: g(x) \to x$, $\rho: a \to b$, $\pi: m(x) \to n(x)$, $\tau: n(x) \to f(x,x)$. Then the rules of the companion iTRSs are \\
$src_T$:
$\mu(x,y) \to f(i(x),y)$ \quad $\nu(x) \to g(x)$ \quad $\rho \to a$ \quad
$\pi(x) \to m(x)$ \quad $\tau(x) \to n(x)$ \\
$tgt_T$:
$\mu(x,y) \to h(y)$ \quad $\nu(x) \to x$ \quad $\rho \to b$ \quad
$\pi(x) \to n(x)$ \quad $\tau(x) \to f(x,x)$

\multistepsAfterDevelopments{Each set of redexes $A_i$ is represented by the underlined term $usrc(A_i)$.}
\multistepsIndependent{For each of the examples, we show the source term, underlining the head symbols of some of its redexes, and the \imstep\ denoting contraction of underlined redexes. Then we develop source and target computation.
To keep notation compact, we will omit parenthesis for unary symbols.}

\begingroup
\multistepsIndependent{\renewcommand{\ulnrule}[2]{\uln{#1}}}

\begin{itemize}
\item 
The \imstep\ corresponding to 
$h(\ulnrule{f}{\mu}(i\ulnrule{a}{\rho},n\ulnrule{m}{\pi} b))$ is 
$\psi_1 \eqdef h(\mu(\rho,n \pi b))$. Computation of $src(\psi_1)$ and $tgt(\psi_1)$ follow: \\
$\psi_1 = h(\mu(\rho,n \pi b)) \ssteptrs{src_T} h(f(i \rho, n \pi b)) \ssteptrs{src_T} h(f(i a, n \pi b)) \ssteptrs{src_T} h(f(i a, n m b))$ \\
$\psi_1 = h(\mu(\rho,n \pi b)) \ssteptrs{tgt_T} hhn \pi b \ssteptrs{tgt_T} hhn n b$.

\item 
$\psi_2 \eqdef \pi^\omega$ corresponds to 
${\ulnrule{m}{\pi}}^\omega$. Let us compute source and target: \\
$\psi_2 = \pi^\omega \ssteptrs{src_T} m(\pi^\omega) \ssteptrs{src_T} mm(\pi^\omega) \infredtrs{src_T} m^\omega$ \\
$\psi_2 = \pi^\omega \ssteptrs{tgt_T} n(\pi^\omega) \ssteptrs{tgt_T} nn(\pi^\omega) \infredtrs{tgt_T} n^\omega$.

\item 
$\psi_3 \eqdef \nu^\omega$ corresponds to ${\ulnrule{g}{\nu}}^\omega$. \\
The computation of source runs as in the previous case: 
$\psi_3 = \nu^\omega \infredtrs{src_T} g^\omega$.
On the other hand, the target of all $tgt_T$ redex occurrences in $\nu^\omega$ (namely, $\langle 1^i, \nu(x) \to x, \set{x \to \nu^\omega} \rangle$) is again $\nu^\omega$. Therefore $tgt(\psi_3)$ is undefined.

\item 
Finally, $\psi_4 = h(\mu(\nu^\omega, \rho))$ corresponds to  
$h(\ulnrule{f}{\mu}(i {\ulnrule{g}{\nu}}^\omega, \ulnrule{a}{\rho}))$. \\
Computation of source follows: \\
$\psi_4 = h(\mu(\nu^\omega, \rho)) \ssteptrs{src_T} h(f(i \nu^\omega, \rho)) \ssteptrs{src_T} h(f(i \nu^\omega, a)) \infredtrs{src_T} h(f(i g^\omega, a))$. \\
Many $tgt_T$ \redseqs\ from $\psi_4$ are possible, \eg: \\
$\psi_4 = h(\mu(\nu^\omega, \rho)) 
		\ssteptrs{tgt_T} hh \rho \ssteptrs{tgt_T} hhb$ \\
$\psi_4 = h(\mu(\nu^\omega, \rho)) 
		\ssteptrs{tgt_T} h(\mu(\nu^\omega, b)) 
		\infredtrs{tgt_T} h(\mu(\nu^\omega, b))
		\ssteptrs{tgt_T} hhb$ where the $i$-th step for $1 \leq i < \omega$ is $\langle h(\mu(\nu^\omega, b)), 11 \cdot 1^i, \nu(x) \to x, \set{x \eqdef \nu^\omega} \rangle$  \\
$\psi_4 = h(\mu(\nu^\omega, \rho)) 
		\infredtrs{tgt_T} h(\mu(\nu^\omega, \rho))$
where all steps are $\langle \psi_4, 11, \nu(x) \to x, \set{x \eqdef \nu^\omega} \rangle$, a divergent $tgt_T$ \redseqs. \\
Then $\psi_4$ admit both convergent and divergent \redseqs\ in $tgt_T$. As $\psi_4$ is $tgt_T$-weakly normalising, we get $tgt(\psi_4) = hhb$. 
\end{itemize}

\endgroup

\subsection{Adding dots properly}
\usingContractionActivity{%
In this section we will give the definition of the set of all legal proof terms, by taking \imsteps\ as the foundation, and giving precise rules for the addition of occurrences of the dot.

As we have discussed in the informal introduction, for a term like $\psi \comp \phi$ to be a well-defined proof term, the concatenation of the contraction activities denoted by $\psi$ and $\phi$ must make sense.
}%
\usingRedseqOnly{%
For a term like $\psi \comp \phi$ to be a well-defined proof term, it must coherently denote any \redseq\ consisting of the concatenation of a \redseq\ denoted by $\psi$ with one denoted by $\phi$.
}%
Two conditions, related with this coherence requirement, are imposed. 
\usingContractionActivity{%
Firstly, the activity denoted by $\psi$ must be \emph{convergent}, \ie, it should exist at least one way to render such activity as a convergent \redseq; this condition implies particularly that the target term of $\psi$ can be uniquely determined.
Secondly, the target term of (the activity denoted by) $\psi$ must coincide with the source term of (that corresponding to) $\phi$.
}%
\usingRedseqOnly{%
Firstly, $\psi$ must denote at least one convergent \redseq. Secondly, the target of any \redseq\ denoted by $\psi$ must coincide with the source of any \redseq\ denoted by $\phi$. 
We remark that these conditions are similar to those required for a pre\redseq\ in order to be a \redseq, \confer\ \refdfn{sred}%
\footnote{Another remark: while the first condition is unique to infinitary rewriting, the second one must be imposed to proof terms denoting finite \redseqs\ as well; \confer\ \eg\ the transitivity rule in \cite{terese}, dfn. 8.2.18.}.
}%

The need of imposing such conditions on the occurrences of the dot implies that the set of proof terms must be defined along with the source, target and convergence condition for each proof term, in a joint definition.
Convergence depends in turn of the depth of the \redseqs\ being denoted; therefore, \emph{minimal activity depth} of proof terms must be merged within the same, huge definition.

\medskip
An additional goal is to define the set of proof terms by an \emph{inductive} construction, taking the set of \imsteps\ as the base case.
By doing so, we will be able to reason about proof terms in an inductive, opposed to coinductive, fashion, taking properties about \imsteps\ as the foundation for the inductive reasonings.

Since the occurrences of the dot are defined inductively, a special treatment is needed to allow a proof term to include an infinite number of them. Such a proof term should denote the concatenation of an infinite series of \redseqs%
\usingContractionActivity{ or, more generally, of contraction activities}%
.
Therefore, special care is taken to guarantee that no component is lost in the construction of the infinite concatenation; \ie, that any component is at a finite distance from the root in the corresponding proof term.

In turn, the separate treatment of binary and infinite concatenation gives rise to potential ambiguities in the construction of a proof term%
\footnote{it is a good idea to cite \cite{Gallier86}, and/or other work, here?}.
To avoid the possibility of such ambiguities, the definition of the set of proof terms is \emph{layered}, such that the proof terms included in a layer can be built taking as components proof terms in previous layers only.

Countable ordinals are used as layers for proof terms, and each proof term belongs to exactly one layer.
Therefore, layers give a transfinite induction principle to reason about the set of valid proof terms. An alternative, simpler induction principle for proof terms is given later in this section.

\medskip
The aforementioned restrictions and considerations try to justify the intricacies of the following definitions.

\begin{definition}[Signature for proof terms]
\label{dfn:sigma-pt}
Let $T = \pair{\Sigma}{R}$ be a (either finitary or infinitary) TRS. We define the signature for the proof terms over $T$ as follows:
$\Sigma^{PT} := \Sigma^R \cup \set{\comp / 2}$ 
. \Confer\ \refdfn{sigma-r} of $\,\Sigma^R$. 
\end{definition}

\begin{definition}[\layerpterm{\alpha}, layer $\alpha$ in the definition of proof terms]
\label{dfn:layer-pterm}
Let \trst\ be a \TRS, and $\alpha$ a countable ordinal. We define \layerpterm{\alpha}, the $\alpha$-th layer in the construction of the set of proof terms for \trst, along with the source, target, convergence condition, and minimal activity depth of any proof term in \layerpterm{\alpha}.
If $\psi \in \layerpterm{\alpha}$, we will write $src(\psi)$, $tgt(\psi)$ and $\mind{\psi}$ for the source, target and minimal activity depth of $\psi$ respectively.

If $\alpha = 0$, then $\layerpterm{\alpha} \eqdef \emptyset$.
Otherwise, we proceed inductively on $\alpha$, defining \layerpterm{\alpha} to be the smallest set in $Ter^\infty(\Sigma^{PT})$ verifying the following conditions.

\begin{enumerate}
\item
\label{rule:ptmstep}
If $\alpha = 1$ and $\psi$ is an \imstep\ for \trst, then $\psi \in \layerpterm{\alpha}$. The source, target, convergence condition and minimal activity depth of $\psi$ coincide with the definitions given for \imsteps\ in \refsec{mstep-defs}.

\item
\label{rule:ptinfC}
Assume that for any $i < \omega$, $\psi_i \in \layerpterm{\alpha_i}$, such that $\alpha = \underset{i < \omega}{\Sigma} \alpha_i$; \confer\ \refdfn{ordinal-infAdd}. Moreover, assume that for all $n$, $\psi_n$ is convergent, and $tgt(\psi_n) = src(\psi_{n+1})$. \\[1mm]
\begin{tabular}{@{}p{72mm}c}
\begin{minipage}{72mm}
Then $\psi \eqdef \pair{P}{F} \in \layerpterm{\alpha}$, where \\
$P \eqdef \set{2^n \setsthat n < w} \cup (\underset{n < \omega}{\bigcup} 2^n 1 \cdot \Pos{\psi_n})$, \\
$F(2^n) \eqdef \comp$, and $F(2^n 1 p) \eqdef \psi_n(p)$. \\
A concise term notation for $\psi$ is $\icomp \psi_i$; \\
being in fact an abbreviation for 
$\psi_1 \comp (\psi_2 \comp (\psi_3 \comp \ldots))$. \\
\end{minipage}
&
\begin{minipage}{55mm}
A graphical representation is \\
\hspace*{\stretch{1}}
$\xymatrix@C=3mm@R=3mm{
& \ar[dl] \ar[dr] \cdot \\
\psi_1 & & \ar[dl] \ar[dr] \cdot \\
& \psi_2 & & \ar[dl] \ar[dr] \cdot \\
& & \psi_3 & & \ddots \\
}$
\hspace*{\stretch{1}}
\end{minipage}
\end{tabular}
We define $src(\psi) \eqdef src(\psi_0)$, $tgt(\psi) \eqdef \lim_{i \to \omega} tgt(\psi_i)$ and $\mind{\psi} \eqdef min(\mind{\psi_i}_{i < \omega})$; notice that $tgt(\psi)$ can be undefined.
We define that $\psi$ is convergent iff for all $k < \omega$, there is some $n < \omega$ such that $\mind{\psi_j} > k$ if $j > n$.

\item
\label{rule:ptbinC}
Assume that $\psi_1 \in \layerpterm{\alpha_1}$, $\psi_2 \in \layerpterm{\alpha_2}$, $\alpha_2$ is a successor ordinal, $\psi_1$ is convergent, $tgt(\psi_1) = src(\psi_2)$, and $\alpha = \alpha_1 + \alpha_2 + 1$.
Then $\psi \eqdef \pair{P}{F} \in \layerpterm{\alpha}$, where 
$P \eqdef \set{\epsilon} \cup (1 \cdot \Pos{\psi_1}) \cup (2 \cdot \Pos{\psi_2})$, 
$F(\epsilon) \eqdef \comp$, and $F(ip) \eqdef \psi_i(p)$ for $i = 1,2$.

A concise term notation for $\psi$ is $\psi_1 \comp \psi_2$. A graphical notation is
$
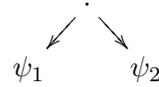
\xymatrix@C=3mm@R=4mm{
& \ar[dl] \ar[dr] \cdot \\
\psi_1 & & \psi_2
}$

We define $src(\psi) \eqdef src(\psi_1)$, $tgt(\psi) \eqdef tgt(\psi_2)$ and $\mind{\psi} = min(\mind{\psi_1}, \mind{\psi_2})$; $\psi$ is convergent iff $\psi_2$ is.

\item
\label{rule:ptsymbol}
Assume that $\psi_i \in \layerpterm{\alpha_i}$ for $i = 1, 2, \ldots, n$, that $\alpha_i > 1$ for at least one $i$, $f/n \in \Sigma$ (resp. $\mu/n$ is a rule symbol), and $\alpha = \alpha_1 + \ldots + \alpha_n + 1$.
Then $\psi \eqdef \pair{P}{F} \in \layerpterm{\alpha}$, where 
$P \eqdef \set{\epsilon} \cup (\underset{1 \leq i \leq n}{\bigcup}i \cdot \Pos{\psi_i})$, $F(\epsilon) \eqdef f$ (resp. $F(\epsilon) \eqdef \mu$), and $F(ip) \eqdef \psi_i(p)$ for $i = 1,2,\ldots,n$.

A concise term notation for $\psi$ is $f(\psi_1, \ldots, \psi_n)$ (resp. $\mu(\psi_1, \ldots, \psi_n)$).

For $f \in \Sigma$, we define $src(\psi) = f(src(\psi_1), \ldots, src(\psi_n))$, $tgt(\psi) = f(tgt(\psi_1), \ldots, tgt(\psi_n))$, and $\mind{\psi} \eqdef 1+min(\mind{\psi_1}, \ldots, \mind{\psi_n})$. $\psi$ is convergent iff all $\psi_i$ are. We observe that $tgt(\psi)$ is undefined if at least one $tgt(\psi_i)$ is.

For $\mu$ being a rule symbol such that $\mu : l \to r$, we define $src(\psi) = l[src(\psi_1), \ldots, src(\psi_n)]$, $tgt(\psi) = r[tgt(\psi_1), \ldots, tgt(\psi_n)]$, and $\mind{\psi} \eqdef 0$. $\psi$ is convergent iff all $\psi_i$ \textbf{corresponding to some $x_i$ occurring in $r$} are. We observe that $tgt(\psi)$ is undefined if at least one $tgt(\psi_i)$ is for the $\psi_i$ already mentioned.
\end{enumerate}
\end{definition}

\begin{definition}[\setpterm, the set of proof terms]
\label{dfn:pterm}
We define the set of proof terms as follows:
$\setpterm \eqdef \underset{\alpha < \omega_1}{\bigcup} \layerpterm{\alpha}$.
\end{definition}

We notice that all proof terms are \emph{closed} terms in $Ter^\infty(\Sigma^{PT})$. This fact is a consequence of the definition of the set of \imsteps, which are the base layer in the definition of \setpterm. \Confer\ the footnote on Dfn.~\ref{dfn:imstep}.


\medskip
We will say that a proof term $\psi$ is an \emph{infinite concatenation} iff $\psi(2^n) = \comp$ for all $n < \omega$.
Observe that all infinite concatenations admit the concise term notation $\psi = \icomp \psi_i$, where $\psi_n = \subtat{\psi}{2^n 1}$.
Furthermore, $\psi$ not being an infinite concatenation implies the existence of some $n < \omega$ such that $2^n \in \Pos{\psi}$ and $\psi(2^n) \neq \comp$.

\subsection{Soundness of the definitions}
In this section we will study the definition of the set of valid proof terms in some detail, stating and proving properties related to its soundness.

\begin{lemma}
\label{rsl:ptinfC-iff-limit}
Let $\psi$, $\alpha$ such that $\psi \in \layerpterm{\alpha}$.
Then $\psi$ is an infinite concatenation iff $\alpha$ is a limit ordinal iff $\psi$ is generated by case \ref{rule:ptinfC} in Dfn.~\ref{dfn:layer-pterm}.
\end{lemma}

\begin{proof}
We proceed by induction on $\alpha$, analysing the rules in Dfn.~\ref{dfn:layer-pterm}.

\noindent
Case \ref{rule:ptmstep}: in this case $\psi$ is an \imstep, so that $\psi(2^0) = \psi(\epsilon) \neq \comp$.

\noindent
Case \ref{rule:ptinfC}: 
in this case $\psi = \icomp \psi_i$, that is, an infinite concatenation.
It is enough to observe that $\layerpterm{0} = \emptyset$, and that $\alpha_i > 0$ for all $i$ implies that $\sum_{i < \omega} \alpha_i$ is a limit ordinal.

\noindent
Case \ref{rule:ptbinC}: 
in this case $\psi = \psi_1 \comp \psi_2$ where $\psi_i \in \layerpterm{\alpha_i}$,  $\alpha_2$ is a successor ordinal, and $\alpha = \alpha_1 + \alpha_2 + 1$, \ie\ a successor ordinal.
\ih\ on $\psi_2$ implies that $\psi_2(2^n) \neq \comp$ for some $n < \omega$.
We conclude by observing that $\psi(2^{n+1}) = \psi_2(2^n)$.

\noindent
Case \ref{rule:ptsymbol}: 
in this case it is immediate that $\psi(2^0) = \psi(\epsilon) \neq \comp$, and that $\alpha$ is a successor ordinal.
\end{proof}

\begin{lemma}
\label{rsl:ptmstep-iff-one}
Let $\psi$, $\alpha$ such that $\psi \in \layerpterm{\alpha}$.
Then $\psi$ is an \imstep\ iff $\alpha = 1$ iff $\psi$ is generated by case 1 in Dfn.~\ref{dfn:layer-pterm}.
\end{lemma}

\begin{proof}
We proceed by induction on $\alpha$, analysing the rules in Dfn.~\ref{dfn:layer-pterm}.

\noindent
Case \ref{rule:ptmstep}: we conclude immediately.

\noindent
Case \ref{rule:ptinfC}: 
in this case $\psi$ is not an \imstep, observe \eg\ that $\psi(\epsilon) = \comp$, and $\alpha$ is a limit ordinal, \confer\ Lem.~\ref{rsl:ptinfC-iff-limit}. Thus we conclude.

\noindent
Case \ref{rule:ptbinC}: 
in this case $\psi$ is not an \imstep, observe \eg\ that $\psi(\epsilon) = \comp$, and $\alpha > \alpha_1 + 1 > 1$, recall $\layerpterm{0} = \emptyset$.
Thus we conclude.

\noindent
Case \ref{rule:ptsymbol}: 
in this case $\psi = f(\psi_1, \ldots, \psi_n)$ where $\psi_i \in \layerpterm{\alpha_i}$ for all $i$, and exists some $k$ such that $\alpha_k > 1$.
Observe that $\alpha > \alpha_k > 1$, then we can apply \ih\ to obtain that $\psi_k$ is not an \imstep, hence $\psi$ is neither. Thus we conclude.
\end{proof}

\medskip
The set \setpterm\ is closed by operations, formally:

\begin{proposition}[Completeness of \setpterm] 
\label{prop:compl-pt}
\hspace*{1mm} \\ \vspace{-5mm}
\begin{enumerate}
	\item \label{it:compl-pt:mul}
	If $\psi$ is an infinite multistep, then $\psi \in \setpterm$.
	
	\item \label{it:compl-pt:binC}
	If $\psi_1, \psi_2 \in \setpterm$, $\psi_1$ is convergent, and $src(\psi_2) = tgt(\psi_1)$, then $\psi_1 \comp \psi_2 \in \setpterm$.

	\item \label{it:compl-pt:infC}
	Given a sequence $\iomegaseq{\psi_i}$ such that for all $i$, $\psi_i \in \setpterm$, $\psi_i$ are convergent, and $tgt(\psi_i) = src(\psi_{i+1})$, 
	then $\icomp \psi_i \in \setpterm$.
	
	\item \label{it:compl-pt:symbol}
	If $\psi_1, \ldots, \psi_n \in \setpterm$ and $f \in \Sigma$, then $f(\psi_1, \ldots, \psi_n) \in \setpterm$.

	\item \label{it:compl-pt:rule}
	If $\psi_1, \ldots, \psi_n \in \setpterm$ and $\mu$ is a rule symbol, then $\mu(\psi_1, \ldots, \psi_n) \in \setpterm$.
\end{enumerate}
\end{proposition}

\begin{proof}
We prove each item separately, referring to cases in Dfn.~\ref{dfn:layer-pterm}. 

\noindent
Item \ref{it:compl-pt:mul}: 
in this case $\psi \in \layerpterm{1}$, this is immediate from case \ref{rule:ptmstep}.

\noindent
Item \ref{it:compl-pt:binC}: 
Let $\alpha_1$, $\alpha_2$ such that $\psi_i \in \layerpterm{\alpha_i}$ for $i = 1,2$.
If $\alpha_2$ is a successor ordinal, then $\psi_1 \comp \psi_2 \in \layerpterm{\alpha_1 + \alpha_2 + 1} \subseteq \setpterm$.
If $\alpha_2$ is a limit ordinal, then Lem.~\ref{rsl:ptinfC-iff-limit} implies that $\psi_2 = \icomp \phi_i$, where for all $i$, $\phi_i$ is convergent and $tgt(\phi_i) = src(\phi_{i+1})$; \confer\ case~\ref{rule:ptinfC}.
On the other hand, hypotheses imply that $\psi_1$ is convergent and $tgt(\psi_1) = src(\psi_2) = src(\phi_0)$.
Then $\psi_1 \comp \psi_2 \in \layerpterm{\alpha_1 + \alpha_2}$, again by case~\ref{rule:ptinfC}. Observe that $\psi_1 \comp \psi_2 = \psi_1 \comp (\icomp \phi_i) = \icomp \phi'_i$ where $\phi'_0 \eqdef \psi_1$ and $\phi'_{i+1} \eqdef \phi_i$ for all $i < \omega$.

\noindent
Item \ref{it:compl-pt:infC}:
we conclude just by observing that case~\ref{rule:ptinfC} implies that $\icomp \psi_i \in \layerpterm{\beta}$, where $\psi_i \in \layerpterm{\alpha_i}$ for all $i < \omega$ and $\beta \eqdef \sum_{i < \omega} \alpha_i$.

\noindent
Item \ref{it:compl-pt:symbol} and Item \ref{it:compl-pt:rule}:
it is enough to observe that case~\ref{rule:ptsymbol} applies.
\end{proof}

\medskip
Now we prove uniqueness of formation, \wrt\ the layered definition, for any valid proof term.

\begin{lemma}
\label{rsl:pterm-layer-uniqueness}
Let $\psi \in \setpterm$. Then there exists a unique $\alpha$ such that $\psi \in \layerpterm{\alpha}$, and moreover there is exactly one case in Dfn.~\ref{dfn:layer-pterm} justifying $\psi \in \layerpterm{\alpha}$.
\end{lemma}

\begin{proof}
We will prove the following statement, which is equivalent to the desired result.
\begin{quote}
Let $\psi \in \setpterm$, $\alpha$ minimal for $\psi \in \layerpterm{\alpha}$, and $\beta$ such that $\psi \in \layerpterm{\beta}$.
Then $\beta = \alpha$, and there is exactly one case in Dfn.~\ref{dfn:layer-pterm} justifying $\psi \in \layerpterm{\alpha}$.
\end{quote}
We proceed by induction on $\alpha$, analysing which case in Dfn.~\ref{dfn:layer-pterm} could justify $\psi \in \layerpterm{\alpha}$.

\noindent
Case \ref{rule:ptmstep}.
In this case $\alpha = 1$ and $\psi$ is an \imstep. We conclude by Lem.~\ref{rsl:ptmstep-iff-one}.

\noindent
Case \ref{rule:ptinfC}.
In this case $\psi = \icomp \psi_i$ such that $\psi_i \in \layerpterm{\alpha_i}$ and $\alpha = \sum_{i < \omega} \alpha_i$. Observe that $\alpha > \alpha_i$ for all $i$, recall $\layerpterm{0} = \emptyset$.
Assume $\psi \in \layerpterm{\beta}$.
Lem.~\ref{rsl:ptinfC-iff-limit} implies that this assertion is generated by case \ref{rule:ptinfC}, implying that $\beta = \sum_{i < \omega} \beta_i$ and $\psi_i \in \layerpterm{\beta_i}$.
Let $i < \omega$ and $\gamma_i$ minimal for $\psi_i \in \layerpterm{\gamma_i}$.
Then $\gamma_i \leq \alpha_i < \alpha$, and therefore \ih\ can be applied twice on each $\psi_i$ obtaining $\beta_i = \alpha_i = \gamma_i$.
Thus we conclude.

\noindent
Case \ref{rule:ptbinC}.
In this case $\psi = \psi_1 \comp \psi_2$, $\alpha = \alpha_1 + \alpha_2 + 1$, $ \alpha_2$ is a successor ordinal, and $\psi_i \in \layerpterm{\alpha_i}$ for $i = 1,2$.
Then Lem.~\ref{rsl:ptinfC-iff-limit} applied to $\psi_2$ implies that it is not an infinite concatenation, thus neither is $\psi$.
On the other hand, observe that $\alpha$ is a successor ordinal verifying $\alpha > \alpha_i$ for $i = 1,2$.
Assume $\psi \in \layerpterm{\beta}$.
Then applying again Lem.~\ref{rsl:ptinfC-iff-limit} yields that this assertion is not justified by case 2 (since $\psi$ is not an infinite concatenation); therefore, the shape of $\psi$ (recall $\psi(\epsilon) = \comp$) leaves case 3 as the only valid option.
Hence $\beta = \beta_1 + \beta_2 + 1$ where $\psi_i \in \layerpterm{\beta_i}$ for $i = 1,2$.
An argument analogous to that used in the previous case, \ie\ resorting to the \ih\ on each $\psi_i$, yields $\beta_i = \alpha_i$. Thus we conclude.

\noindent
Case \ref{rule:ptsymbol}.
In this case $\psi = f(\psi_1, \ldots, \psi_m)$ and $\alpha = \alpha_1 + \ldots + \alpha_m + 1$, where $\psi_i \in \layerpterm{\alpha_i}$ for all $i$, and exists some $k$ veriyfing $\alpha_k > 1$.
Then Lem.~\ref{rsl:ptmstep-iff-one} implies that $\psi_k$ is not an \imstep, so that neither is $\psi$.
Therefore, the shape of $\psi$ (recall $\psi(\epsilon) \neq \comp$) leaves case 4 as the only valid option, implying that $\beta = \beta_1 + \ldots + \beta_m + 1$ where $\psi_i \in \layerpterm{\beta_i}$ for all $i$.
We conclude by obtaining $\beta_i = \alpha_i$ through an argument resorting to the \ih, like in the previous cases.
\end{proof}

\subsection{A simplified induction principle}
\label{sec:pterm-induction-principle}

The layered definition of \setpterm\ allows to perform inductive reasonings over proof terms, based in their concise notation. This makes an induction principle easy to work with. Formally:

\begin{proposition}[Simple induction principle for \setpterm]
\label{rsl:pterm-induction-principle}
Let $P$ an unary predicate satisfying all the following conditions: \\
\begin{tabular}{lp{125mm}}
	1. & If $\psi$ is an infinitary multistep, then $P(\psi)$ holds. \\
	2. & For all $\psi_1, \psi_2$ such that $\psi_1 \comp \psi_2 \in \setpterm$, $P(\psi_1)$ and  $P(\psi_2)$ imply $P(\psi_1 \comp \psi_2)$. \\
	3. & Given $\iomegaseq{\psi_i}$ such that $\icomp \psi_i \in \setpterm$, $P(\psi_i)$ for all $i$ imply $P(\icomp \psi_i)$. \\
	4. & For all $\psi_1, \ldots, \psi_n \in \setpterm$ and for all $f \in \Sigma$, $P(\psi_1), \ldots, P(\psi_n)$ imply $P(f(\psi_1, \ldots, \psi_n))$. \\ 
	5. & For all $\psi_1, \ldots, \psi_n \in \setpterm$ and for any rule symbol $\mu$, $P(\psi_1), \ldots, P(\psi_n)$ imply $P(\mu(\psi_1, \ldots, \psi_n))$. \\ 
\end{tabular} 

\noindent
Then $P(\psi)$ holds for all $\psi \in \setpterm$.
\end{proposition}

\begin{proof}
We proceed by induction on $\alpha$ where $\psi \in \layerpterm{\alpha}$, referring to the conditions in the lemma statement.

\noindent
If $\alpha = 1$, then Lem.~\ref{rsl:ptmstep-iff-one} implies $\psi$ to be an \imstep, so that we conclude by condition 1.

\noindent
Assume that $\alpha$ is a successor ordinal.
If $\psi(\epsilon) = \comp$, then Lem~\ref{rsl:ptinfC-iff-limit} implies that $\psi = \psi_1 \comp \psi_2$, such that for $i = 1,2$, $\psi_i \in \layerpterm{\alpha_i}$ for some $\alpha_i$ satisfying $\alpha > \alpha_i$. 
Then \ih\ can be applied on each $\psi_i$ yielding $P(\psi_1)$ and $P(\psi_2)$ to hold. We conclude by condition 2.
Otherwise, \ie\ if $\psi = f(\psi_1, \ldots, \psi_m)$ or $\psi = \mu(\psi_1, \ldots, \psi_m)$, then Lem.~\ref{rsl:ptmstep-iff-one} implies that $\psi$ is not an \imstep, therefore for all $i$, $\psi_i \in \layerpterm{\alpha_i}$ where $\alpha > \alpha_i$. Then \ih\ on each $i$ yield $P(\psi_i)$ to hold for all $i$. We conclude by condition 4.

\noindent
Assume that $\alpha$ is a limit ordinal.
In this case, Lem~\ref{rsl:ptinfC-iff-limit} implies that $\psi = \icomp \psi_i$, such that for all $i < \omega$, $\psi_i \in \layerpterm{\alpha_i}$ where $\alpha_i < \alpha$. 
Then we can apply \ih\ on each $\psi_i$ obtaining that $P(\psi_i)$ holds for all $i < \omega$. We conclude by condition 3.
\end{proof}

\medskip
We will resort to the induction principle given by Prop.~\ref{rsl:pterm-induction-principle} in forthcoming proofs, where we will indicate as \emph{induction hypotheses} the hypotheses of each case in the Proposition.
\Eg\ when proving a property for proof terms having the form $\psi_1 \comp \psi_2$, we will refer to the hypohteses of case 2 in Prop.~\ref{rsl:pterm-induction-principle}, namely that the property holds for $\psi_1$ and $\psi_2$, as induction hypothesis in the proof.
The intent is to produce intuitively simple yet rigorously valid proofs of properties on the set of proof terms.

\subsection{Basic properties of proof terms}
The following lemma shows that the target of a convergent proof term is always defined, and also a correspondence between $mind(\psi)$ and the existence of a fixed prefix for the activity denoted by $\psi$.
These two results are merged in the same lemma because they need to be proved simultaneously.

\begin{lemma}
\label{rsl:mind-big-then-tdist-little}
Let $\psi$ be a convergent proof.
Then 
\begin{enumerate}[(a)]
\item \label{it:convergent-then-has-tgt}
\vspace{-1mm}
$tgt(\psi)$ is defined.
\item \label{it:mind-big-then-tdist-little}
\vspace{-1mm}
For all $n < \omega$, $mind(\psi) > n$ implies $\tdist{src(\psi)}{tgt(\psi)} < 2^{-n}$.
\end{enumerate}
\end{lemma}

\begin{proof}
We proceed by induction on $\alpha$ where $\psi \in \layerpterm{\alpha}$, analysing the case in Dfn.~\ref{dfn:layer-pterm} corresponding to $\psi$.
If $\psi$ is an \imstep, then item~(\ref{it:convergent-then-has-tgt}) is immediate from Dfn.~\ref{dfn:imstep-convergence}, and for item~(\ref{it:mind-big-then-tdist-little}) an easy induction on $n$ suffices.

Assume $\psi = \psi_1 \comp \psi_2$.
Item~(\ref{it:convergent-then-has-tgt}) can be proved by just applying \ih\ on $\psi_2$.
To obtain item~(\ref{it:mind-big-then-tdist-little}), observe that \ih\ applies to $\psi_i$ for $i = 1,2$, since $mind(\psi_i) \geq mind(\psi) > n$, yielding $\tdist{src(\psi_i)}{tgt(\psi_i)} < 2^{-n}$.
Moreover Lemma~\ref{rsl:tdist-is-ultrametric} implies $\tdist{src(\psi)}{tgt(\psi)} \leq 
 max(\tdist{src(\psi)}{src(\psi_2)}, \tdist{src(\psi_2)}{tgt(\psi)}$. Thus we conclude by observing $src(\psi) = src(\psi_1)$, $src(\psi_2) = tgt(\psi_1)$, and $tgt(\psi) = tgt(\psi_2)$.
 
Assume $\psi = \icomp \psi_i$.

We prove item~(\ref{it:convergent-then-has-tgt}).
For any $i < \omega$, $\psi_i$ being convergent implies that \ih\ applies to obtain that $tgt(\psi_i)$ is defined.
Let $n < \omega$, and $k_n$ such that $\mind{\psi_i} > n$ if $k_n < i < \omega$.
Let $j$ such that $k_n < j$.
Then \ih:(\ref{it:mind-big-then-tdist-little}) applies on $\psi_{k_n+1} \comp \ldots \comp \psi_j$, implying $\tdist{tgt(\psi_{k_n+1})}{tgt(\psi_j)} < 2^{-n}$%
\footnote{A possible shortcut from here is observing that the sequence $\langle tgt(\psi_i) \rangle_{i < \omega}$ is Cauchy-convergent, and therefore has a limit. We can refer to Thm.~12.2.1 in \cite{terese}, or its proof.}%
.
Therefore, for any position $p$ and $j \geq k_{\posln{p}}+1$, $p \in \Pos{tgt(\psi_j)}$ iff $p \in \Pos{tgt(\psi_{k_{\posln{p}}+1})}$, and in such case, $tgt(\psi_j)(p) = tgt(\psi_{k_{\posln{p}}+1})(p)$.
We define $t = \langle P, F \rangle$ as follows:
$p \in P$ iff $p \in \Pos{tgt(\psi_{k_{\posln{p}}+1})}$, and
$F(p) \eqdef tgt(\psi_{k_{\posln{p}}+1})(p)$ for all $p \in P$.
To conclude this part of the proof, it is enough to verify that $tgt(\psi) = \lim_{i \to \omega} tgt(\psi_i) = t$.
\begin{itemize}
\item 
We verify that $P$ is a tree domain, \confer\ Dfn.~\ref{dfn:tree-domain}.
Let $pq \in P$, then \linebreak $pq \in \Pos{tgt(\psi_{k_{\posln{pq}}+1})}$, implying that $p \in \Pos{tgt(\psi_{k_{\posln{pq}}+1 })}$. Then $p \in \Pos{tgt(\psi_{k_{\posln{p}+1}})}$, hence $p \in P$.
Let $pj \in P$ and $i$ such that $1 \leq i \leq j$. Observing $\posln{pj} = \posln{pi}$, a straightforward argument based on $\psi_{k_{\posln{pj}}+1}$ yields $pi \in P$.
\item
We verify that $t$ is a well-defined term, \confer\ Dfn.~\ref{dfn:term}. 
Let $p \in P$, $f/m \eqdef F(p)$, and $i < \omega$. 
Observe $f = \psi_{k_{\posln{p}}+1}(p) = \psi_{k_{\posln{p}+1}+1}(p)$.
Then $pi \in P$ iff $pi \in \Pos{\psi_{k_{\posln{pi}}+1}}$ iff $i \leq m$.
\item
We verify that $t = \lim_{i \to \omega} tgt(\psi_i)$. 
Let $n < \omega$, $j > k_n$, and $p$ a position verifying $\posln{p} \leq n$, so that $k_{\posln{p}} \leq k_n$, implying in turn $k_{\posln{p}} + 1 \leq j$.
Then $p \in \Pos{t}$ iff $p \in \Pos{tgt(\psi_{k_{\posln{p}}+1})}$ iff $p \in \Pos{tgt(\psi_j)}$, and in such case, $t(p) = tgt(\psi_{k_{\posln{p}}+1})(p) = tgt(\psi_j)(p)$.
Hence $\tdist{tgt(\psi_j)}{t} < 2^{-n}$.
Consequently, $t = \lim_{i \to \omega} tgt(\psi_i)$.
\end{itemize}

We prove item~(\ref{it:mind-big-then-tdist-little}).
For all $i < \omega$, $mind(\psi_i) \geq mind(\psi) > n$, and then an easy induction on $i$ using an argument similar to the one just described for binary composition yields $\tdist{src(\psi)}{tgt(\psi_i)} < 2^{-n}$.
Recall that $tgt(\psi) = \lim_{i \to \omega} tgt(\psi_i)$, then there exists some $k$ such that $\tdist{tgt(\psi_j)}{tgt(\psi)} < 2^{-n}$ if $j > k$.
Then $\tdist{src(\psi)}{tgt(\psi_{k+1})} < 2^{-n}$ and $\tdist{tgt(\psi_{k+1})}{tgt(\psi)} < 2^{-n}$.
We conclude by Lemma~\ref{rsl:tdist-is-ultrametric}.

Assume $\psi = f(\psi_1, \ldots, \psi_m)$ and that it is not an \imstep.
Then $\psi$ being convergent implies that all $\psi_i$ are. 
Therefore a straightforward argument based on \ih\ implies item~(\ref{it:convergent-then-has-tgt}) to hold.
Moreover, the way in which $src$, $tgt$ and $\mindfn$ for this case, implies that a natural inductive argument yields also item~(\ref{it:mind-big-then-tdist-little}).

Assume $\psi = \mu(\psi_1, \ldots, \psi_m)$, and that it is not an \imstep.
Then $\psi$ being convergent implies that $\psi_i$ is if $x_i$ occurs in the right-hand side of $\mu$, thus \ih:(\ref{it:convergent-then-has-tgt}) implies that $tgt(\psi_i)$ is defined for those $\psi_i$. Hence, definition of $tgt$ for this case yields item~(\ref{it:convergent-then-has-tgt}).
On the other hand, $\mind{\psi} = 0$ contradicting the hypotheses of item~(\ref{it:mind-big-then-tdist-little}).
Thus we conclude.
\end{proof}

\begin{lemma}
\label{rsl:mind-ctx}
Let $C$ be a context in $\SigmaTerms$ having $k$ holes, and $\psi_1, \ldots, \psi_k$ proof terms.
Then $\mind{C[\psi_1, \ldots, \psi_k]} = min \set{\mind{\psi_i} + \posln{\BPos{C}{i}} \setsthat 1 \leq i \leq k}$.
\end{lemma}

\begin{proof}
An easy, although somewhat cumbersome, induction on $max \set{\posln{\BPos{C}{i}}}$ suffices.
If $C = \Box$, then both sides of the equation in the lemma conclusion equates to $\psi$, thus we conclude.

Assume $C = f(C_1, \ldots, C_m)$. \\
Observe that $C[\psi_1, \ldots, \psi_k] = f(C_1[\psi_{1_1}, \ldots, \psi_{1_{q1}}], \ldots, C_m[\psi_{m_1}, \ldots, \psi_{m_{qm}}])$, where $\set{\psi_{j_i}} = \set{\psi_1, \ldots, \psi_k}$.
Consequently, for any $i$ such that $1 \leq i \leq k$, $\BPos{C}{i} = e \, p$ for some $e$ verifying $1 \leq e \leq m$, and therefore $p = \BPos{C_e}{l}$ for some $l$. In turn, this implies $\posln{\BPos{C}{i}} = 1 + \posln{\BPos{C_e}{l}}$.
Conversely, for any $e$ such that $1 \leq e \leq m$, and for any $\BPos{C_e}{i}$, there is an index $j$ such that $\BPos{C}{j} = e \cdot \BPos{C_e}{i}$. 
Furthermore, $\mind{C[\psi_1, \ldots, \psi_k]} = 1 + min \{\mind{C_j[\psi_{j_1}, \ldots, \psi_{j_{qj}}]} $ $\setsthat 1 \leq j \leq m\}$.

Let $j$ minimal for $\mind{\psi_j} + \posln{\BPos{C}{j}}$, so that showing $\mind{C[\psi_1, \ldots, \psi_k]} = \mind{\psi_j} + \posln{\BPos{C}{j}}$ is enough to conclude.
Let $e, i$ such that $\BPos{C}{j} = e \cdot \BPos{C_e}{i}$.
The existence of some $j', i'$ such $\BPos{C}{j'} = e \cdot \BPos{C_e}{i'}$ and $\mind{\psi_j'} + \posln{\BPos{C_e}{i'}} < \mind{\psi_j} + \posln{\BPos{C_e}{i}}$ would contradict minimality of $\psi_j$ \wrt\ $C$, so that $j, i$ are minimal for $\mind{\psi_j} + \posln{\BPos{C_e}{i}}$.
Therefore, applying \ih\ on $C_j$, yields that $\mind{C_e[\psi_{e_1}, \ldots, \psi_{e_{qe}}]} = \mind{\psi_j} + \posln{\BPos{C_e}{i}}$. 

Assume for contractiction the existence of some $m,h$ such that \\
$\mind{C_h[\psi_{h_1}, \ldots, \psi_{h_{qh}}]} < \mind{C_e[\psi_{e_1}, \ldots, \psi_{h_{eh}}]}$.
Applying \ih\ on $C_h$ we obtain $\mind{C_h[\psi_{h_1}, \ldots, \psi_{h_{qh}}]} = \mind{\psi_g} + \posln{\BPos{C_h}{f}}$ for some $f$ and $g$ such that $\BPos{C}{g} =h \comp \BPos{C_h}{f}$.
But then our assumption would imply $\mind{\psi_g} + \posln{\BPos{C}{g}} = \mind{\psi_g} + \posln{\BPos{C_h}{f}} + 1 < \mind{\psi_j} + \posln{\BPos{C_e}{i}} + 1 = \mind{\psi_j} + \posln{\BPos{C}{j}}$, contradicting minimality of $j$ \wrt\ $C$.

Hence, $\mind{C[\psi_1, \ldots, \psi_k]} = 1 + \mind{C_e[\psi_{e_1}, \ldots, \psi_{e_{qe}}]} = \mind{\psi_j} + \posln{\BPos{C}{j}}$.
Thus we conclude.
\end{proof}

\medskip
Some properties related with convergence follow.

\begin{lemma}
\label{rsl:imstep-convergent-args}
Let $\psi = f(\psi_1, \ldots, \psi_m)$ be a convergent \imstep, and $i$ such that $1 \leq i \leq m$. 
Then $\psi_i$ is a convergent \imstep.
\end{lemma}

\begin{proof}
Dfn.~\ref{dfn:imstep} yields immediately that $\psi_i$ is an \imstep.
Moreover, $f(\psi_1, \ldots, \psi_m)$ being convergent means the existence of a convergent $\tgtt$-\redseq\ $\reda$ such that $f(\psi_1, \ldots, \psi_m) \infredxtrs{\reda}{\tgtt} t$ and $t$ is a $\tgtt$-normal form, \ie\ $t \in \iSigmaTerms$. Observe that $\mind{\reda} > 0$, since $f$ does not occur in any left-hand side of a rule in $\tgtt$.
Then Lem.~\ref{rsl:redseq-mind-big-src-tgt} implies $t = f(t_1, \ldots, t_m)$. In turn, Lem.~\ref{rsl:proj-redseq-well-defined} implies $\psi_i \infredxtrs{\proj{\reda}{i}}{\tgtt} t_i$. Thus we conclude.
\end{proof}

\begin{lemma}
\label{rsl:fnsymbol-convergence}
Let $\psi = f(\psi_1, \ldots, \psi_m)$ be a convergent proof term. Then $\psi$ is convergent iff $\psi_i$ is convergent for all suitable $i$.
\end{lemma}

\begin{proof}
If $\psi$ is an \imstep, then the $\Rightarrow )$ direction is an immediate corollary of Lem.~\ref{rsl:imstep-convergent-args}.
For the $\Leftarrow )$ direction, recall that for any $i$, $\psi_i$ being convergent means the existence of a $\tgtt$-\redseq\ $\reda_i$ verifying $\psi_i \infredxtrs{\reda_i}{\tgtt} t_i$ where $t_i \in \iSigmaTerms$.
Then $f(\psi_1, \ldots, \psi_m) \infredxtrs{\reda}{\tgtt} f(t_1, \ldots, t_m)$, where $\reda \eqdef (1 \cdot \reda_1 ) ; \ldots ; (m \cdot \reda_m)$, and $i \cdot \reda_i$ is defined as follows: $\redln{i \cdot \reda_i} \eqdef \redln{\reda_i}$ and $\redel{i \cdot \reda_i}{\alpha} \eqdef \langle f(t_1, \ldots, \phi, \ldots \psi_m), ip, \mu \rangle$ where $\redel{\reda_i}{\alpha} = \langle \phi, p, \mu \rangle$. A simple transfinite induction yields $f(t_1, \ldots, t_{i-1}, \psi_i, \psi_{i+1}, \ldots, \psi_m) \infredxtrs{i \cdot \reda_i}{\tgtt} f(t_1, \ldots, t_{i-1}, t_i, \psi_{i+1}, \ldots, \psi_m)$.

If $\psi$ is not an \imstep, then the result is an immediate consequence of Dfn.~\ref{dfn:layer-pterm}, case~(\ref{rule:ptsymbol}).
Thus we conclude.
\end{proof}

\begin{lemma}
\label{rsl:ctx-convergence}
Let $C$ be a context in $\iSigmaTerms$ having exactly $m$ holes, and $\psi_1, \ldots, \psi_m$ proof terms.
Then $C[\psi_1, \ldots, \psi_m]$ is convergent iff $\psi_i$ is convergent for all suitable $i$.
\end{lemma}

\begin{proof}
A straightforward induction on $max \set{\posln{\BPos{C}{i}} \setsthat 1 \leq i \leq m}$, resorting to Lem.~\ref{rsl:fnsymbol-convergence} in the inductive case, suffices to conclude.
\end{proof}

\begin{lemma}
\label{rsl:rulesymbol-convergence}
Let $\mu : l[x_1, \ldots, x_m] \to h[x_1, \ldots, x_m]$ be a rule included in a certain \TRS; and $\psi_1, \ldots, \psi_m$ proof terms. 
Then $\psi = \mu(\psi_1, \ldots, \psi_m)$ is convergent iff $\psi_i$ is convergent for all $i$ such that $x_i$ occurs in $h[x_1, \ldots, x_m]$.
\end{lemma}

\begin{proof}
Assume that $\psi$ is an \imstep.
We verify $\Rightarrow )$.
Convergence of $\psi$ implies $\psi \infredxtrs{\reda}{\tgtt} t$ for some \redseq\ $\reda$, where $t \in \iSigmaTerms$.
Notice that $\mind{\reda} > 0$ would imply $t(\epsilon) = \mu$ (\confer\ Lem.~\ref{rsl:redseq-mind-big-src-tgt}), contradicting $t \in \iSigmaTerms$. 
Therefore $\mind{\reda} = 0$, implying $\reda = \reda_1; \langle \chi, \epsilon, \uln{\nu} \rangle, \reda_2$ where $\mind{\reda_1} > 0$.
In turn, $\mind{\reda_1} > 0$ implies that $tgt(\reda_1) = \chi = \mu(\chi_1, \ldots, \chi_m)$ where $\psi_i \infredxtrs{\proj{\reda}{i}}{\tgtt} \chi_i$, \confer\ Lem.~\ref{rsl:redseq-mind-big-src-tgt} and Lem~\ref{rsl:proj-redseq-well-defined}.
Hence $\uln{\nu} = \uln{\mu} : \mu(x_1, \ldots, x_m) \to h[x_1, \ldots, x_m]$, implying $src(\reda_2) = h[\chi_1, \ldots, \chi_m]$.
Observe that $\chi_i$ occurs in $src(\reda_2)$ iff $x_i$ occurs in $h$.
We analyse two cases: \\[5pt]
\begin{tabular}{@{$\ \ \bullet\ \ $}p{.9\textwidth}}
$h[x_1, \ldots, x_m] = x_j$, so that $src(\reda_2) = \chi_j$.
In this case $\psi_j \infredx{\proj{\reda}{j}} \chi_j \infredx{\reda_2} t$. We conclude by observing that only convergence of $\psi_j$ is required in this case. \\
$h \notin \thevar$.
In this case $h[\chi_1, \ldots, \chi_m] \infredx{\reda_2} t$. 
Observe that all the steps in $\reda_2$ lies ``below'' (an argument of) $h$. Then Lem.~\ref{rsl:redseq-respects-src-tgt} implies $t = h[t_1, \ldots, t_m]$ and, moreover, that a \redseq\ $\reda'_i$ exists which verifies $\chi_i \infredx{\reda'_i} t_i$ for all $i$ such that $x_i$ occurrs in $h[x_1, \ldots x_m]$.
Therefore, for any of those indices, say $i$, $\psi_i \infredx{\proj{\reda_1}{i}} \chi_i \infredx{\reda'_i} t_i$.
Thus we conclude.
\end{tabular} 

To verify the $\Leftarrow )$ direction, observe that all the $\psi_i$ corresponding to variables occurring in $h$ being convergent implies 
$\psi \to h[\psi_1, \ldots \psi_m] \infredx{\reda_1} h[t_1, \ldots, \psi_m] \ldots \infredx{\reda_m} h[t_1, \ldots, t_m]$, where eventually some $\reda_i$ are performed more than once, if the corresponding $x_i$ occurs more than once in $h[x_1, \ldots, x_m]$.
Hence $\psi$ is $\tgtt$-$WN^\infty$, \ie\ it is a convergent \imstep.

Finally, if $\psi$ is not an \imstep, then Dfn.~\ref{dfn:layer-pterm}, case~(\ref{rule:ptsymbol}), allows to conclude immediately.
\end{proof}

\subsection{Trivial proof terms}
This section deals with the proof terms denoting no activity, which will be termed \emph{trivial proof terms}.
The structure of trivial proof terms can be arbitrarily complex, \ie\ $\compomega{j}(\icomp a)$ is a trivial proof term.
We prove that some expected properties hold for these proof terms. These properties will be used later in this work.

\begin{definition}
\label{dfn:trivial-pterms}
Let $\psi$ be a proof term. We will say that $\psi$ is a \emph{trivial proof term} iff it does not include any rule symbol occurrences.
\end{definition}

\begin{lemma}
\label{rsl:trivial-pterm-mind-omega}
Let $\psi$ be a proof term. Then $\psi$ is trivial iff $\mind{\psi} = \omega$.
\end{lemma}

\begin{proof}
For the $\Rightarrow )$ direction, a straightforward induction on $\psi$ (\ie\ on $\alpha$ such that $\psi \in \layerpterm{\alpha}$) suffices.
For the base case, \ie\ when $\psi$ is an \imstep, we just refer to Dfn.~\ref{dfn:dmin-imstep}.

For the $\Leftarrow )$ direction, a similar induction on $\psi$ yields the counterpositive, \ie\ that if $\psi$ includes at least one rule symbol occurrence, then $\mind{\psi} < \omega$.
If $\psi$ is an \imstep, then we define $n$ to be the least depth of a rule symbol occurrence in $\psi$. An easy induction on $n$ yields $\mind{\psi} = n$. If $\psi = \mu(\psi_1, \ldots, \psi_m)$, then $\mind{\psi} = 0$. For the other cases, \ih\ suffices to conclude.
\end{proof}

\newpage
\section{\Peqence}
\label{sec:peqence}
\usingRedseqOnly{%
Two proof terms denoting reduction processes which are essentially the same, or more generally consist of the same steps performed in different order, should be recognised as being \emph{\peqent}.}
\usingContractionActivity{%
Two proof terms can be the result of arranging the same contraction activity in different ways, regarding parallelism/nesting degree, sequential order, and/or localisation of contractions.
Such proof terms should be recognised as being \emph{\peqent}.}

In this section we give a criterion to decide equivalence between proof terms. 
The approach is to extend the \emph{\peqence} criterion, as it is defined in \cite{terese} Sec. 8.3, to the infinitary setting.
\Peqence, for which the notation \peq\ will be used henceforth in this document, is defined there for finitary proof terms as the congruence generated by the following equation schemes 
\[
\begin{array}{lrcl}
\peqidleft & 1 \comp \psi & \eqnpeq & \psi \\
\peqidright & \psi \comp 1 & \eqnpeq & \psi \\
\peqassoc & \psi \comp (\phi \comp \chi) 
	& \eqnpeq & 
	(\psi \comp \phi) \comp \chi \\
\peqstruct & f(\psi_1, \ldots, \psi_m) \comp f(\phi_1, \ldots, \phi_m) 
	& \eqnpeq & 
	f(\psi_1 \comp \phi_1, \ldots, \psi_m \comp \phi_m) \\
\peqoutin & \mu(\psi_1, \ldots, \psi_m) 
	& \eqnpeq & 
	\mu(s_1, \ldots, s_m) \comp r[\psi_1, \ldots, \psi_m] \\
\peqinout & \mu(\psi_1, \ldots, \psi_m) 
	& \eqnpeq & 
	l[\psi_1, \ldots, \psi_m] \comp \mu(t_1, \ldots, t_m) 
\end{array}
\]
where $\mu: l \to r$, $s_i = src(\psi_i)$ and $t_i = tgt(\psi_i)$.

\usingStrEqIdEq{%
Additionally, two subrelations of \peq\ are remarked in \cite{terese}, namely the congruence generated by the first three and first four equations. 
These relations are known as \emph{identity equivalence} and \emph{structural equivalence} respectively. 
In this document, the notation $\ideq$ and $\streq$ resp. will be used for these relations%
\footnote{The notations used for the three introduced relations differs from those used in \cite{terese}, which I find somewhat confusing.}.%
}%

\medskip
Some challenges must be addressed in order to extend the \peqence\ definition to the infinitary setting.
Consider \eg\ the rules $\mu: f(x) \to g(x)$, $\nu: g(x) \to h(x)$ and $\rho: j(x) \to k(x)$, and the \redseqs\ \\[2pt]
\minicenter{
$j(f\om) \infred j(g\om) \sstep k(g\om)$ \qquad\qquad
$j(f\om) \to k(f\om) \infred k(g\om)$} \\[2pt]
which can be denoted by the proof terms \\[2pt]
\minicenter{
$\icomp j(g^i(\mu(f\om))) \comp \rho(g\om)$ \qquad\qquad
$\rho(f\om) \comp \icomp k(g^i(\mu(f\om)))$} \\[2pt]
respectively.
\usingRedseqOnly{%
These reduction sequences are \peqent. 
Both consist of a $\rho$ head step and an infinite number of $\mu$ steps forming a convergent reduction sequence of length $\omega$. 
}%
\usingContractionActivity{%
These proof terms denote the same contraction activity, namely a $\rho$ step transforming the head $j$ into a $k$, and an infinite number of $\mu$ steps transforming each occurrence of $f$ into one of $g$.
Therefore, they should be stated as \peqent. 
Observe that both proof terms are sequential, denoting precisely each of the described \redseqs.
}%
The difference lies in the order in which the two operations are performed: first the $\mu$ steps and then the $\rho$ step in the sequence to the left, and viceversa in the sequence of the right. The difference is apparent in the proof terms who describe the sequences.
Both considered \redseqs\ are convergent.

\smallskip
In order to equate %
\usingRedseqOnly{these \redseqs}%
\usingContractionActivity{the given \redseqs}%
, an \emph{infinite} number of step permutations must be performed: the $\rho$ step must be permuted in turn with each of the infinite $\mu$ steps. 
It is even impossible to determine which should be the \emph{first} $\mu$ step to be permuted with the $\rho$ step in order to transform the sequence to the left into that to the right. 
If we proceed the other way around, we can by finite means permute the initial $\rho$ step with a finite prefix of the infinite $\mu$ reduction, obtaining 
$j(f\om) \to j(g(f\om)) \to \ldots \to j(g^n(f\om)) \to k(g^n(f\om)) \infred k(g\om)$,
but there will always be an infinite $\mu$ sequence ``still to be permuted'' with the $\rho$ step.

This situation is reflected in the %
\usingRedseqOnly{proof terms}%
\usingContractionActivity{sequential proof terms}%
. There is no way to extract a ``last'' component in the infinite composition $\icomp j(g^i(\mu(f\om)))$, in order to permute it with $\rho(g\om)$. On the other hand, by applying the congruence on \peqence\ equations to infinitary terms, we can permute the leading $\rho(f\om)$ with a finite number of component of the following infinite composition in $\rho(f\om) \comp \icomp k(g^i(\mu(f\om)))$, \ie \\
$
\begin{array}{@{\hspace*{1cm}}cl}
\multicolumn{2}{l}{\rho(f\om) \comp \icomp k(g^i(\mu(f\om)))} \\ 
  \peq & 
  \rho(f\om) \comp k(\mu(f\om)) \comp \icomp k(g^{i+1}(\mu(f\om))) \\
  \peq & 
  j(\mu(f\om)) \comp \rho(g(f\om)) \comp \icomp k(g^{i+1}(\mu(f\om))) \\
  \peq & 
  j(\mu(f\om)) \comp \rho(g(f\om)) \comp k(g(\mu(f\om))) \comp \icomp k(g^{i+2}(\mu(f\om))) \\
  \peq & 
  j(\mu(f\om)) \comp j(g(\mu(f\om))) \comp \rho(g^2(f\om)) \comp \icomp k(g^{i+2}(\mu(f\om))) \\
  \peq & \ldots \\
  \peq & 
  j(\mu(f\om)) \comp \ldots \comp \rho(g^n(f\om)) \comp \icomp k(g^{i+n}(\mu(f\om))) \\  
\end{array}$ \\
therefore having still an infinite composition to the right of the $\rho$ step.
An adequate characterisation of \peqence\ for the infinitary setting should sanction the equivalence of these sequential proof terms.

\usingContractionActivity{%
\smallskip
Moreover, notice that all the redexes contracted in (the activity included in) either considered \redseq\ are present in the source term $j(f\om)$, so that the same activity can be denoted also by an \imstep\ (\ie\ a fully nested proof term), which is $\nu(\mu\om)$. Combinations of sequential and nested descriptions are possible as well, \eg\ $\rho(\icomp g^i(\mu(f\om)))$ and $\rho(f\om) \comp k(\mu\om)$. A sound \peqence\ characterisation should allow to state the equivalence of either of these proof terms \wrt\ any of the sequential versions introduced before.

To conclude the \peqence\ of either sequential proof term and (say) the multistep counterpart, an infinite number of step (de)nesting, using the \peqoutin\ or \peqinout\ equations, should be performed.

Using congruence on equations, a finite (though arbitrary) number of (de)nestings can be performed. 
\Eg\ the equivalence between
$\rho(f\om) \comp \icomp k(g^i(\mu(f\om)))$ and 
$\rho(\mu^3(f\om)) \comp \icomp k(g^{i+3}(\mu(f\om)))$ can be proved by nesting the three outer $\mu$ steps inside the $\rho$-step, as follows: \\
$
\begin{array}{@{\hspace*{1cm}}cl}
\multicolumn{2}{l}{\rho(f\om) \comp \icomp k(g^i(\mu(f\om)))} \\ 
  \peq & 
  \rho(f\om) \comp k(\mu(f\om)) \comp k(g(\mu(f\om))) \comp k(g(g(\mu(f\om)))) \comp \icomp k(g^{i+3}(\mu(f\om))) \\
  \peq & 
	\rho(f\om) \comp k(\mu(f\om)) \comp k(g(\mu(f\om) \comp g(\mu(f\om)))) \comp \icomp k(g^{i+3}(\mu(f\om))) \\
  \peq & 
	\rho(f\om) \comp k(\mu(f\om)) \comp k(g(\mu(\mu(f\om)))) \comp \icomp k(g^{i+3}(\mu(f\om))) \\
  \peq & 
	\rho(f\om) \comp k(\mu(f\om) \comp g(\mu(\mu(f\om)))) \comp \icomp k(g^{i+3}(\mu(f\om))) \\
  \peq & 
	\rho(f\om) \comp k(\mu^3(f\om)) \comp \icomp k(g^{i+3}(\mu(f\om))) \\
  \peq & 
	\rho(\mu^3(f\om)) \comp \icomp k(g^{i+3}(\mu(f\om))) \\
\end{array}$ \\
We describe briefly this schematic description of the \peqence\ derivation.
Firstly, \peqassoc\ is used to separate the first components of the infinite composition; notice that in this abrigded description, other uses of \peqassoc\ are left implicit.
Then \peqstruct\ is used twice (albeit described as one ``step'' in this description) from $k(g(\mu(f\om))) \comp k(g(g(\mu(f\om))))$, \wrt\ the symbols $k$ and $g$ respectively, thus obtaining $k(g(\mu(f\om) \comp g(\mu(f\om))))$.
This allows to subsequently apply \peqoutin\ on $\mu(f\om) \comp g(\mu(f\om)$, yielding $\mu(\mu(f\om))$.
The fourth and fifth lines describe a similar process, applied in order to obtain a concise description of the first three $\mu$ steps. This description is furthermore condensed with the leading $\rho(f\om)$ step, by applying \peqoutin\ once more.

An analogous process can be performed with any finite number of $\mu$ steps, yielding $\rho(f\om) \comp \icomp k(g^i(\mu(f\om))) \peq \rho(\mu^n(f\om)) \comp \icomp k(g^{i+n}(\mu(f\om)))$. 
In any case, there will always remain an infinite quantity of $\mu$ steps separated from the nested part.
}%

\medskip
Let us analyse an additional example using the same rules. Consider the \redseqs \\[2pt]
\minicenter{
$f\om \infred g\om \infred h\om$ \quad and \quad
$f\om \sstep g(f\om) \sstep h(f\om) \sstep h(g(f\om)) \sstep h^2(f\om) \infred h\om$} \\[2pt]
which can be denoted by the sequential proof terms \\[2pt]
\minicenter{
$\icomp g^i(\mu(f\om)) \comp \icomp h^i(\nu(g\om))$ \qquad and \qquad
$\icomp (h^i(\mu(f\om)) \comp h^i(\nu(f\om)))$} \\[2pt]
respectively.

Again, the reduction sequences are equivalent: they consist of an infinite number of $\mu$ steps and an infinite number of $\nu$ steps. In the left-hand sequence, first all the $\mu$ steps are performed, followed by the $\nu$ steps. In the right-hand sequence, $\mu$ and $\nu$ steps are interleaved.
Therefore, the proof terms describing these reductions should be sanction as \peqent.

We remark that in this case, \emph{each of the infinite number} of $\nu$ steps must be permuted with an infinite number of $\mu$ steps.
We will see that this added complexity of the needed permutations on \redseqs\ is reflected in the \peqent\ characterisation for proof terms, by means of an additional device needed to cope with this case.

\usingContractionActivity{%
The contraction activity included in either \redseq\ can also be described by non-sequential proof terms, remarkably $\mu\om \comp \nu\om$, but also \eg\
$\mu(f\om) \comp \icomp k^i(\nu(\mu(f\om)))$.
In this case, as the contraction of \emph{created} redexes is involved (since each $\nu$ step is created by the corresponding $\mu$ step), there is no way of describing this contraction activity by an \imstep.
}%

\includeStandardisation{%
\medskip
Notice that in both examples given so far, the \redseq\ shown to the right is standard, so that the permutation which equate the \redseq\ to the left to that to the right corresponds to a \emph{standardisation} of the former. 
}%

\usingRedseqOnly{%
\medskip
Proof terms denoting sequential vs. nested version of the same reduction must be considered as well. 
\Eg, the characterisation of \peqence\ for the finitary case allows to obtain
$\mu(\mu(a)) \peq \mu(f(a)) \comp g(\mu(a))$.
In the infinitary realm, an infinite number of steps could be nested in an infinite multistep.
The corresponding characterisation of \peqence\ should take such cases into account, allowing to obtain \eg\
$\mu\om \peq \icomp g^i(\mu(f\om))$.
}%

\bigskip
We remark that even when the characterisation of \peqence\ to be introduced can be applied to any well-formed proof term, the study of infinitary rewriting based on this characterisation we develop afterwards, mostly applies only to \emph{convergent} proof terms. Therefore, most of the additional definitions and results to come assume that the proof terms under consideration are convergent.
A study of \peqence\ considering also \emph{divergent} proof terms is left as future work.

\subsection{The formal infinitary \peqence\ relation}
In the following, we formally state the \peqence\ criterion we propose for infinitary proof terms.
As we have indicated in the introduction to this Section, the definition will be based on equational logic, so that a set of basic equations and another of equational rules will be introduced. 
The basic equations model the basic operations needed to perform a permutation of steps using the description of contraction activity given by proof terms, while the rules model the equivalence closure and the closure by the operations corresponding to the symbols in the signature of proof terms.
The need to reason about (proof terms including) infinite concatenations implies the inclusion of one  equation schema and one rule which specifically account for their infinite nature.
Therefore, the relations which formalise the notion of \peqence\ use an explicit form of \emph{infinitary equational logic}. 

In order to obtain a formal \peqence\ relation that is intuitively adequate, \ie\ which models adequately the concept of \peqence\ behaving as expected in a variety of examples, a very special rule must be added to the rules corresponding to equivalence and operations closure.
This rule allows to incorporate the idea of \emph{limit} into infinitary equational logic judgements.
In turn, to obtain an intuitively reasonable ``limit rule'', some particular requirements must be put in its premises, to limit the way in which this rule can be applied in a judgement. These requirements force to define a separate, previous ``base'' relation, which is used to define the ``limit rule'' for the \peqence\ relation. We will use $\peqe$ to denote the ``base'' relation, and $\peq$ for \peqence.

In the rest of this work, we will need to reason about the base \peqence\ relation. As we want to be able to proceed by some sort of transfinite induction on the complexity of the \peqence\ judgement, we will give a \emph{layered} definition of \peqence, like we did for the definition of proof terms in Sec.~\ref{sec:pterm}. Therefore, we will define, for each countable ordinal $\alpha$, the relations $\layerpeqe{\alpha}$ and $\layerpeq{\alpha}$.
Induction on \peqent\ terms can be performed by induction on the (say, minimal) layer to which the pair of terms belongs. The same holds for terms related by the ``base'' \peqence\ relation.

Formal definitions of the $\peqe$ and $\peq$ relations follow:

\begin{definition}[Layer of base \peqence]
\label{dfn:layer-peqe}
Let $\alpha$ be a countable ordinal.
We define the $\alpha$-th \emph{level of base \peqence}, notation $\layerpeqe{\alpha}$, as follows: given $\psi$ and $\phi$ proof terms, $\psi \layerpeqe{\alpha} \phi$ iff the equation $\psi \layerpeqx{\alpha} \phi$ can be obtained by means of the equational logic system whose basic equations are the instances of the following schemata for which both lhs and rhs are proof terms%
\footnote{hence they are particularly \emph{closed} terms, \confer\ Dfn.~\ref{dfn:imstep} and Dfn.~\ref{dfn:layer-pterm}.}
\[
\begin{array}{lrcl}
\peqidleft & 1 \comp \psi & \eqnpeq & \psi \\
\peqidright & \psi \comp 1 & \eqnpeq & \psi \\
\peqassoc & \psi \comp (\phi \comp \chi) 
	& \eqnpeq & 
	(\psi \comp \phi) \comp \chi \\
\peqstruct & f(\psi_1, \ldots, \psi_m) \comp f(\phi_1, \ldots, \phi_m) 
	& \eqnpeq & 
	f(\psi_1 \comp \phi_1, \ldots, \psi_m \comp \phi_m) \\
\peqinfstruct &
  \icomp f(\psi^1_i, \ldots, \psi^m_i)
	& \eqnpeq & 
  f(\icomp \psi^1_i, \ldots, \icomp \psi^m_i) \\
\peqoutin & \mu(\psi_1, \ldots, \psi_m) 
	& \eqnpeq & 
	\mu(s_1, \ldots, s_m) \comp r[\psi_1, \ldots, \psi_m] \\
\peqinout & \mu(\psi_1, \ldots, \psi_m) 
	& \eqnpeq & 
	l[\psi_1, \ldots, \psi_m] \comp \mu(t_1, \ldots, t_m) 
\end{array}
\]
verifying also the following conditions: $1 = src(\psi)$ for $\peqidleft$; 
$\psi$ convergent and $1 = tgt(\psi)$ for $\peqidright$; 
$\mu: l \to r$ for both \peqoutin\ and \peqinout; 
$s_i = src(\psi_i)$ for \peqoutin;
$\psi_i$ convergent and $t_i = tgt(\psi_i)$ for all $i$ for \peqinout.

Equational logic rules are defined by transfinite recursion on $\alpha$ as follows%
\footnote{
An alternative could be to consider \emph{open} instances of the equations, \ie\ one instance of \peqstruct\ and \peqinfstruct\ for each object function symbol plus one instance of \peqinout\ and \peqoutin\ for each rule symbol, where all the $\psi_i$, $\phi_i$, $\chi_i$, $s_i$ and $t_i$ would be considered as variables.
In order to equate instances of the such generated equations, a \emph{substitution} rule should be added at the equational logic level. 
In this way, considering the rules $\nu(x) : g(x) \to h(x)$, $\rho(x) : j(x) \to k(x)$ and $\pi : a \to b$, the equivalence $\nu(\rho(\pi)) \peqe \nu(j(a)) \comp h(\rho(\pi))$ would be justified by a two-step reasoning: a step using \eqleqn\ to obtain $\nu(\psi) \peqe \nu(s) \comp k(\psi)$ by the $\nu$ instance of the \peqoutin\ equtation, followed by the replacement of the $\psi$ and $s$ variables by the proof term $\rho(\pi)$ and its source, namely $j(a)$, by resorting to the substitution rule.

Unfortunately, this would be a rather inadequate approach because of the characteristics of proof terms in general, and of some of the equations in particular.
On one hand, an eventual extension of the set of proof terms in order to encompass open terms \emph{would not be closed by substitutions}.
A simple example considering the rule $\pi : a \to b$ follows: while $x \comp x$ would be a legal proof term, $\pi \comp \pi$ is not.
I guess this fact lies behind the difficulties for handling concatenation in the proposal of proof terms for HRS described in \cite{bruggink2008}; \confer\ particularly page 33.
On the other hand, not any instance of the equations correspond to their intent.
Firstly, the equation instance should correspond to valid proof terms at both lhs and rhs. Additionally, for the \peqinout equation, the $t_i$s are intended to be precisely $tgt(\psi_i)$, and not an arbitrary proof term verifying $tgt(\psi_i) = src(t_i)$. A similar condition holds for \peqoutin.
Observe that all these restrictions are considered when defining the set of legal instances of equations which can be used when applying the \eqleqn\ rule.
} %
\\[5pt]
$\begin{array}{c}
 \begin{array}{c} \\ \hline \psi \layerpeqx{1} \psi \end{array} 
 \ \ \eqlrefl
 \qquad 
 \begin{array}{c} \psi \eqnpeq \phi \textnormal{ is a basic equation}\\ \hline \psi \layerpeqx{1} \phi \end{array} 
 \ \ \eqleqn
 \\ \\
 \begin{array}{c} 
   \psi \layerpeqx{\alpha_1} \phi \\ 
 	 \hline 
   \phi \layerpeqxb{\alpha_1 + 1} \psi 
 \end{array} 
 \ \ \eqlsymm
 \qquad
 \begin{array}{c} 
   \psi \layerpeqx{\alpha_1} \phi \quad \phi \layerpeqx{\alpha_2} \xi \\ 
 	 \hline 
   \psi \layerpeqxb{\alpha_1 + \alpha_2 + 1} \xi 
 \end{array} 
 \ \ \eqltrans
 \\ \\
 \begin{array}{c} 
   \psi_1 \layerpeqx{\alpha_1} \phi_1 \quad \ldots \quad \psi_n \layerpeqx{\alpha_n} \phi_n \quad
   f/n \in \Sigma \\ 
 	 \hline 
   f(\psi_1, \ldots, \psi_n) \layerpeqxb{\alpha_1 + \ldots + \alpha_n + 1} f(\phi_1, \ldots, \phi_n)
 \end{array} 
 \ \ \eqlfun
 \\ \\
 \begin{array}{c} 
   \psi_1 \layerpeqx{\alpha_1} \phi_1 \quad \ldots \quad \psi_n \layerpeqx{\alpha_n} \phi_n \quad
   \mu/n \textrm{ is a rule symbol} \\ 
 	 \hline 
   \mu(\psi_1, \ldots, \psi_n) \layerpeqxb{\alpha_1 + \ldots + \alpha_n + 1} \mu(\phi_1, \ldots, \phi_n)
 \end{array} 
 \ \ \eqlrule
 \\ \\
 \begin{array}{c} 
   \psi_1 \layerpeqx{\alpha_1} \phi_1 \quad \psi_2 \layerpeqx{\alpha_2} \phi_2 \\ 
 	 \hline 
   \psi_1 \comp \psi_2 \layerpeqxb{\alpha_1 + \alpha_2 + 1} \phi_1 \comp \phi_2
 \end{array} 
 \ \ \eqlcomp
 \qquad
 \begin{array}{c} 
   \psi_i \layerpeqx{\alpha_i} \phi_i \quad \textforall i < \omega \\ 
 	 \hline 
   \icomp \psi_i \ \layerpeqxb{\Sigma_{i < \omega} \alpha_i}\  \icomp \phi_i
 \end{array} 
 \ \ \eqlinfcomp
 \end{array}$
\end{definition}

\begin{definition}[Base \peqence]
\label{dfn:peqe}
Let $\psi$, $\phi$ be proof terms. 
We say that $\psi$ and $\phi$ are \emph{base-\peqent}, notation $\psi \peqe \phi$, iff $\psi \layerpeqe{\alpha} \phi$ for some $\alpha < \omega_1$.
\end{definition}

\includeStandardisation{%
\medskip
The subrelation of $\peqe$ whose definition follows will be used in some proofs. We say that $\psi \peqefs \phi$ (where $EFS$ stands for ``equational, finitary and structural'') iff $\psi \peqe \phi$ can be obtained by using only the equations \peqidleft, \peqidright, \peqassoc\ and \peqstruct, and not resorting to the \eqlinfcomp\ equational rule.
}%

\begin{definition}[Layer of \peqence]
\label{dfn:layer-peq}
Let $\alpha$ be a countable ordinal.
We define the $\alpha$-th \emph{level of \peqence}, notation $\layerpeq{\alpha}$, as follows: given $\psi$ and $\phi$ proof terms, $\psi \layerpeq{\alpha} \phi$ iff the equation $\psi \layerpeqx{\alpha} \phi$ can be obtained by means of the equational logic system whose basic equations are those described in Dfn.~\ref{dfn:layer-peqe}, and the set of equational logic rules is the result of adding the rule \eqllim\ defined as follows
\\[5pt]
$\begin{array}{c}
 \begin{array}{c} 
   \begin{array}{ll}
	   \left.
     \begin{array}{l}
     \psi \layerpeqe{\alpha_k} \chi_k \comp \psi'_k \ \ \mind{\psi'_k} > k 
	   \\
	   \phi \layerpeqe{\beta_k} \chi_k \comp \phi'_k \ \ \mind{\phi'_k} > k 
		 \end{array}
		 \right\}
   &
	   \textforall k < \omega
	 \end{array}
	 \\ 
 	 \hline 
   \psi \ \layerpeqx{\alpha}\  \phi 
	 \qquad
	 \textnormal{where } \alpha = \sum_{i < \omega} \alpha_i + \sum_{i < \omega} \beta_i
 \end{array} 
 \ \ \eqllim
 \end{array}$
\\[5pt]
to the rules introduced in Dfn.~\ref{dfn:layer-peqe}.
\end{definition}

Notice that the explicit reference to the relations $\layerpeqe{\alpha_k}$ and $\layerpeqe{\beta_k}$ prevents the ``stacking'' of uses of the rule \eqllim\ in a \peqence\ judgement, \ie, that judgements leading to the premises of an application of the \eqllim\ rule cannot include other applications of the same rule. 
This condition does not imply that a valid \peqence\ judgement can include at most one occurrence of \eqllim. \Eg\ a \peqence\ derivation having the following shape

\smallskip
$
\prooftree
  \[
		\ldots
		\ \ 
		\begin{array}{l}
		\psi_1 \peqe \xi_k \comp \psi'_1 \\	
		\phi_1 \peqe \xi_k \comp \phi'_1
		\end{array} 
		\ \ 
		\ldots
	\justifies
		\psi_1 \peq \phi_1
	\using
		\eqllim
	\]
	\quad
	\[
		\begin{array}{l}
		\psi_2 \peqe \chi_k \comp \psi'_2 \\	
		\phi_2 \peqe \chi_k \comp \phi'_2
		\end{array} 
		\ \ 
		\ldots
	\justifies
		\psi_2 \peq \phi_2
	\using
		\eqllim
	\]
\justifies
	\psi_1 \comp \psi_2 \peq \phi_1 \comp \phi_2
\using
	\eqlcomp
\endprooftree
$

\smallskip\noindent
is valid according to Dfn.~\ref{dfn:layer-peq}.

\begin{definition}[\Peqence]
\label{dfn:peq}
Let $\psi$, $\phi$ be proof terms. 
We say that $\psi$ and $\phi$ are \emph{\peqent}, notation $\psi \peq \phi$, iff $\psi \layerpeq{\alpha} \phi$ for some $\alpha < \omega_1$.
\end{definition}

Observe that for any countable ordinal $\alpha$, $\layerpeqe{\alpha} \,\subseteq\, \layerpeq{\alpha}$, and therefore $\peqe \,\subseteq\, \peq$.

\medskip
As discussed prior to the formal definitions, this characterisation of \peqence\ for infinitary proof terms adds, to the rules corresponding to the closure of the description of step permutation, a rule which allows to resort to the concept of \emph{limit} inside judgements. We found this necessary to obtain a complete characterisation, \ie, one which covers all the examples we have studied.

If the difference between the activity denoted by two proof terms can be proven to tend to zero, then we can resort to limits to assert that such difference is equal to zero, and therefore, that the proof terms must be considered equivalent.
The measure used to compute the difference between two proof terms \wrt\ their denoted activity is the \emph{minimal activity depth}.

\bigskip
The equational logic used to reason about infinitary derivations adds three features to its finitary counterpart, besides operating on infinitary proof terms instead of just finite ones.
These additions are: the \peqinfstruct\ \emph{equation schema}, and the \eqlinfcomp\ and \eqllim\ \emph{equational rules}.

\medskip
The first addition is the generalisation of \peqstruct\ to the infinite composition. 
It allows \eg\ the following \peqence\ reasoning \\
$\icomp j(g^i(\mu(f\om))) \comp \rho(g\om)
 \peq
 j(\icomp g^i(\mu(f\om))) \comp \rho(g\om)
 \peq
 \rho(\icomp g^i(\mu(f\om)))
 \peq
 \rho(f\om) \comp k(\icomp g^i(\mu(f\om)))$ \\
thus addressing the first example given in the introduction of this Section.

\smallskip
We observe that \peqinfstruct\ includes occurrences of an infinite number of variables: for each $j$ from 1 to the arity of $f$, $\psi^j_i$ is a distinct variable for each $i$ verifying $0 \leq i < \omega$. 
On the other hand, the restriction to convergent proof terms imposes a convergence condition to the substitutes for these variables when applying this equation%
\footnote{more precisely, an equation corresponding to this equation scheme}.
The use, in equational logic, of a convergence condition as a restriction for the application of an equation having occurrences of an infinite number of different variables, could be the object of further analysis.


\medskip
The equational rule \eqlinfcomp\ allows transformations to be performed in each term of an infinite composition.

Consider the proof terms
$\psi_1 \eqdef \icomp (j(h^i(\mu(f\om))) \comp j(h^i(\nu(f\om)))) \comp \rho(h\om)$
and \\
$\psi_2 \eqdef \rho(f\om) \comp \icomp (k(h^i(\mu(f\om))) \comp k(h^i(\nu(f\om))))$, which represent equivalent \redseqs.
In order to transform $\psi_1$ into $\psi_2$, the $\rho$ step must be permuted with the preceding infinite composition, which in turn must be transformed into a proof term having the form $j(\psi'_1)$ in order to enable the permutation to be applied using the equations \peqinout\ and then \peqoutin.
To perform the desired transformation to 
$\icomp (j(h^i(\mu(f\om))) \comp j(h^i(\nu(f\om))))$,
the equation \peqstruct\ must be applied \emph{on each of the infinite number of components}, so obtaining 
$\icomp j(h^i(\mu(f\om)) \comp h^i(\nu(f\om)))$, 
and then the equation \peqinfstruct\ transforms the latter into 
$j(\icomp h^i(\mu(f\om)) \comp h^i(\nu(f\om)))$.

The rule \eqlinfcomp\ allows to obtain
$\icomp (j(h^i(\mu(f\om))) \comp j(h^i(\nu(f\om))))
 \peqe
 \icomp j(h^i(\mu(f\om)) \comp h^i(\nu(f\om)))$, taking as premises
$j(h^i(\mu(f\om))) \comp j(h^i(\nu(f\om))) \peqe 
 j(h^i(\mu(f\om)) \comp h^i(\nu(f\om)))$
for each $i < \omega$.
Therefore, the assertion $\psi_1 \peqe \psi_2$ can be justified by the following schematic equational judgement \\[5pt]
$\begin{array}{@{\hspace*{2cm}}cll}
 \multicolumn{2}{l}{\icomp (j(h^i(\mu(f\om))) \comp j(h^i(\nu(f\om)))) \comp \rho(h\om)} \\ 	
 \peqe &
 \icomp j(h^i(\mu(f\om)) \comp h^i(\nu(f\om))) \comp \rho(h\om)  \\
 \peqe &
 j(\icomp h^i(\mu(f\om)) \comp h^i(\nu(f\om))) \comp \rho(h\om) & \textnormal{by } \peqinfstruct \\
 \peqe &
 \rho(\icomp h^i(\mu(f\om)) \comp h^i(\nu(f\om))) \\
 \peqe &
 \rho(f\om) \comp k(\icomp h^i(\mu(f\om)) \comp h^i(\nu(f\om))) \\
 \peqe &
 \rho(f\om) \comp \icomp k(h^i(\mu(f\om)) \comp h^i(\nu(f\om))) &  \textnormal{by } \peqinfstruct \\
 \peqe &
 \rho(f\om) \comp \icomp (k(h^i(\mu(f\om))) \comp k(h^i(\nu(f\om)))) \\
\end{array}$ \\
where the first and last ``steps'' involve, in fact, an infinite number of equation occurrences.

When reasoning about convergent proof terms, the convergence conditions on the sequence $\langle \psi_i \rangle_{i < \omega}$ (resp. $\langle \phi_i \rangle_{i < \omega}$) for $\icomp \psi_i$ (resp. $\icomp \phi_i$) are implicit conditions to apply \eqlinfcomp. 
Particularly, the minimal activity depth of the components must tend to $\omega$ for both $\psi$ and $\phi$, thus entailing a convergence condition on the infinite number of premises. 
As we have remarked for the \peqinfstruct\ equation, the implications of such convergence conditions on equational reasoning could be object of future work.

\bigskip
To motivate the inclusion of the \eqllim\ equational rule, and consequently the need to define a separated \emph{base} relation, we recall the proof terms
$\psi_3 \eqdef \icomp g^i(\mu(f\om)) \comp \icomp h^i(\nu(g\om))$ 
and
$\psi_4 \eqdef \icomp (h^i(\mu(f\om)) \comp h^i(\nu(f\om)))$
from the introduction to this Section.
By using the base \peqence\ relation given in Dfn.~\ref{dfn:peqe}, we can permute the first $\nu$ step with all the $\mu$ steps but the first, obtaining
$\psi_3 \peqe \mu(f\om) \comp \nu(f\om) \comp h(\icomp g^i(\mu(f\om)) \comp \icomp h^i(\nu(g\om)))$; so that the first component in $\psi_4$ can be ``extracted'' from $\psi_3$.
Such a process can be repeated in order to ``extract'' more components, arriving to 
$\psi_3 \peqe \mu(f\om) \comp \nu(f\om) \comp \ldots \comp h^n(\mu(f\om)) \comp h^n(\nu(f\om)) \comp h^{n+1}(\icomp g^i(\mu(f\om)) \comp \icomp h^i(\nu(g\om)))$
for each $n < \omega$.
On the other hand, it is straightforward to observe that $\psi_4 \peqe \mu(f\om) \comp \nu(f\om) \comp \ldots \comp h^n(\mu(f\om)) \comp h^n(\nu(f\om)) \comp h^{n+1}(\icomp h^{i}(\mu(f\om)) \comp h^{i}(\nu(f\om)))$.
In order to conclude $\psi_3 \peq \psi_4$, it is needed to resort to the \eqllim\ rule added in Dfn.~\ref{dfn:peq}.
We observe that the minimal activity depth of the successive ``differences'' 
$h^{n+1}(\icomp g^i(\mu(f\om)) \comp \icomp h^i(\nu(g\om)))$ and
$h^{n+1}(\icomp h^{i}(\mu(f\om)) \comp h^{i}(\nu(f\om)))$ tend to infinity, as required in the premises of the \eqllim\ rule.

\begin{remark}
\label{rmk:peqinout-restriction}
We notice that the requirement of lhs and rhs convergence put on the instances of the equation schemata does not imply that every variable in a scheme must necessarily be replaced by a convergent proof term. 
\Eg, considering $\mu: f(x) \to g(x)$, $\nu: g(x) \to k(x)$, $\rho: h(x,y) \to j(y)$, and $\tau : i(x) \to x$, the following instance of \peqoutin:
$\rho(\tau\om, \mu(a) \cdot \nu(a)) \peqe \rho(i\om, f(a)) \comp j(\mu(a) \comp \nu(a))$,
is legal even when $\psi_1$ is replaced by the divergent proof term $\tau\om$.
Observe particularly that $\rho(\tau\om, \mu(a) \cdot \nu(a))$ is a convergent proof term: convergence of $\tau\om$ is not asked since the corresponding variable in the lhs of the $\rho$ rule does not occur in the rhs; \confer\ Dfn~\ref{dfn:layer-pterm}, case \ref{rule:ptsymbol}.

\smallskip
On the other hand, let us try to decompose the proof term $\rho(\tau\om, \mu(a) \cdot \nu(a))$ ``the other way around'', namely by using \peqinout\ instead of \peqoutin.
The form of \peqinout\ for proof terms having $\rho$ as root symbol is \\
\hspace*{1cm} 
$\rho(\psi_1, \psi_2) \peq h(\psi_1, \psi_2) \comp \rho(tgt(\psi_1), tgt(\psi_2))$ \\
In turn, replacing $\psi_1$ with $\tau\om$ and $\psi_2$ with $\mu(a) \comp \nu(a)$ yields \\
\hspace*{1cm} 
$\rho(\tau\om ,\, \mu(a) \comp \nu(a)) \peq h(\tau\om ,\, \mu(a) \comp \nu(a)) \comp \rho(tgt(\tau\om) ,\, tgt(\mu(a) \comp \nu(a)))$ \\
Therefore, applying the equation having the given proof term as left-hand side would require $tgt(\tau\om)$ to be defined, which is not the case.

A similar situation occurs with the (intuitively very simple) equation \peqidright. In this case, the target of the proof term at the right-hand side of an instance must be defined in order for the corresponding left-hand side to make sense. 

To avoid this kind of situations, an additional requirement will be put to the uses of the \eqleqn\ equational rule, when the equation involved is either \peqinout\ or \peqidright. 
For \peqidright, we ask $\psi$ to be convergent.
For \peqinout\, convergence must be asked, not only of the proof term $\mu(\psi_1, \ldots, \psi_m)$ at the left-hand side of the equation, but convergence must be required to all the (proof terms taking the place of each variable) $\psi_i$ as well.
Therefore, \wrt\ the motivating example, $\rho(\tau\om ,\, \mu(a) \comp \nu(a))$ is not a valid left-hand side to apply \peqinout, even if it is a convergent proof term.
When using either \peqinout\ or \peqidright\ in proofs involving the relation $\peqe$, it should be checked those uses to correspond to valid instances%
\footnote{We notice that in the development of the %
\includeStandardisation{standardisation proof}%
\doNotIncludeStandardisation{compression proof in Section~\ref{sec:compression}}%
, the uses of \peqinout\ correspond to situations in which the convergence of the proof term to be put at the right-hand side is known in advance. In fact, the intent of the uses of this equation is to ``obtain'' a condensed form of some contraction activity, corresponding to the left-hand side, in order to subsequently decomposing the obtained condensed form in a top-to-bottom fashion, through the \peqoutin\ equation. \Confer\ the proof of Lem.~\ref{rsl:jump-one-step}.}%
.

The equations \peqidright\ and \peqinout\ are the only elements in Dfn.~\ref{dfn:peqe} for which a well-formed proof term being the element for one side in a possible instance does not have a convergent proof term as the correspondent element for the other side%
\footnote{This claim will be proved shortly.}.
\end{remark}

\includeStandardisation{\subsubsection{Basic properties of \peqence}}
\doNotIncludeStandardisation{\subsection{Basic properties of \peqence}}
\label{sec:peqence-basic-properties}

\begin{lemma}
\label{rsl:peq-then-same-src-mind-tgt}
Let $\psi$, $\phi$ be convergent proof terms such that $\psi \peq \phi$. Then $src(\psi) = src(\phi)$, $tgt(\psi) = tgt(\phi)$ and $mind(\psi)  = mind(\phi)$.
\end{lemma}

\begin{proof}
We proceed by induction on $\alpha$ where $\psi \layerpeq{\alpha} \phi$, analysing the equational logic rule used in the final step of that judgement.
Observe particularly that Lem~\ref{rsl:mind-big-then-tdist-little}:(\ref{it:convergent-then-has-tgt}) implies both $tgt(\psi)$ and $tgt(\phi)$ to be defined.
If the rule is \eqleqn, then we analyse the equation of which the pair $\pair{\psi}{\phi}$ is an instance.
It turns out that the only non-trivial cases are those corresponding to the \peqinfstruct\ equation and the \eqlinfcomp\ and \eqllim\ rules. We prove the result for each of these cases.

\medskip
Assume that $\pair{\psi}{\phi}$ is an instance of the \peqinfstruct\ equation, \ie, that \\
$\psi = \icomp f(\psi^1_i, \ldots, \psi^m_i)$ and 
$\phi = f(\icomp \psi^1_i, \ldots, \icomp \psi^m_i)$. 
\begin{itemize}
\item 
We verify $\mind{\psi} = \mind{\phi}$. \\
Observe that 
$\mind{\psi} 
		= min_{i < \omega}(\mind{f(\psi^1_i, \ldots, \psi^m_i)})
		= \mind{f(\psi^1_a, \ldots, \psi^m_a)} 
		= 1 + min(\mind{\psi^1_a}, \ldots, \mind{\psi^m_a})
		= 1 + \mind{\psi^b_a}$
where 
\begin{eqnarray}
\mind{f(\psi^1_a, \ldots, \psi^m_a)} & \leq & \mind{f(\psi^1_i, \ldots, \psi^m_i)}
		\quad \textforall i < \omega
		\label{eq:infstruct-mind-1} \\
\mind{\psi^b_a} & \leq & \mind{\psi^j_a} 
		\quad \textif 1 \leq j \leq m
		\label{eq:infstruct-mind-2} 
\end{eqnarray}
On the other hand, 
$\mind{\phi} 
		= 1 + min(\mind{\icomp \psi^1_i}, \ldots, \mind{\icomp \psi^m_i})
		= 1 + \mind{\icomp \psi^{b'}_i}
		= 1 + \mind{\psi^{b'}_{a'}}$
where
\begin{eqnarray}
\mind{\icomp \psi^{b'}_i} & \leq & \mind{\icomp \psi^{j}_i}
		\quad \textif 1 \leq j \leq m
		\label{eq:infstruct-mind-3} \\
\mind{\psi^{b'}_{a'}} & \leq & \mind{\psi^{b'}_{i}}
		\quad \textforall i < \omega
		\label{eq:infstruct-mind-4}
\end{eqnarray}
Assume for contradiction $\mind{\psi^{b}_{a}} < \mind{\psi^{b'}_{a'}}$.
Then $b \neq b'$ would imply 
$\mind{\icomp \psi^{b}_i} \leq \mind{\psi^{b}_{a}} 
		< \mind{\psi^{b'}_{a'}} = \mind{\icomp \psi^{b'}_i}$, contradicting (\ref{eq:infstruct-mind-3}), and $b = b'$ would immediately contradict (\ref{eq:infstruct-mind-4}).
Analogously, if we assume $\mind{\psi^{b'}_{a'}} < \mind{\psi^{b}_{a}}$, then $a \neq a'$ would imply 
$\mind{f(\psi^1_{a'}, \ldots, \psi^m_{a'})}	\leq 1 + \mind{\psi^{b'}_{a'}}
		< 1 + \mind{\psi^{b}_{a}} = \mind{f(\psi^1_a, \ldots, \psi^m_a)}$, contradicting (\ref{eq:infstruct-mind-1}), and $a = a'$ would immediately contradict (\ref{eq:infstruct-mind-2}).
Hence we conclude.

\item
To verify the condition about source terms, it is enough to observe that
$src(\psi) = src(\phi) = f(src(\psi^1_0), \ldots, src(\psi^m_0))$.

\item
We verify $tgt(\psi) = tgt(\phi)$. 
Observe that 
$tgt(\psi) = \lim_{i \to \omega} f(tgt(\psi^1_i), \ldots, tgt(\psi^m_i))$ and 
$tgt(\phi) = f(\lim_{i \to \omega} tgt(\psi^1_i), \ldots, \lim_{i \to \omega} tgt(\psi^m_i))$. \\
Let $t_j \eqdef \lim_{i \to \omega} tgt(\psi^j_i)$, so that $tgt(\phi) = f(t_1, \ldots, t_m)$.
Then it is enough to prove that $\tdist{tgt(\psi)}{f(t_1, \ldots, t_m)} = 0$. \\
Let $n < \omega$. Let $k$ such that for all $j$, $i > k$ implies $\tdist{tgt(\psi^j_i)}{t_j} < 2^{-(n-1)}$ and also
$\tdist{f(tgt(\psi^1_i), \ldots, tgt(\psi^m_i))}{tgt(\psi)} < 2^{-n}$. \\
Let $i \eqdef k+1$. Then 
$\tdist{f(tgt(\psi^1_i), \ldots, tgt(\psi^m_i))}{f(t_1, \ldots, t_m)} = $\\
		$\frac{1}{2} * max(\tdist{tgt(\psi^1_i)}{t_1}, \ldots, \tdist{tgt(\psi^1_m)}{t_m})
		< 2^{-n}$.
Hence Lem.~\ref{rsl:tdist-is-ultrametric} yields 
$\tdist{tgt(\psi)}{f(t_1,\ldots,t_m)} < 2^{-n}$.
Thus we conclude.
\end{itemize}

Assume that the rule justifying $\psi \layerpeq{\alpha} \phi$ is \eqlinfcomp, so that $\psi = \icomp \psi_i$, $\phi = \icomp \phi_i$, and for all $i < \omega$, $\psi_i \layerpeq{\alpha_i} \phi_i$ where $\alpha_i < \alpha$. \\
Source terms: it is enough to apply \ih\ on $\psi_0 \layerpeq{\alpha_0} \phi_0$ obtaining $src(\psi) = src(\psi_0) = src(\phi_0) = src(\phi)$. \\
Target terms and $\mindfn$:
Observe that \ih\ can be applied on each $\psi_i \layerpeq{\alpha_i} \phi_i$, yielding $tgt(\psi_i) = tgt(\phi_i)$ and $\mind{\psi_i} = \mind{\phi_i}$.
Then recalling the definitions of target and $\mindfn$ on $\psi$ and $\phi$ suffices to conclude.

\medskip
Assume that the rule used in the last step of the judgement $\psi \layerpeq{\alpha} \phi$ is \eqllim, so that for all $n < \omega$, 
$\psi \layerpeqe{\alpha_n} \chi_n \comp \psi'_n$ and $\phi \layerpeqe{\alpha_n} \chi_n \comp \phi'_n$, where $\mind{\psi'_n} > n$, $\mind{\phi'_n} > n$, $\alpha_n < \alpha$ and $\beta_n < \alpha$.
Observe that $\layerpeqe{\alpha} \,\subseteq\, \layerpeq{\alpha}$ for any ordinal $\alpha$, so that \ih\ can be applied to any premise of the \eqllim\ rule. \\[2pt]
Source terms: applying \ih\ on $\psi \layerpeq{\alpha_0} \chi_0 \comp \psi'_0$ and $\phi \layerpeq{\alpha_0} \chi_0 \comp \phi'_0$, we obtain $src(\psi) = src(\phi) = src(\chi_0)$. \\[2pt]
Target terms: we prove $\tdist{tgt(\psi)}{tgt(\phi)} = 0$.
Let $n < \omega$. 
Then \ih\ on $\psi \layerpeq{\alpha_n} \chi_n \comp \psi'_n$ and $\phi \layerpeq{\alpha_n} \chi_n \comp \phi'_n$ yields 
$tgt(\psi) = tgt(\psi'_n)$ and $tgt(\phi) = tgt(\phi'_n)$. Moreover, it is immediate to obtain $src(\psi'_n) = src(\phi'_n) = tgt(\chi_n)$. \\
Recalling that $\mind{\psi'_n} > n$ and $\mind{\phi'_n} > n$, Lem.~\ref{rsl:mind-big-then-tdist-little} can be applied to obtain 
$\tdist{tgt(\chi_n)}{tgt(\psi)} = \tdist{src(\psi'_n)}{tgt(\psi'_n)} < 2^{-n}$ and analogously 
$\tdist{tgt(\chi_n)}{tgt(\phi)} = \tdist{src(\phi'_n)}{tgt(\phi'_n)} < 2^{-n}$.
Therefore Lem.~\ref{rsl:tdist-is-ultrametric} yields $\tdist{tgt(\psi)}{tgt(\phi)} < 2^{-n}$. Thus we conclude. \\[2pt]
Minimal activity depth:
Assume for contradiction $n \eqdef \mind{\psi} < \mind{\phi}$.
Observe $\psi \peq \chi_n \comp \psi'_n$ and $\phi \peq \chi_n \comp \phi'_n$, where $\mind{\psi'_n} > n$ and $\mind{\phi'_n} > n$.
Then $\mind{\psi} = n$ implies $\mind{\chi_n} = n$, and therefore $\mind{\phi} = n$, contradicting the assumption.
The assertion $\mind{\phi} < \mind{\psi}$ can be contradicted analogously.
Thus we conclude.
\end{proof}

\medskip
The result about $\mindfn$ and $src$ allows to prove that $\peqe$ is closed \wrt\ the set of convergent proof terms. 

\begin{lemma}
\label{rsl:peqe-soundness-convergent}
Let $\psi$ and $\phi$ proof terms such that $\psi \peqe \phi$.
Then $\psi$ is a well-formed and convergent proof term iff $\phi$ is.
\end{lemma}

\begin{proof}
We proceed by induction on $\alpha$ where $\psi \layerpeqe{\alpha} \phi$, analysing the equational rule used in the last step in the corresponding $\peqe$ derivation.

If the rule is \eqleqn, then we analyse the basic equation used.
\begin{itemize}
\item 
\peqidleft, \ie\ $\psi = src(\phi) \comp \phi$. It is immediate to verify the desired result.  
\item 
\peqidright, \ie\ $\psi = \phi \comp tgt(\phi)$. Observe that Remark~\ref{rmk:peqinout-restriction} implies that $\phi$ must be a convergent proof term. Thus we conclude immediately.
\item
\peqassoc, \ie\ $\psi = \chi \comp (\xi \comp \gamma)$ and $\phi = (\chi \comp \xi) \comp \gamma$. In this case, $\psi$ is well-formed iff $\phi$ is well-formed iff $\chi$, $\xi$ and $\gamma$ are well formed, and moreover $\chi$ and $\xi$ are convergent. Moreover, $\psi$ is convergent iff $\phi$ is convergent iff $\gamma$ is convergent. Thus we conclude.
\item
\peqstruct, \ie\ $\psi = f(\chi_1, \ldots, \chi_m) \comp f(\xi_1, \ldots, \xi_m)$ and $\phi = f(\chi_1 \comp \xi_1, \ldots, \chi_m \comp \xi_m)$.
In this case, $\psi$ is well formed iff $\phi$ is well-formed iff all $\chi_i$ and $\xi_i$ are well-formed, all the $\chi_i$ are also convergent (\confer\ Lem.~\ref{rsl:fnsymbol-convergence} for $\psi$), and $tgt(\chi_i) = src(\xi_i)$ for all $i$.
Moreover, 
$\psi$ is convergent
  iff all the $\xi_i$ are convergent (\confer\ again Lem.~\ref{rsl:fnsymbol-convergence})
	iff all the $\chi_i \comp \xi_i$ are convergent
	iff $\phi$ is convergent.
Thus we conclude.
\item
\peqinfstruct, \ie\ $\psi = \icomp f(\chi^1_i, \ldots, \chi^m_i)$ and $\phi = f(\icomp \chi^1_i, \ldots, \icomp \chi^m_i)$.

$\Rightarrow )$
Assume that $\psi$ is well-formed and convergent. Given $n < \omega$, let $k_n < \omega$ be an index verifying $\mind{f(\chi^1_i, \ldots \chi^m_i)} > n$ if $k_n < i$.
Let $j$ such that $1 \leq j \leq m$.
Then for all $i < \omega$, $f(\chi^1_i, \ldots \chi^m_i)$ convergent implies $\chi^j_i$ convergent, \confer\ Lem~\ref{rsl:fnsymbol-convergence}.
In turn $src(f(\chi^1_{i+1}, \ldots \chi^m_{i+1})) = tgt(f(\chi^1_i, \ldots \chi^m_i))$ implies immediately $src(\chi^j_{i+1}) = tgt(\chi^j_i)$. 
Finally, if $i > k_{n+1}$, then $\mind{f(\chi^1_i, \ldots \chi^m_i)} > n + 1$ implies $\mind{\chi^j_i} > n$. Hence $\icomp \chi^j_i$ is well-formed and convergent. Consequently, so is $\phi$.

$\Leftarrow )$
Assume that $\phi$ is well-formed and convergent. Given $j$ such that $1 \leq j \leq m$ and $n < \omega$, let $k_{(n,j)}$ be an index verifying $\mind{\psi^j_i} > n$ if $k_{(n,j)} < i$.
Let $i < \omega$. Then $\chi^j_i$ convergent and $src(\psi^j_{i+1}) = tgt(\psi^j_i)$ for all $j$ implies $f(\chi^1_i, \ldots, \chi^m_i)$ convergent and $src(f(\chi^1_{i+1}, \ldots \chi^m_{i+1})) = tgt(f(\chi^1_i, \ldots \chi^m_i))$. Then $\psi$ is a well-formed proof term.
Moreover, for all $n < \omega$, if $i > max \set{k_{(n,j)} \setsthat 1 \leq j \leq m}$, then $\mind{f(\chi^1_i, \ldots, \chi^m_i)} > n$. Consequently, $\psi$ is convergent.

\item
\peqinout, \ie\ $\psi = \mu(\chi_1, \ldots, \chi_m)$ and $\phi = l[\chi_1, \ldots, \chi_m] \comp \mu(t_1, \ldots, t_m)$.
In this case, Remark~\ref{rmk:peqinout-restriction} implies that all $\chi_i$ are convergent proof terms.
Then both $\psi$ and $\phi$ are well-formed and convergent.

\item
\peqoutin, \ie\ $\psi = \mu(\chi_1, \ldots, \chi_m)$ and $\phi = \mu(s_1, \ldots, s_m) \comp r[\chi_1, \ldots, \chi_m]$.
In this case $\psi$ is well-formed iff $\phi$ is well-formed iff $\chi_i$ are well-formed.
Moreover, $\psi$ is convergent iff $\phi$ is convergent iff all $\chi_i$ corresponding to variables occurring in the right-hand side $r$, which are exactly those occurring in $r[\chi_1, \ldots, \chi_m]$, are convergent; \confer\ Lem.~\ref{rsl:rulesymbol-convergence} and Lem.~\ref{rsl:ctx-convergence} respectively.	
\end{itemize}

\smallskip
If the equational rule used in the last step of the derivation ending in $\psi \layerpeqe{\alpha} \phi$ is \eqlrefl, \eqlsymm\ or \eqltrans, then a straightforward argument suffices to conclude.

If the rule is \eqlfun, \eqlrule\ or \eqlcomp, then a simple argument based on Lem.~\ref{rsl:fnsymbol-convergence}, Lem~\ref{rsl:rulesymbol-convergence} or just Dfn.~\ref{dfn:layer-pterm} case~(\ref{rule:ptbinC}) respectively, and \ih, suffices to conclude.

Assume that the rule used in the last step of the derivation is \eqlinfcomp. As the rule is symmetric, then it suffices to prove one side of the biconditional in the lemma statement.
Then assume that $\psi = \icomp \psi_i$ is a well-formed and convergent proof term.
Let $i < \omega$. Then $\psi_i$ is convergent and $src(\psi_{i+1}) = tgt(\psi_i)$.
Therefore \ih\ implies convergence of $\phi_i$, and Lem.~\ref{rsl:peq-then-same-src-mind-tgt} yields $src(\phi_{i+1}) = tgt(\phi_i)$. Hence $\phi$ is well-formed.
Let $n < \omega$. Then convergence of $\psi$ implies the existence of some $k_n < \omega$ verifying $\mind{\psi_i} > n$ if $k_n < i$. In turn, Lem.~\ref{rsl:peq-then-same-src-mind-tgt} implies $\mind{\phi_i} > n$ if $k_n < i$. Consequently, $\psi$ is convergent. 
\end{proof}

\medskip
The following lemma shows that \peqence\ is compatible with infinitary contexts.

\begin{lemma}
\label{rsl:peqe-compatible-ctx}
Let $C$ be a context having $k < \omega$ holes, and $\langle \psi_i \rangle_{i \leq k}$ and $\langle \phi \rangle_{i \leq k}$ two sequences of proof terms verifying $\psi_i \peqe \phi_i$ for all $i$.
Then $C[\psi_1, \ldots, \psi_k] \peqe C[\phi_1, \ldots, \phi_k]$.
\end{lemma}

\begin{proof}
An easy induction on $max \set{\posln{\BPos{C}{i}}}$ suffices. Resort to the \eqlfun\ equational rule for the inductive case.
\end{proof}

\medskip
\doNotIncludeStandardisation{\renewcommand{\peqefs}{\peqe}}
The following lemma shows that the \peqstruct\ equation can be extended to contexts having a finite number of holes.
\begin{lemma}
\label{rsl:struct-ctx}
Let $C$ be a context in $\Sigma$ (\ie\ built from function symbols only) having exactly $n < \omega$ occurrences of the box; and $\psi_1, \ldots, \psi_n$, $\phi_1, \ldots, \phi_n$ proof terms.
Then 
$C[\psi_1, \ldots, \psi_n] \comp C[\phi_1, \ldots, \phi_n]
\peqefs 
C[\psi_1 \comp \phi_1, \ldots, \psi_n \comp \phi_n]$.
\end{lemma}

\begin{proof}
We proceed by induction on 
$max(\set{\posln{\BPos{C}{i}}})$.

If $C = \Box$, then we conclude immediately, notice that in this case $n = 1$.

Otherwise $C = f(C_1, \ldots, C_m)$. In this case \\
$C[\psi_1, \ldots, \psi_n] \comp C[\phi_1, \ldots, \phi_n] \ = $ \\
\hspace*{1cm} $f(C_1[\psi_1, \ldots, \psi_{k1}], \ldots, C_m[\psi_{k(m-1)+1}, \ldots, \psi_{n}]) \ \comp $ \\
\hspace*{1cm} $f(C_1[\phi_1, \ldots, \phi_{k1}], \ldots, C_m[\phi_{k(m-1)+1}, \ldots, \phi_{n}])$, and \\
$C[\psi_1 \comp \phi_1, \ldots, \psi_n \comp \phi_n] \ = $ \\
\hspace*{1cm} $f(C_1[\psi_1 \comp \phi_1, \ldots, \psi_{k1} \comp \phi_{k1}], \ldots, C_m[\psi_{k(m-1)+1} \comp \phi_{k(m-1)+1}, \ldots, \psi_{n} \comp \phi_{n}])$.
We conclude by \ih\ on each $C_i$, and then by the \eqlfun\ equational rule.
\end{proof}

\begin{lemma}
\label{rsl:trivial-pterm-peq-src}
Let $\psi$ be a trivial proof term. Then $\psi \peq src(\psi)$.
\end{lemma}

\begin{proof}
Observe $\psi \peqe src(\psi) \comp \psi$ by \peqidleft. 
On the other hand, $src(\psi) \peqe src(src(\psi)) \comp src(\psi) = src(\psi) \comp src(\psi)$, by \peqidleft\ and Dfn.~\ref{dfn:src-tgt-imstep} respectively; recall that $src(\psi)$ is a trivial \imstep.
Moreover, for any $n < \omega$, $\mind{\psi} = \mind{src(\psi)} = \omega > n$, \confer\ Lem.~\ref{rsl:trivial-pterm-mind-omega}.
Therefore the rule \eqllim\ can be applied to obtain $\psi \peq src(\psi)$.
\end{proof}

\newpage
\section{Denotation of \redseqs}
\label{sec:pterm-denotation}
As stated in Sec.~\ref{sec:pterm}, the aim of the introduction of proof terms is to denote and study \redseqs\ in infinitary rewriting.

A basic question arises: can \emph{any} \redseq\ be denoted by a proof term?
In order to answer this question, we will resort to proof terms which denote a \redseq\ in a close, stepwise way, without condensing parallel or embedded steps. 
Formally, we will define a proper subset of the set of valid proof terms, which we will call \emph{\ppterms}, which include only (denotation of) single steps and dots. 
Then we will prove that any \redseq\ whose length is a countable ordinal can be denoted by means of a \ppterm. Observe that particularly this result applies to all convergent \redseqs, \confer\ Thm.~2 in \cite{inf-normalization}.

\medskip
Once denotation of all countable-length \redseqs\ is stated, the issue of \emph{uniqueness of stepwise denotation} arises.
It is easy to realize that stepwise denotation of a \redseq\ is not unique, because of different ``bracketings'', \ie\ different ways to associate dots.
A simple example follows, using the rules 
$\mu(x) : f(x) \to g(x)$, $\nu(x) : g(x) \to k(x)$, $\rho(x,y) : h(x,y) \to x$.
The proof terms $(\rho(f(a),b) \comp \mu(a)) \comp \nu(a)$ and $\rho(f(a),b) \comp (\mu(a) \comp \nu(a))$ are different stepwise denotations of the same \redseq, namely $h(f(a),b) \to f(a) \to g(a) \to k(a)$.
On the other hand, observe that these proof terms are \peqent, and moreover, its equivalence can be stated by using only the equation \peqassoc.

In the finitary setting, it is fairly intuitive that \ppterms\ being \emph{denotationally equivalent}, \ie\ such that they denote the same \redseq, can be proven to be \peqent\ by ``rebracketing'', \ie\ by applying equational logic using only the \peqassoc\ equation.
The reciprocal property also holds: if two \ppterms\ are \emph{rebracketing equivalent} (or, phrased differently, ``equal up to rebracketing'') then they denote the same \redseq.

\medskip
The concepts we have just introduced allow to state the question about denotation uniqueness in a more precise way: 
\textbf{do denotational and rebracketing equivalences coincide}?

For the finitary case, it is fairly simple to prove that the answer to this question is positive.
Indeed, by orienting the \peqassoc\ equation in either direction, \emph{standard} denotations of \redseqs\ can be obtained. 
These standard \ppterms\ can also be seen as the result of coherently associating dots to the left or to the right.

For \ppterms\ denoting infinite \redseqs, the question seems less obvious.
\Eg\ consider the sequence 
$f\om \to g(f\om) \to g(g(f\om)) \infred g\om$
which can be denoted \eg\ by the \ppterms\
$\icomp g^i(\mu(f\om))$ and $\icomp g^{2*i}(\mu(f\om)) \comp g^{2*i+1}(\mu(f\om))$.
For any $n < \omega$, it is easy to obtain, using only the equation \peqassoc, that
$\psi \peqe 
		( \mu(f\om) \comp \ldots \comp g^{2*n+1}(\mu(f\om)) ) \comp g^{2*(n+1)}(\psi)$
and 
$\phi \peqe 
		( \mu(f\om) \comp \ldots \comp g^{2*n+1}(\mu(f\om)) ) \comp g^{2*(n+1)}(\phi)$.
Then we can obtain $\psi \peq \phi$ \textbf{by resorting to a limit argument}, \ie\ by applying the \eqllim\ rule.
On the other hand, we did not find a way to justify \peqence\ between these \ppterms\ which avoids the use of \eqllim.

In this Section we will prove that, provided the characterisation of \peqence\ given in Sec.~\ref{sec:peqence}, denotational and rebracketing equivalences do coincide for infinitary term rewriting.
The corresponding proofs make evident the role of the limit \peqence\ argument in order to verify this coincidence.


\subsection{\Ppterms}
\label{sec:ppterm}
\denotationDistributed{%
In this section we will introduce a particular subclass of the set of proof terms, namely the \emph{\ppterms}. 
Just like \imsteps\ are in correspondence with maximal developments of \orthoredexsets, so that any maximal development can be \emph{denoted} by an \imstep (\confer\ \refsec{mstep-orthoredexset}); \ppterms\ are related with the set of all \redseqs\ for a given \TRS, so that \emph{any \redseq} can be denoted by a \ppterm.
A formal characterisation of the idea of denoting a \redseq\ by a (stepwise) proof term will be given in the following Section~\ref{sec:pterm-denotation}.
}%
\denotationInOwnChapter{%
In the following, we introduce the set of \ppterms, give some additional related definitions and state some basic properties of this subset of the set of valid proof terms.
}%

\begin{definition}[One-step]
\label{dfn:one-step}
A \emph{one-step} is an \imstep\ including exactly one occurrence of a rule symbol.
If $\psi$ is a one-step, then we define the redex position of $\psi$, notation $\RPos{\psi}$, as the position of the unique rule symbol occurrence in $\psi$, and the depth of $\psi$, notation $\sdepth{\psi}$, as $\posln{\RPos{\psi}}$; \confer\ \refdfn{step} for the analogy with the corresponding notions as defined for a reduction step.
\end{definition}

\begin{definition}[\Ppterm, \Pnpterm]
\label{dfn:ppterm}
A \emph{\ppterm} is any proof term $\psi$ whose formation satisfies any of the following conditions, where we refer to cases in \refdfn{layer-pterm}:
\begin{itemize}
	\item $\psi$ is a one-step, so it is built by case \ref{rule:ptmstep},
	\item $\psi$ is built by case \ref{rule:ptinfC}, so that $\psi = \icomp \psi_i$, and all of the $\psi_i$ are \ppterms, or
	\item $\psi$ is built by case \ref{rule:ptbinC}, so that $\psi = \psi_1 \comp \psi_2$, and both $\psi_1$ and $\psi_2$ are \ppterms.
\end{itemize}
A \emph{\pnpterm} is any proof term $\psi$ such that either $\psi$ is a \ppterm\ or $\psi \in \iSigmaTerms$.
\end{definition}

\begin{definition}[Steps of a \pnpterm]
\label{dfn:steps}
For any $\psi$ \pnpterm, we define the number of \emph{steps} of $\psi$, notation $\ppsteps{\psi}$, as the countable ordinal defined as follows: \\
\begin{tabular}{l}
if $\psi \in \iSigmaTerms$, then $\ppsteps{\psi} \eqdef 0$. \\
if $\psi$ is a one-step, then $\ppsteps{\psi} \eqdef 1$. \\
if $\psi = \icomp \psi_i$ then $\ppsteps{\psi} \eqdef \sum_{i < \omega} \ppsteps{\psi_i}$; \confer \refdfn{ordinal-infAdd}. \\
if $\psi = \psi_1 \comp \psi_2$ then $\ppsteps{\psi} \eqdef \ppsteps{\psi_1} + \ppsteps{\psi_2}$.
\end{tabular}
\end{definition}

\begin{lemma}
\label{rsl:steps-ordinal-coherence}
Let $\psi$ be a \ppterm, and let $\alpha$ the ordinal such that $\psi \in \layerpterm{\alpha}$. Then $\ppsteps{\psi}$ is a limit ordinal iff $\alpha$ is.
\end{lemma}

\begin{proof}
Easy induction on $\alpha$ where $\psi \in \layerpterm{\alpha}$.
\end{proof}

\begin{definition}[$\alpha$-th component of a \ppterm]
\label{dfn:ppterm-component}
Let $\psi$ be a \ppterm\ and $\alpha$ an ordinal such that $\alpha < \ppsteps{\psi}$. We define the \emph{$\alpha$-th component} of $\psi$, notation $\psi[\alpha]$, as the one-step defined as follows: \\
\begin{tabular}{p{.95\textwidth}}
if $\psi$ is a one-step, then $\psi[0] \eqdef \psi$. \\
if $\psi = \icomp \psi_i$, then there are unique $k$ and $\gamma$ such that $\alpha = \ppsteps{\psi_0} + \ldots + \ppsteps{\psi_{k-1}} + \gamma$ and $\gamma < \ppsteps{\psi_k}$; \confer\ \reflem{ordinal-lt-infAdd-then-unique-representation}. We define $\psi[\alpha] \eqdef \psi_k[\gamma]$. \\
if $\psi = \psi_1 \comp \psi_2$ and $\alpha < \ppsteps{\psi_1}$ then $\psi[\alpha] \eqdef \psi_1[\alpha]$. \\
if $\psi = \psi_1 \comp \psi_2$ and $\ppsteps{\psi_1} \leq \alpha$, then $\psi[\alpha] \eqdef \psi_2[\beta]$ such that $\ppsteps{\psi_1} + \beta = \alpha$.
\end{tabular}
\end{definition}

\begin{definition}
\label{dfn:maxd}
Let $\psi$ be a \ppterm\ such that $\ppsteps{\psi} < \omega$.
Then we define the \emph{maximal depth activity} of $\psi$ as $\maxd{\psi} \eqdef max(\sdepth{\redel{\psi}{n}} \setsthat n < \ppsteps{\psi})$.
We also define the \emph{maximal step depth} of $\psi$ as $\maxsd{\psi} \eqdef max(\Pdepth{\mu} \setsthat \mu \in R)$ where $R$ is the set of all the rule symbols occurring in $\psi$.
\end{definition}

\medskip
We show some expected properties of the components of a \ppterm.
These properties particularly entail that a \ppterm\ can be seen as the concatenation of its components, so that the particular way in which they are associated is irrelevant. 
\denotationDistributed{%
More on this in Section~\ref{sec:peqence}, specifically in Section~\ref{sec:peqence-basic-properties} and Section~\ref{sec:frso}.}

\begin{lemma}
\label{rsl:ppterm-mind-big-then-tdist-little}
Let $\psi$ be a \ppterm, $\alpha$ an ordinal and $n < \omega$, such that $\mind{\psi} > n$ and $\alpha < \ppsteps{\psi}$. Then 
\begin{enumerate}
	\item \label{it:ppterm-mind-big-then-step-mind-big}
	$\sdepth{\redel{\psi}{\alpha}} > n$.
	\item \label{it:ppterm-mind-big-then-tdist-little-step}
	$\tdist{src(\redel{\psi}{\alpha})}{tgt(\redel{\psi}{\alpha})} < 2^{-n}$.
	\item \label{it:ppterm-mind-big-then-tdist-little-src}
	$\tdist{src(\psi)}{tgt(\redel{\psi}{\alpha})} < 2^{-n}$.
\end{enumerate}
\end{lemma}

\begin{proof}
We proceed by induction on $\psi$, \confer\ Prop.~\ref{rsl:pterm-induction-principle}.
If $\psi$ is a one-step then $\alpha = 0$ and $\redel{\psi}{\alpha} = \psi$.
Then we conclude immediately; \confer\ Lemma~\ref{rsl:mind-big-then-tdist-little} for (\ref{it:ppterm-mind-big-then-tdist-little-step}) and (\ref{it:ppterm-mind-big-then-tdist-little-src}).

Assume $\psi = \psi_1 \comp \psi_2$.
If $\alpha < \ppsteps{\psi_1}$, so that $\redel{\psi}{\alpha} = \redel{\psi_1}{\alpha}$, then we conclude by \ih\ on $\psi_1$.
Otherwise $\alpha = \ppsteps{\psi_1} + \beta$, so that $\redel{\psi}{\alpha} = \redel{\psi_2}{\beta}$. Then by applying \ih\ on $\psi_2$ we obtain (\ref{it:ppterm-mind-big-then-step-mind-big}) and (\ref{it:ppterm-mind-big-then-tdist-little-step}) immediately, and also $\tdist{src(\psi_2)}{tgt(\redel{\psi}{\alpha})} < 2^{-n}$.
On the other hand we can apply Lemma~\ref{rsl:mind-big-then-tdist-little} to $\psi_1$, obtaining $\tdist{src(\psi)}{tgt(\psi_1)} < 2^{-n}$. Thus we conclude by Lemma~\ref{rsl:tdist-is-ultrametric} since $tgt(\psi_1) = src(\psi_2)$.

Assume $\psi = \icomp \psi_i$. Let $k$, $\beta$ such that $\redel{\psi}{\alpha} = \redel{\psi_k}{\beta}$, so that $\beta < \ppsteps{\psi_k}$.
Then \ih\ on $\psi_k$ yields immediately (\ref{it:ppterm-mind-big-then-step-mind-big}) and (\ref{it:ppterm-mind-big-then-tdist-little-step}), and also $\tdist{src(\psi_k)}{tgt(\redel{\psi}{\alpha})} < 2^{-n}$.
On the other hand, for each $i < k$ it is immediate that $\mind{\psi_i} \geq \mind{\psi} > n$, then an easy induction on $k$ using Lemma~\ref{rsl:mind-big-then-tdist-little} and Lemma~\ref{rsl:tdist-is-ultrametric} yields $\tdist{src(\psi)}{src(\psi_{k})} < 2^{-n}$. Thus we conclude by Lemma~\ref{rsl:tdist-is-ultrametric}.
\end{proof}

\begin{lemma}
\label{rsl:ppterm-mind-big-then-tdist-little-tgt}
Let $\psi$ be a convergent \ppterm\ such that $\mind{\psi} > p$, and $\alpha < \ppsteps{\psi}$.
Then $\tdist{tgt(\redel{\psi}{\alpha})}{tgt(\psi)} < 2^{-p}$.
\end{lemma}

\begin{proof}
We proceed by induction on $\psi$.
If $\psi$ is a one-step then $\alpha = 0$ and it suffices to observe that $\redel{\psi}{0} = \psi$.

Assume $\psi = \psi_1 \comp \psi_2$.
If $\alpha < \ppsteps{\psi_1}$, then \ih\ on $\psi_1$ yields $\tdist{tgt(\redel{\psi}{\alpha})}{tgt(\psi_1)} < 2^{-p}$.
On the other hand, Lemma~\ref{rsl:mind-big-then-tdist-little} implies $\tdist{src(\psi_2)}{tgt(\psi)} < 2^{-p}$.
We conclude by Lemma~\ref{rsl:tdist-is-ultrametric} since $tgt(\psi_1) = src(\psi_2)$.
Otherwise, $\alpha = \ppsteps{\psi_1} + \beta$, then $\redel{\psi}{\alpha} = \redel{\psi_2}{\beta}$.
In this case we can apply \ih\ on $\psi_2$ obtaining $\tdist{tgt(\redel{\psi_2}{\beta})}{tgt(\psi_2)} < 2^{-p}$, thus we conclude.

Assume $\psi = \icomp \psi_i$ and let $k$, $\gamma$ such that $\redel{\psi}{\alpha} = \redel{\psi_k}{\gamma}$.
Then \ih\ on $\psi_k$ yields $\tdist{tgt(\redel{\psi}{\alpha})}{tgt(\psi_k)} < 2^{-p}$.
Moreover, Lemma~\ref{rsl:mind-big-then-tdist-little} on $\icomp \psi_{k+1+i}$ implies $\tdist{src(\psi_{k+1})}{tgt(\psi)} < 2^{-p}$. Ths we conclude by Lemma~\ref{rsl:tdist-is-ultrametric}.
\end{proof}

\begin{lemma}
\label{rsl:ppterm-src}
Let $\psi$ be a \ppterm. Then $src(\redel{\psi}{0}) = src(\psi)$.
\end{lemma}

\begin{proof}
Easy induction on $\psi$.
\end{proof}

\begin{lemma}
\label{rsl:ppterm-tgt-successor}
Let $\psi$ be a \ppterm\ such that $\ppsteps{\psi} = \alpha + 1$. Then $tgt(\psi) = tgt(\redel{\psi}{\alpha})$.
\end{lemma}

\begin{proof}
We proceed by induction on $\psi$.
If $\psi$ is a one-step then $\alpha = 0$ and we conclude immediately.

Assume $\psi = \psi_1 \comp \psi_2$.
Then $\alpha < \ppsteps{\psi_1}$ would imply $\alpha + 1 = \ppsteps{\psi} \leq \ppsteps{\psi_1}$, which is not possible since $\ppsteps{\psi_2} > 0$.
Then let $\beta$ be the ordinal verifying $\ppsteps{\psi_1} + \beta = \alpha$, so that $\redel{\psi}{\alpha} = \redel{\psi_2}{\beta}$. 
We observe that $\ppsteps{\psi_1} + \beta + 1 = \alpha + 1 = \ppsteps{\psi}$, then $\ppsteps{\psi_2} = \beta + 1$.
We conclude by \ih\ on $\psi_2$.

Finally, $\psi = \icomp \psi_i$ contradicts $\ppsteps{\psi}$ to be a successor ordinal. Thus we conclude.
\end{proof}

\begin{lemma}
\label{rsl:ppterm-tgt-limit}
Let $\psi$ be a convergent \ppterm\ such that $\ppsteps{\psi}$ is a limit ordinal. 
Then $tgt(\psi) = \lim_{\alpha \to \ppsteps{\psi}} tgt(\redel{\psi}{\alpha})$.
\end{lemma}

\begin{proof}
Observe $\ppsteps{\psi}$ being a limit ordinal implies $\psi = \icomp \psi_i$
\denotationInOwnChapter{(\confer\ Lem.~\ref{rsl:steps-ordinal-coherence} and Lem.~\ref{rsl:ptinfC-iff-limit})}%
, so that $tgt(\psi)$ is defined to be equal to $\lim_{i \to \omega} tgt(\psi_i)$. Observe that Lem~\ref{rsl:mind-big-then-tdist-little}:(\ref{it:convergent-then-has-tgt}) implies this limit to be defined.
Let $p \in \Nat$, let $k'$ such that $k' < j < \omega$ implies $\tdist{tgt(\psi_j)}{tgt(\psi)} < 2^{-p}$, $k''$ such that $\mind{\psi_j} > p$ if $j > k''$, and $k \eqdef max(k', k'')$.

Let $\beta = \ppsteps{\psi_0} + \ldots + \ppsteps{\psi_k}$ and $\gamma > \beta$.
Then $\gamma = \ppsteps{\psi_0} + \ldots + \ppsteps{\psi_j} + \gamma'$ where $\gamma' < \ppsteps{\psi_{j+1}}$ and $j \geq k$, so that $\redel{\psi}{\gamma} = \redel{\psi_{j+1}}{\gamma'}$. 
Then $j + 1 > k \geq k''$, so that Lemma~\ref{rsl:ppterm-mind-big-then-tdist-little-tgt} implies $\tdist{tgt(\redel{\psi}{\gamma})}{tgt(\psi_{j+1})} < 2^{-p}$.
On the other hand, $j + 1 > k \geq k'$ implies $\tdist{tgt(\psi_{j+1})}{tgt(\psi)} < 2^{-p}$.
Hence Lemma~\ref{rsl:tdist-is-ultrametric} yields $\tdist{tgt(\redel{\psi}{\gamma})}{tgt(\psi)} < 2^{-p}$.
Consequently, we conclude.
\end{proof}

\begin{lemma}
\label{rsl:ppterm-tgt-src-coherence}
Let $\psi$ be a \ppterm\ and $\alpha < \ppsteps{\psi}$ such that $\alpha = \alpha' + 1$. Then $src(\redel{\psi}{\alpha}) = tgt(\redel{\psi}{\alpha'})$.
\end{lemma}

\begin{proof}
We proceed by induction on $\psi$. Observe $\psi$ is a one-step would imply $\alpha = 0$, contradicting $\alpha = \alpha' + 1$.

Assume $\psi = \psi_1 \comp \psi_2$. We consider three cases
\begin{itemize}
\item 
If $\alpha < \ppsteps{\psi_1}$ then we conclude just by \ih\ on $\psi_1$.
\item
If $\alpha = \ppsteps{\psi_1}$, then $\redel{\psi}{\alpha} = \redel{\psi_2}{0}$ and $\redel{\psi}{\alpha'} = \redel{\psi_1}{\alpha'}$ where $\alpha' + 1 = \alpha = \ppsteps{\psi_1}$. 
Then $tgt(\redel{\psi}{\alpha'}) = tgt(\psi_1)$ and $src(\redel{\psi}{\alpha}) = src(\psi_2)$, by Lemma~\ref{rsl:ppterm-tgt-successor} and Lemma~\ref{rsl:ppterm-src} respectively. Thus we conclude.
\item
If $\alpha > \ppsteps{\psi_1}$, then $\alpha' = \ppsteps{\psi_1} + \beta'$ and $\alpha = \ppsteps{\psi_1} + (\beta' + 1)$, therefore $\redel{\psi}{\alpha} = \redel{\psi_2}{\beta' + 1}$ and $\redel{\psi}{\alpha'} = \redel{\psi_2}{\beta'}$. Observe that $\alpha < \ppsteps{\psi}$ implies $\beta' + 1 < \ppsteps{\psi_2}$. Hence we conclude by \ih\ on $\psi_2$.
\end{itemize}

Assume $\psi = \icomp \psi_i$. Let $k$, $\gamma$ such that $\alpha = \ppsteps{\psi_0} + \ldots + \ppsteps{\psi_{k-1}} + \gamma$ and $\gamma < \ppsteps{\psi_k}$, so that $\redel{\psi}{\alpha} = \redel{\psi_k}{\gamma}$.
If $\gamma = 0$, then $\ppsteps{\psi_{k-1}} = \beta + 1$ for some $\beta$, and $\alpha' = \ppsteps{\psi_0} + \ldots + \ppsteps{\psi_{k-2}} + \beta$, so that $\redel{\psi}{\alpha'} = \redel{\psi_{k-1}}{\beta}$.
Therefore $src(\redel{\psi}{\alpha}) = src(\psi_k)$ and $tgt(\redel{\psi}{\alpha'}) = tgt(\psi_{k-1})$, by Lemma~\ref{rsl:ppterm-src} and Lemma~\ref{rsl:ppterm-tgt-successor} respectively. Thus we conclude.
Otherwise $\gamma = \gamma' + 1$; notice that $\gamma$ being a limit ordinal would contradict $\alpha$ being a successor one.
In this case $\redel{\psi}{\alpha'} = \redel{\psi_k}{\gamma'}$, thus we conclude by \ih\ on $\psi_k$.
\end{proof}

\begin{lemma}
\label{rsl:ppterm-seq-mind}
Let $\psi$ be a \ppterm. 
Then \\ 
$\begin{array}{rcl}
\mind{\psi} & = &
min(\sdepth{\redel{\psi}{\alpha}} \setsthat \alpha < \ppsteps{\psi}) 
\denotationDistributed{
\\ 
& = &
\posln{\,min_{\leq_{DL}} (\RPos{\redel{\psi}{\alpha}} \setsthat \alpha < \ppsteps{\psi})}
}
\\
& = &
min(\mind{\redel{\psi}{\alpha}} \setsthat \alpha < \ppsteps{\psi}) 
\end{array}$ %
\denotationDistributed{%
, \\
where if $P$ is a set of positions, then $p \eqdef min_{\leq_{DL}} (P)$ is the element of $P$ verifying $p \leq_{DL} q$ if $q \in P$.} %
\end{lemma}

\begin{proof}
We prove that $\mind{\psi} = min(\mind{\redel{\psi}{\alpha}} \setsthat \alpha < \ppsteps{\psi})$.
The rest of the statement follows immediately since it is trivial to verify $\sdepth{\redel{\psi}{\alpha}} = \mind{\redel{\psi}{\alpha}}$ for any $\alpha$; \confer\ Dfn.~\ref{dfn:dmin-imstep}.

We proceed by induction on $\psi$; \confer\ Prop.~\ref{rsl:pterm-induction-principle}.
We define $\mindp{\psi} \eqdef min(\mind{\redel{\psi}{\alpha}} \setsthat \alpha < \ppsteps{\psi})$, so we must verify $\mind{\psi} = \mindp{\psi}$.
If $\psi$ is a one-step then the result holds immediately.

Assume $\psi = \psi_1 \comp \psi_2$.
In this case, \ih\ on $\psi_i$ yields $\mind{\psi_i} = \mindp{\psi_i}$ for each $i = 1,2$, and Dfn.~\ref{dfn:layer-pterm} implies $\mind{\psi} = min(\mind{\psi_1}, \mind{\psi_2})$.
Then it suffices to verify $\mindp{\psi} = min(\mindp{\psi_1}, \mindp{\psi_2})$.
From the definition of $\mindpfn$, it is immediate that $\mindp{\psi} \leq \mindp{\psi_i}$ for $i = 1,2$.
Assume $\mindp{\psi_1} \leq \mindp{\psi_2}$. Notice $\mindp{\psi} < \mindp{\psi_1}$ would imply the existence of some $\gamma$ verifying $\mindp{\redel{\psi}{\gamma}} < \mindp{\psi_1}$, contradicting either the definition of $\mindp{\psi_1}$ (if $\gamma < \ppsteps{\psi_1}$) or the assertion $\mindp{\psi_1} \leq \mindp{\psi_2}$ (otherwise). Hence $\mindp{\psi} = \mindp{\psi_1}$.
A similar argument for the case $\mindp{\psi_2} < \mindp{\psi_1}$ is enough to conclude.

If $\psi = \icomp \psi_i$, then an argument similar to that used for binary composition applies.
To verify that $\mindp{\psi} = min_{i < \omega}(\mindp{\psi_i})$, observe that $\mindp{\psi} \leq \mindp{\psi_i}$ for all $i$, and consider $n$ such that $\mindp{\psi_n} \leq \mindp{\psi_i}$ for all $i$.
Then we can contradict $\mindp{\psi} < \mindp{\psi_n}$ proceeding as in the previous case, hence $\mindp{\psi} = \mindp{\psi_n}$. Thus we conclude.
\end{proof}

\includeStandardisation{
\medskip
The view of a \ppterm\ as the concatenation of its components will be used extensively to obtain standardisation results in Section~\ref{sec:peqence}, where the ability of separating a \ppterm\ in head (\ie\ first component) and tail (\ie\ the concatenation of components from the second one on) will be needed as well. This consideration motivates the following formalisation of the concept of tail of a \ppterm.

\begin{definition}
\label{dfn:ppterm-tail}
Let $\psi$ be a proof term. We define the \emph{tail} of $\psi$, notation $\redfrom{\psi}{1}$, as follows: \\
\begin{tabular}{@{$\ \ \bullet\ \ $}p{.9\textwidth}}
If $\psi$ is a one-step, then $\redfrom{\psi}{1} \eqdef tgt(\psi)$. \\
If $\psi = \psi_1 \comp \psi_2$ and $\psi_1$ is a one-step, then $\redfrom{\psi}{1} \eqdef \psi_2$. \\
If $\psi = \psi_1 \comp \psi_2$ and $\psi_1$ is not a one-step, then $\redfrom{\psi}{1} \eqdef \redfrom{\psi_1}{1} \comp \psi_2$. \\
If $\psi = \icomp \psi_i$ and $\psi_0$ is a one-step, then $\redfrom{\psi}{1} \eqdef \icomp \psi_{1+i}$. \\
If $\psi = \icomp \psi_i$ and $\psi_0$ is not a one-step, then $\redfrom{\psi}{1} \eqdef \redfrom{\psi_0}{1} \comp (\icomp \psi_{1+i})$. \\
\end{tabular}
\end{definition}
}

\denotationInOwnChapter{%
\subsection{Denotation -- formal definition and proof of existence}
\label{sec:pterm-denotation-defs-existence}
}%
\denotationDistributed{
\subsection{Denotation of \redseqs}
\label{sec:pterm-denotation}
}%
In this section, we formalise the notion of a \pnpterm\ \emph{denoting} a \redseq, resorting to the definitions of length and $\alpha$-th component of \pnpterms, given in the presentation of such terms.
Then we prove the existence, for any \redseq\ having a countable ordinal length, of a \pnpterm\ which denotes it. 

\denotationInOwnChapter{%
As we have discussed in the introduction to Section~\ref{sec:pterm-denotation}, denotation of a \redseq\ is not unique. 
In the next subsection, we will investigate how to characterise the proof terms denoting the same \redseq.}

\begin{definition}[Denotation for reduction steps]
\label{dfn:redstep-denotation}
Let $a = \langle t, p, \mu \rangle$ be a reduction step, and $\psi$ a one-step. 
Then $\psi$ \emph{denotes} $a$ iff all the following apply: $src(\psi) = t$, $tgt(\psi) = tgt(a)$, and $\psi(p) = \mu$, therefore $\sdepth{a} = \mind{\psi}$.
\end{definition}

\begin{definition}[Mapping from one-steps to reduction steps]
\label{dfn:rstepden}
Let \trst\ be a \TRS. We define the mapping $\rstepdenfn$ from the set of one-steps for \trst\ to the set of reduction steps for \trst, as follows: 
$\rstepden{\psi} \eqdef \langle src(\psi), \RPos{\psi}, \psi(\RPos{\psi}) \rangle$.
\end{definition}

\begin{lemma}
\label{rsl:rstepden-denotes}
Let $\psi$ be a one-step and $\stepa$ a reduction step. 
Then $\psi$ denotes $a$ iff $a = \rstepden{\psi}$.
\end{lemma}

\begin{proof}
We prove each direction of the biconditional. 

\noindent
$\Rightarrow )$:
Let us say $a = \langle t, p, \mu \rangle$. Hypotheses imply immediately $t = src(\psi)$, and also $\psi(p) = \mu$, so that $p = \RPos{\psi}$ and $\mu = \psi(\RPos{\psi})$. Thus we conclude.
\noindent
$\Leftarrow )$:
Let us say $\rstepden{\psi} = \langle t, p, \mu \rangle$ and $\mu: l \to h$.
Then it is immediate from Dfn.~\ref{dfn:rstepden} to verify $src(\psi) = t$ and $\psi(p) = \mu$.
In turn, observe that $tgt(\psi) = \repl{\psi}{h[t_1, \ldots, t_m]}{p}$ where $\subtat{\psi}{p} = \mu(t_1, \ldots, t_m)$, and
$t = src(\psi) = \repl{\psi}{l[t_1, \ldots, t_m]}{p}$, so that it is straightforward to verify $tgt(\rstepden{\psi}) = tgt(\psi)$. Thus we conclude.
\end{proof}

\begin{definition}[Denotation for \redseqs]
\label{dfn:redseq-denotation}
Let \reda\ be a \redseq, and $\psi$ a \pnpterm. 
We will say that $\psi$ \emph{denotes} \reda\ iff $\ppsteps{\psi} = \redln{\reda}$, $src(\psi) = src(\reda)$ and $\redel{\psi}{\alpha}$ denotes $\redel{\reda}{\alpha}$ for all $\alpha < \redln{\reda}$.
\end{definition}

\begin{lemma}
\label{rsl:redseq-denotation-implications}
Let \reda\ be a \redseq, and $\psi$ a \pnpterm, such that $\psi$ denotes \reda.
Then $\mind{\psi} = \mind{\reda}$, $\psi$ is convergent iff \reda\ is, and in that case, $tgt(\psi) = tgt(\reda)$.
\end{lemma}

\begin{proof}
If $\psi \in \iSigmaTerms$, then the result holds immediately.

Otherwise, the result about $\mindfn$ stems immediately from \reflem{ppterm-seq-mind}.

We prove the result about convergence.
Assume that $\ppsteps{\psi}$ is a limit ordinal, then $\psi = \icomp \psi_i$; \confer\ Lem.~\ref{rsl:steps-ordinal-coherence} and Lem.~\ref{rsl:ptinfC-iff-limit}.
Assume \reda\ convergent, consider some $k < \omega$, and $\alpha$ such that $\sdepth{\redel{\reda}{\beta}} > k$ if $\beta > \alpha$.
\refLem{ordinal-lt-infAdd-then-unique-representation} implies that $\alpha = \sum_{i < n} \ppsteps{\psi_i} + \gamma$ and $\gamma < \ppsteps{\psi_n}$ for some $n$; so that $\alpha < \sum_{i \leq n} \ppsteps{\psi_i}$.
Consider $j > n$, and $\gamma < \ppsteps{\psi_j}$. Observe $\redel{\psi_j}{\gamma} = \redel{\psi}{\beta}$ where $\beta = \sum_{i < j} \ppsteps{\psi_i} + \gamma$, so that $\beta \geq \sum_{i \leq n} \ppsteps{\psi_i} > \alpha$. Therefore $\mind{\redel{\psi_j}{\gamma}} = \mind{\redel{\psi}{\beta}} = \sdepth{\redel{\reda}{\beta}} > k$.
Hence \reflem{ppterm-seq-mind} implies that $\mind{\psi_j} > k$. Consequently, $\psi$ is convergent.

Conversely, assume $\psi$ convergent, let $k < \omega$, consider $n < \omega$ such that $\mind{\psi_j} > k$ if $j > n$. 
Let $\alpha \eqdef \sum_{i \leq n} \ppsteps{\psi_i}$, and take $\beta$ such that $ \alpha < \beta < \redln{\reda}$. 
Then Lem.~\ref{rsl:ordinal-lt-infAdd-then-unique-representation} implies $\beta = \sum_{i < j} \ppsteps{\psi_i} + \gamma$ and $\gamma < \ppsteps{\psi_j}$, moreover, $\beta > \alpha$ implies $j > n$.
Hence $\sdepth{\redel{\reda}{\beta}} = \mind{\redel{\psi_j}{\gamma}} > k$ by \reflem{ppterm-seq-mind}. Consequently, the requirement about depths in the characterisation of convergent \redseqs, \ie\ condition~(\ref{it:dfn-sred-depth}) in Dfn.~\ref{dfn:sred}, holds for \reda.
To prove the existence of $\lim_{\alpha \to \redln{\reda}} tgt(\redel{\reda}{\alpha})$, \ie\ condition~(\ref{it:dfn-sred-limit-existence}) in Dfn.~\ref{dfn:sred}, it suffices to observe that Lem.~\ref{rsl:mind-big-then-tdist-little}:(\ref{it:convergent-then-has-tgt}) implies that $tgt(\psi)$ is defined, and in turn Lem.~\ref{rsl:ppterm-tgt-limit} implies the desired limit to equal $tgt(\psi)$. Hence $\reda$ is convergent.

If $\ppsteps{\psi}$ is a successor ordinal, then assuming $\reda$ is convergent, a straightforward induction on $\psi$ suffices to prove that $\psi$ is convergent as well; observe that Lem.~\ref{rsl:steps-ordinal-coherence} and Lem~\ref{rsl:ptinfC-iff-limit} imply that only one-step and binary concatenation must be considered. For the other direction, it is enough to observe that $\redln{\reda}$ being a successor ordinal implies immediately convergence of $\reda$.

Finally, the result about targets stems immediately from \reflem{ppterm-tgt-limit} and \reflem{ppterm-tgt-successor}.
\end{proof}

\begin{proposition}
\label{rsl:denotation-existence}
Let $\reda$ be a \redseq\ having a countable length. Then there exists a \pnpterm\ $\psi$ such that $\psi$ denotes $\reda$.
\end{proposition}

\begin{proof}
We proceed by induction on $\redln{\reda}$.

If $\redln{\reda} = 0$, \ie\ $\reda = \redid{t}$, then it suffices to take $\psi \eqdef t$.

Assume that $\redln{\reda} = 1$. Let us say $\redel{\reda}{0} = \langle t, p, \mu \rangle$ where $\mu: l \to h$, implying that $\subtat{t}{p} = l[t_1, \ldots, t_m]$.
Take $\psi \eqdef \repl{t}{\mu(t_1, \ldots, t_m)}{p}$.
It is immediate to verify that $\psi$ is a \ppterm\ verifying $\ppsteps{\psi} = 1$. Moreover, a simple analysis yields $src(\psi) = src(\redel{\reda}{0}) = src(\reda) = t$.
Furthermore, $\psi(p) = \mu$, and $tgt(\psi) = tgt(\redel{\reda}{0}) = \repl{t}{h[t_1, \ldots, t_m]}{p}$; therefore $\redel{\psi}{0} = \psi$ denotes $\redel{\reda}{0}$. Hence $\psi$ denotes $\reda$.

Assume $\redln{\reda} = \alpha+1$ and $\alpha > 0$. 
In this case, applying twice \ih\ yields the existence of $\psi_1$, $\psi_2$ such that $\psi_1$ denotes $\redupto{\reda}{\alpha}$ and $\psi_2$ denotes $\redsublt{\reda}{\alpha}{\alpha+1}$.
Then a straightforward analysis allows to obtain that $\psi \eqdef \psi_1 \comp \psi_2$ denotes $\reda$.

Assume $\alpha \eqdef \redln{\reda}$ is a limit ordinal; recall that $\alpha$ is countable. Then Prop.~\ref{rsl:cofinality-omega} implies $\alpha = \sum_{i < \omega} \alpha_i$ where $\alpha_i < \alpha$ for all $i < \omega$.
Therefore, for any $n < \omega$, \ih\ can be applied to obtain some $\psi_n$ denoting $\redsublt{\reda}{\sum_{i < n} \alpha_i}{\sum_{i \leq n} \alpha_i}$.
We take $\psi \eqdef \icomp \psi_i$.

Let $n < \omega$. It is easy to verify that $\redsublt{\reda}{\sum_{i < n} \alpha_i}{\sum_{i \leq n} \alpha_i}$ is convergent, then Lem.~\ref{rsl:redseq-denotation-implications} implies 
$tgt(\psi_n) 
    = tgt(\redsublt{\reda}{\sum_{i < n} \alpha_i}{\sum_{i \leq n} \alpha_i})
		= src(\redsublt{\reda}{\sum_{i \leq n} \alpha_i}{\sum_{i \leq n+1} \alpha_i}
		= src(\psi_{n+1})$; \confer\ conditions about sources and targets in Dfn.~\ref{dfn:sred}.
Hence $\psi$ is a well-formed proof term.
Recalling that $\redln{\redsublt{\reda}{\sum_{i < n} \alpha_i}{\sum_{i \leq n} \alpha_i}} = \alpha_n$, it is straightforward to obtain $\ppsteps{\psi} = \redln{\reda} = \alpha$.
Moreover, $src(\psi) = src(\psi_0) = src(\redupto{\reda}{\alpha_0}) = src(\reda)$, recall that $\psi_0$ denotes $\redupto{\reda}{\alpha_0}$.
Let $\beta < \alpha$. Then Lem.~\ref{rsl:ordinal-lt-infAdd-then-unique-representation} implies the existence of unique $k$ and $\gamma$ such that $\beta = \sum_{i < k} \alpha_i + \gamma$ and $\gamma < \alpha_k$.
Therefore $\redel{\psi}{\beta} = \redel{\psi_k}{\gamma}$ and $\redel{\reda}{\beta} = \redel{\redsublt{\reda}{\sum_{i<k} \alpha_i}{\sum_{i \leq k} \alpha_i}}{\gamma}$, \confer\ Dfn.~\ref{dfn:ppterm-component} and Dfn.~\ref{dfn:redseq-section}.
Hence $\psi_k$ denoting $\redsublt{\reda}{\sum_{i<k} \alpha_i}{\sum_{i \leq k} \alpha_i}$ implies that $\redel{\psi}{\beta}$ denotes $\redel{\reda}{\beta}$.
Consequently, we conclude.
\end{proof}

\subsection{Uniqueness of denotation}
\label{sec:pterm-denotation-uniqueness}

In this section we will prove the claim we made at the beginning of Section~\ref{sec:pterm-denotation}: \emph{rebracketing equivalence}, which is the result of restricting the \peqence\ relation introduced in Section~\ref{sec:peqence} by allowing only associativity instances as basic equations, is an adequate syntactic counterpart of the relation of ``denoting the same \redseq'', \ie\ \emph{denotational equivalence},  between \ppterms.

In the following we will give formal definitions for the concepts of denotational and rebracketing equivalence, and subsequently prove that the defined relations coincide.

\begin{definition}
\label{dfn:deneq}
Let $\psi$, $\phi$ be \pnpterms.
We say that $\psi$ and $\phi$ are \emph{denotationally equivalent}, notation $\psi \deneq \phi$, iff 
either $\ppsteps{\psi} = \ppsteps{\phi} = 0$ and $\psi = \phi$, 
or $\ppsteps{\psi} = \ppsteps{\phi} > 0$ and $\redel{\psi}{\alpha} = \redel{\phi}{\alpha}$ for all $\alpha < \ppsteps{\psi}$.
\end{definition}

\begin{definition}
\label{dfn:layer-breqe}
\label{dfn:layer-breq}
Let $\alpha$ be a countable ordinal.

We define the \emph{$\alpha$-th level of base rebracketing equivalence} relation, notation $\layerbreqe{\alpha}$, on the set of \pnpterms, as follows.
Given $\psi$ and $\phi$ \pnpterms, $\psi \layerbreqe{\alpha} \phi$ iff the equation $\psi \layerpeqx{\alpha} \phi$ can be obtained by means of the equational logic system whose basic equations are the instance \peqassoc\ described in Dfn.~\ref{dfn:layer-peqe}, and whose equational rules are \eqlrefl, \eqleqn, \eqlsymm, \eqltrans, \eqlcomp\ and \eqlinfcomp, described also in Dfn.~\ref{dfn:layer-peqe}.

We also define the \emph{$\alpha$-th level of rebracketing equivalence} relation, notation $\layerbreq{\alpha}$, on the set of \pnpterms, analogously, the only difference being that a rule is added, namely the version of the \eqllim\ rule which results from changing, in the premises, the references to the $\layerpeqe{\alpha_k}$ and $\layerpeqe{\beta_k}$ relations, to $\layerbreqe{\alpha_k}$ and $\layerbreqe{\beta_n}$ respectively.
\end{definition}

\begin{definition}
\label{dfn:breqe}
\label{dfn:breq}
Let $\psi$, $\phi$ be \pnpterms. We say that $\psi$ and $\phi$ are (base) rebracketing equivalent, notation ($\psi \breqe \phi$) $\psi \breq \phi$, iff ($\psi \layerbreqe{\alpha} \phi$) $\psi \layerbreq{\alpha} \phi$ for some $\alpha < \omega_1$.
\end{definition}

Observe that all the following inclusions hold where $\alpha$ is any countable ordinal: $\layerbreqe{\alpha} \subseteq \layerbreq{\alpha}$, $\layerbreqe{\alpha} \subseteq \layerpeqe{\alpha}$, $\layerbreq{\alpha} \subseteq \layerpeq{\alpha}$, and consequently $\breqe \subseteq \breq$, $\breqe \subseteq \peqe$ and $\breq \subseteq \peq$.
Therefore, several results stated for \peqence\ hold also for rebracketing equivalence. Particularly, properties proved for the $\peqe$ relation also apply to $\breqe$.

\begin{lemma}
\label{rsl:layerpterm-ppsteps}
Let $\psi$ a \ppterm, and $\alpha$ such that $\psi \in \layerpterm{\alpha}$. Then $\exists n < \omega$ such that $\alpha = \ppsteps{\psi} + n$.
Moreover, if $\alpha$ is a limit ordinal, then $n = 0$, \ie\ $\alpha = \ppsteps{\psi}$.
\end{lemma}

\begin{proof}
We proceed by induction on $\alpha$.
If $\alpha = 1$ then $\psi$ is a one-step, and then $\ppsteps{\psi} = 1 = \alpha$.

Assume $\alpha$ is a successor ordinal and $\alpha > 1$.
In this case, Lem.~\ref{rsl:ptinfC-iff-limit} and Lem.~\ref{rsl:ptmstep-iff-one} imply that $\psi = \psi_1 \comp \psi_2$, $\psi_i \in \layerpterm{\alpha_i}$ for $i = 1.2$, $\alpha_2$ is successor, and $\alpha = \alpha_1 + \alpha_2 + 1$.
\Ih\ implies $\alpha_1 = \ppsteps{\psi_1} + n_1$ and $\alpha_2 = \ppsteps{\psi_2} + n_2$.
If $\ppsteps{\psi_2} < \omega$, then $\alpha = \ppsteps{\psi} + n_1 + n_2 + 1$, otherwise $\alpha = \ppsteps{\psi} + n_2 + 1$. In either case the conclusion holds, thus we conclude.

Assume that $\alpha$ is a limit ordinal, so that Lem.~~\ref{rsl:ptinfC-iff-limit} implies $\psi = \icomp \psi_i$ and $\alpha = \sum_{i < \omega} \alpha_i$ where $\psi_i \in \layerpterm{\alpha_i}$ for all $i < \omega$.
Observe $\alpha_i < \alpha$ for all $i$.
Then we can apply \ih\ on each $i$ obtaining $\alpha_i = \ppsteps{\psi_i} + n_i$, so that proving $\sum_{i < \omega} \ppsteps{\psi_i}  + n_i = \sum_{i < \omega} \ppsteps{\psi_i}$ suffices to conclude.

Let $k < \omega$. Observe $\sum_{i < k} \ppsteps{\psi_i}  + n_i \leq \sum_{i < k} \ppsteps{\psi_i} + \sum_{i < k} n_i < \sum_{i < k} \ppsteps{\psi_i} + \omega$.
On the other hand, $\sum_{i < \omega} \ppsteps{\psi_i} = \sum_{i < k} \ppsteps{\psi_i} + \sum_{i < \omega} \ppsteps{\psi_{k+i}} \geq \sum_{i < k} \ppsteps{\psi_i} + \omega$. Then $\sum_{i < k} \ppsteps{\psi_i} + n_i < \sum_{i < \omega} \ppsteps{\psi_i}$. Consequently $\sum_{i < \omega} \ppsteps{\psi_i} + n_i \leq \sum_{i < \omega} \ppsteps{\psi_i}$.
We conclude by observing that it is straightforward to obtain $\sum_{i < \omega} \ppsteps{\psi_i} \leq \sum_{i < \omega} \ppsteps{\psi_i} + n_i$.
\end{proof}

\begin{lemma}
\label{rsl:ptinfC-iff-ppsteps-limit}
Let $\psi$ be a \ppterm. Then $\ppsteps{\psi}$ is a limit ordinal iff $\psi$ is an infinite concatenation.
\end{lemma}

\begin{proof}
We proceed by induction on $\alpha$ where $\psi \in \layerpterm{\alpha}$; \confer\ Dfn.~\ref{dfn:layer-peqe}.
If $\psi$ is a one-step, then we conclude immediately.
If $\psi = \psi_1 \comp \psi_2$ and it is not an infinite concatenation, then $\psi_2$ is neither. Therefore we can apply \ih\ on $\psi_2$ obtaining that $\ppsteps{\psi_2}$ is a successor ordinal. We conclude by recalling that $\ppsteps{\psi} = \ppsteps{\psi_1} + \ppsteps{\psi_2}$.
Finally, if $\psi$ is an infinite concatenation, then Lem.~\ref{rsl:ptinfC-iff-limit} implies that $\psi \in \layerpterm{\alpha}$ where $\alpha$ is a limit ordinal. In turn, Lem.~\ref{rsl:layerpterm-ppsteps} implies that $\ppsteps{\psi} = \alpha$.
\end{proof}

\begin{lemma}
\label{rsl:ppterm-binC-partition}
Let $\psi$ be a \ppterm, $\alpha$ an ordinal verifying $0 < \alpha < \ppsteps{\psi}$, and $\beta$ such that $\psi \in \layerpterm{\beta}$.
Then there exist $\phi$, $\chi$ such that $\psi \breqe \phi \comp \chi$ and $\ppsteps{\phi} = \alpha$.
Moreover, if $\phi \in \layerpterm{\gamma}$ and $\chi \in \layerpterm{\delta}$, then $\gamma < \beta$ and $\delta \leq \beta$.
\end{lemma}

\begin{proof}
We proceed by induction on $\psi$.

If $\psi \in \iSigmaTerms$ or $\psi$ is a one-step, then no $\alpha$ verifies the hypotheses.

\smallskip
Assume $\psi = \psi_1 \comp \psi_2$, so that $\beta = \beta_1 + \beta_2 + 1$ where $\psi_i \in \layerpterm{\beta_i}$ for $i = 1,2$.

\begin{itemize}
\item 
If $\ppsteps{\psi_1} < \alpha$, so that $\alpha = \ppsteps{\psi_1} + \alpha'$, then \ih\ on $\psi_2$ yields the existence of $\phi_2, \chi_2$ satisfying $\psi_2 \breqe \phi_2 \comp \chi_2$, $\ppsteps{\phi_2} = \alpha'$, $\gamma_2 < \beta_2$ and $\delta \leq \beta_2$, where $\phi_2 \in \layerpterm{\gamma_2}$ and $\chi_2 \in \layerpterm{\delta}$.

Therefore, $\psi \breqe \psi_1 \comp (\phi_2 \comp \chi_2) \breqe (\psi_1 \comp \phi_2) \comp \chi_2$ and $\ppsteps{\psi_1 \comp \phi_2} = \ppsteps{\psi_1} + \alpha' = \alpha$. Moreover, $\psi_1 \comp \phi_2 \in \layerpterm{\gamma}$ where $\gamma = \beta_1 + \gamma_2 + 1 < \beta_1 + \beta_2 + 1 = \beta$, and $\delta \leq \beta_2 < \beta$.

\item
If $\ppsteps{\psi_1} = \alpha$ then the result holds trivally.

\item 
If $\ppsteps{\psi_1} > \alpha$, then \ih\ on $\psi_1$ yields $\psi_1 \breqe \phi_1 \comp \chi_1$, $\ppsteps{\phi_1} = \alpha$, $\gamma < \beta_1$ and $\delta_1 \leq \beta_1$, where $\phi_1 \in \layerpterm{\gamma}$ and $\chi_1 \in \layerpterm{\delta_1}$.

Therefore $\psi \breqe (\phi_1 \comp \chi_1) \comp \psi_2 \breqe \phi_1 \comp (\chi_1 \comp \psi_2)$. 
Moreover, $\gamma < \beta_1 < \beta$, and $\chi_1 \comp \psi_2 \in \layerpterm{\delta}$ where $\delta = \delta_1 + \beta_2 + 1 \leq \beta_1 + \beta_2 + 1 = \beta$.
\end{itemize}

\smallskip
Assume $\psi = \icomp \psi_i$, so that $\ppsteps{\psi} = \sum_{i < \omega} \ppsteps{\psi_i}$.
In this case, 
Lem~\ref{rsl:ptinfC-iff-limit} and Lem~\ref{rsl:layerpterm-ppsteps} imply that $\beta$ is a limit ordinal, and therefore $\beta = \ppsteps{\psi}$. 
Moreover, Lem~\ref{rsl:ordinal-lt-infAdd-then-unique-representation} implies $\alpha = \sum_{i < n} \ppsteps{\psi_i} + \alpha'$ where $\alpha' < \ppsteps{\psi_n}$, for some $n$ and $\alpha'$. 
\Ih\ on $\psi_n$ yields $\psi_n \breqe \phi_n \comp \chi_n$ such that $\ppsteps{\phi_n} = \alpha'$; observe that $\ppsteps{\chi_n} \leq \ppsteps{\psi_n}$.
Therefore \\
$\begin{array}{rcl}
\psi & \breqe & ((\psi_0 \comp \ldots \comp \psi_{n-1}) \comp \psi_n) \comp \icomp \psi_{n+1+i} \\
& \breqe & ((\psi_0 \comp \ldots \comp \psi_{n-1}) \comp (\phi_n \comp \chi_n)) \comp \icomp \psi_{n+1+i} \\
& \breqe & ((\psi_0 \comp \ldots \comp \psi_{n-1} \comp \phi_n) \comp \chi_n) \comp \icomp \psi_{n+1+i} \\
& \breqe & (\psi_0 \comp \ldots \comp \psi_{n-1} \comp \phi_n) \comp (\chi_n \comp \icomp \psi_{n+1+i}) 
\end{array}$ \\
where $\ppsteps{\psi_0 \comp \ldots \comp \psi_{n-1} \comp \phi_n} = \sum_{i < n} \ppsteps{\psi_i} + \alpha' = \alpha$.

Moreover, if $\psi_0 \comp \ldots \comp \psi_{n-1} \comp \phi_n \in \layerpterm{\gamma}$, then Lem.~\ref{rsl:layerpterm-ppsteps} implies the existence of some $k < \omega$ such that
$\gamma 
    = \ppsteps{\psi_0} + \ldots + \ppsteps{\psi_{n-1}} + \alpha' + k
		< \ppsteps{\psi_0} + \ldots + \ppsteps{\psi_{n-1}} + \ppsteps{\psi_n} + \omega
		\leq \ppsteps{\psi} = \beta$.
On the other hand, notice that $\chi_n \comp \icomp \psi_{n+1+i}$ is an infinitary concatenation, so that $\chi_n \comp \icomp \psi_{n+1+i} \in \layerpterm{\delta}$ implies $\delta$ to be a limit ordinal; \confer\ Lem.~\ref{rsl:ptinfC-iff-limit}.
%
Therefore, recalling that $\ppsteps{\chi_n} \leq \ppsteps{\psi_n}$, Lem.~\ref{rsl:layerpterm-ppsteps} yields
$\delta 
    = \ppsteps{\chi_n} + \sum_{i < \omega} \ppsteps{\psi_{n+1+i}} 
		\leq \sum_{i < \omega} \ppsteps{\psi_{n+i}} 
		\leq \ppsteps{\psi} = \beta$.
\end{proof}

\begin{lemma}
\label{rsl:deneq-then-same-tgt}
Let $\psi \deneq \phi$, such that both are convergent.
Then $tgt(\psi) = tgt(\phi)$.
\end{lemma}

\begin{proof}
Easy, \confer\ Lem.~\ref{rsl:ppterm-tgt-successor} and Lem~\ref{rsl:ppterm-tgt-limit}.
\end{proof}

\begin{lemma}
\label{rsl:deneq-binC-right}
Let $\psi \comp \phi \deneq \psi' \comp \phi'$ and $\psi \deneq \psi'$. 
Then $\phi \deneq \phi'$.
\end{lemma}
\begin{proof}
Observe that definition of \ppterms\ implies that $\ppsteps{\phi} > 0$ and $\ppsteps{\phi'} > 0$.
Given $\ppsteps{\psi \comp \phi} = \ppsteps{\psi' \comp \phi'}$ and $\ppsteps{\psi} = \ppsteps{\psi'}$, properties of ordinals yield $\ppsteps{\phi} = \ppsteps{\phi'}$.
We conclude by observing that for any suitable $\alpha$, 
$\redel{\phi}{\alpha} = 
		\redel{(\psi \comp \phi)}{\ppsteps{\psi} + \alpha} = 
		\redel{(\psi' \comp \phi')}{\ppsteps{\psi'} + \alpha} =
		\redel{\phi'}{\alpha}$.
\end{proof}

\begin{proposition}
\label{rsl:breq-then-deneq}
Let $\psi$, $\phi$ be \pnpterms\ such that $\psi \breq \phi$.
Then $\psi \deneq \phi$.
\end{proposition}

\begin{proof}
We proceed by induction on $\alpha$ where $\psi \layerbreq{\alpha} \phi$.
We analyse the rule used in the last step of the rebracketing equivalence derivation.

\smallskip
For the rules \eqlrefl, \eqlsymm\ and \eqltrans, the result holds immediately.

\smallskip
Assume that the last used rule in the derivation is \eqleqn, so that 
$\psi = (\psi_1 \comp \psi_2) \comp \psi_3$ and $\phi = \psi_1 \comp (\psi_2 \comp \psi_3)$.
In this case we can obtain $\ppsteps{\psi} = \ppsteps{\phi} > 0$ immediately.
Let $\gamma < \ppsteps{\psi}$.
If $\gamma < \ppsteps{\psi_1}$, then $\redel{\psi}{\gamma} = \redel{(\psi_1 \comp \psi_2)}{\gamma} = \redel{\psi_1}{\gamma} = \redel{\phi}{\gamma}$. The other cases, \ie\ $\ppsteps{\psi_1} \leq \gamma < \ppsteps{\psi_1} + \ppsteps{\psi_2}$ and $\ppsteps{\psi_1} + \ppsteps{\psi_2} \leq \gamma$, admit analogous arguments.

\smallskip
Assume that the last used rule is \eqlinfcomp, so that $\psi = \icomp \psi_i$, $\phi = \icomp \phi_i$, and $\psi_n \layerbreq{\beta_n} \phi_n$ where $\beta_n < \alpha$, for all $n < \omega$.
Then \ih\ on each $\beta_n$ implies $\psi_n \deneq \phi_n$. Therefore we obtain $\ppsteps{\psi} = \ppsteps{\phi} > 0$ immediately.
To conclude it is enough to observe, for any $\gamma < \ppsteps{\psi}$, that Lem.~\ref{rsl:ordinal-lt-infAdd-then-unique-representation} implies $\gamma = \sum_{i < n} \ppsteps{\psi_i} + \gamma_0$ where $\gamma_0 < \ppsteps{\psi_n}$, then (given \ih\ on each $\psi_i \layerbreq{\beta_i} \phi_i$) $\redel{\psi}{\gamma} = \redel{\psi_n}{\gamma_0} = \redel{\phi_n}{\gamma_0} = \redel{\phi}{\gamma}$. 

\smallskip
If the last used rule is \eqlcomp, then a similar argument applies.

\medskip
Assume that the rule used in the last derivation step is \eqllim.
Assume for contradiction $\ppsteps{\phi} > \ppsteps{\psi}$, so that the step $\redel{\phi}{\ppsteps{\psi}}$ exists. 
Consider $k \eqdef max(\mind{\redel{\phi}{0}}, \mind{\redel{\phi}{\ppsteps{\psi}}})$. Then there exist $\chi_k$, $\phi'_k$, $\psi'_k$ verifying
$\phi \layerbreqe{\alpha_k} \chi_k \comp \phi'_k$, $\psi \layerbreqe{\beta_k} \chi_k \comp \psi'_k$, $\mind{\phi'_k} > k \geq \mind{\redel{\phi}{\ppsteps{\psi}}}$, $\mind{\psi'_k} > k$, $\alpha > \alpha_k$, and $\alpha > \beta_k$.
Recalling that $\layerbreqe{\gamma} \,\subseteq\, \layerbreq{\gamma}$ for any $\gamma$, we can apply \ih\ to $\alpha_k$ obtaining $\phi \deneq \chi_k \comp \phi'_k$, so that $\redel{\phi}{\ppsteps{\psi}} = \redel{(\chi_k \comp \phi'_k)}{\ppsteps{\psi}}$.
Therefore, assuming $\ppsteps{\psi} = \ppsteps{\chi_k} + \gamma$ would imply $\redel{\phi'_k}{\gamma} = \redel{\phi}{\ppsteps{\psi}}$ contradicting $\mind{\phi'_k} > \mind{\redel{\phi}{\ppsteps{\psi}}}$; \confer\ Lem.~\ref{rsl:ppterm-seq-mind}.
Then $\ppsteps{\psi} < \ppsteps{\chi_k}$.
On the other hand, \ih\ can be applied also to $\beta_k$, yielding $\psi \deneq \chi_k \comp \psi'_k$, and therefore $\ppsteps{\psi} \geq \ppsteps{\chi_k}$, \ie\ a contradiction.
Consequently $\ppsteps{\phi} \leq \ppsteps{\psi}$.
A similar argument yields $\ppsteps{\psi} \leq \ppsteps{\phi}$. Thus $\ppsteps{\psi} = \ppsteps{\phi}$.

Let $\gamma < \ppsteps{\psi}$. Then there exists $\chi$, $\psi'$, $\phi'$ such that
$\psi \layerbreqe{\alpha_0} \chi \comp \psi'$, $\phi \layerbreqe{\beta_0} \chi \comp \phi'$, $\mind{\psi'} > \mind{\redel{\psi}{\gamma}}$, $\mind{\phi'} > \mind{\redel{\psi}{\gamma}}$, $\alpha_0 < \alpha$ and $\beta_0 < \alpha$.
Then \ih\ on $\alpha_0$ and $\beta_0$ yields $\psi \deneq \chi \comp \psi'$ and $\phi \deneq \chi \comp \phi'$, so that $\redel{\psi}{\gamma} = \redel{(\chi \comp \psi')}{\gamma}$ and $\redel{\phi}{\gamma} = \redel{(\chi \comp \phi')}{\gamma}$.
Observing that $\gamma = \ppsteps{\chi} + \gamma_0$ would imply $\redel{\psi}{\gamma} = \redel{\psi'}{\gamma_0}$, and then $\mind{\psi'} \leq \mind{\redel{\psi}{\gamma}}$ (\confer\ Lem.~\ref{rsl:ppterm-seq-mind}) thus producing a contradiction, we obtain $\gamma < \ppsteps{\chi}$. Then $\redel{\psi}{\gamma} = \redel{\chi}{\gamma}$, and also $\redel{\phi}{\gamma} = \redel{\chi}{\gamma}$. Hence $\redel{\psi}{\gamma} = \redel{\phi}{\gamma}$.
\end{proof}

\begin{proposition}
\label{rsl:deneq-then-breq}
Let $\psi$, $\phi$ such that $\psi \deneq \phi$.
Then $\psi \breq \phi$.
\end{proposition}

\begin{proof}
We proceed by induction on $\pair{\alpha}{\beta}$ such that $\psi \in \layerpterm{\alpha}$ and $\phi \in \layerpterm{\beta}$.

If $\psi \in \iSigmaTerms$, so that $\ppsteps{\psi} = 0$, then $\psi \deneq \phi$ implies $\psi = \phi$, hence we conclude immediately.

If $\psi$ is a one-step, so that $\ppsteps{\psi} = 1$, then $\psi \deneq \phi$ implies $\psi = \redel{\psi}{0} = \redel{\phi}{0} = \phi$.

\smallskip
Assume $\psi = \psi_1 \comp \psi_2$ and that it is not an infinite concatenation. 
In this case, $\ppsteps{\psi} = \ppsteps{\phi} > 1$ is a successor ordinal, so that $\phi = \phi_1 \comp \phi_2$ and it is neither an infinite concatenation; \confer\ Lem.~\ref{rsl:ptinfC-iff-ppsteps-limit}.
Observe that $\alpha = \alpha_1 + \alpha_2 + 1$ and $\beta = \beta_1 + \beta_2 + 1$, where $\psi_i \in \layerpterm{\alpha_i}$ and $\phi_i \in \layerpterm{\beta_i}$ for $i = 1,2$.
We analyse the different cases arising from the comparison between $\ppsteps{\psi_1}$ and $\ppsteps{\phi_1}$.
\begin{itemize}
\item 
Assume $\ppsteps{\psi_1} < \ppsteps{\phi_1}$.
In this case we apply Lem.~\ref{rsl:ppterm-binC-partition}, obtaining that $\phi_1 \breq \chi_1 \comp \chi_2$ and $\ppsteps{\chi_1} = \ppsteps{\psi_1}$ for some \ppterms\ $\chi_1 \in \layerpterm{\gamma_1}$ and $\chi_2 \in \layerpterm{\gamma_2}$, and moreover, that $\gamma_1 < \beta_1$ and $\gamma_2 \leq \beta_1$. 

Therefore $\phi \ \breq \  (\chi_1 \comp \chi_2) \comp \phi_2 \ \breq \  \chi_1 \comp (\chi_2 \comp \phi_2)$, and hence Prop.~\ref{rsl:breq-then-deneq} and hypotheses yield $\psi = \psi_1 \comp \psi_2 \ \deneq \  \chi_1 \comp (\chi_2 \comp \phi_2) \deneq \phi$.
Observe that for any $\beta < \ppsteps{\psi_1}$, $\redel{\psi_1}{\beta} = \redel{\psi}{\beta} = \redel{\phi}{\beta} = \redel{(\chi_1 \comp (\chi_2 \comp \phi_2))}{\beta} = \redel{\chi_1}{\beta}$; consequently, $\psi_1 \deneq \chi_1$.
In turn, Lem.~\ref{rsl:deneq-binC-right} yields $\psi_2 \deneq \chi_2 \comp \phi_2$.

Observing that $\alpha_i < \alpha$ for $i = 1,2$ suffices to enable the application of \ih\ to both $\psi_1 \deneq \chi_1$ and $\psi_2 \deneq \chi_2 \comp \phi_2$. Therefore, we conclude by \eqlcomp, \eqlsymm\ and \eqltrans.

\item
Assume $\ppsteps{\psi_1} > \ppsteps{\phi_1}$. In this case, an analysis similar to that of the previous case yields $\psi_1 \breqe \chi_1 \comp \chi_2$ such that $\ppsteps{\chi_1} = \ppsteps{\phi_1}$, $\gamma_1 < \alpha_1$ and $\gamma_2 \leq \alpha_1$ where $\chi_i \in \layerpterm{\gamma_i}$ for $i = 1,2$; therefore $\chi_1 \comp (\chi_2 \comp \psi_2) \deneq \psi \deneq \phi = \phi_1 \comp \phi_2$; and consequently $\chi_1 \deneq \phi_1$ and $\chi_2 \comp \psi_2 \deneq \phi_2$.

Observe $\gamma_1 < \alpha_1 < \alpha$.
On the other hand, $\chi_2 \comp \psi_2 \in \layerpterm{\delta}$ where $\delta = \gamma_2 + \alpha_2 + 1 \leq \alpha_1 + \alpha_2 + 1 = \alpha$, and $\beta_2 < \beta$.
Therefore, \ih\ can be applied to both $\chi_1 \deneq \phi_1$ and $\chi_2 \comp \psi_2 \deneq \phi_2$, so that we conclude as in the previous case.

\item
Assume $\ppsteps{\psi_1} = \ppsteps{\phi_1}$. Then a simple analysis of the components of $\psi_1$ and $\phi_1$ yields $\psi_1 \deneq \phi_1$. In turn, this assertion allows to apply Lem.~\ref{rsl:deneq-binC-right} to obtain $\psi_2 \deneq \phi_2$.
Applying \ih\ to both $\psi_i$ we obtain $\psi_1 \breq \phi_1$ and $\psi_2 \breq \phi_2$. Hence we conclude by \eqlcomp.
\end{itemize}

\medskip
Assume $\psi = \icomp \psi_i$.
In this case, a simple argument based on Lem.~\ref{rsl:ptinfC-iff-ppsteps-limit} yields $\phi = \icomp \phi_i$.

As the verification for this case involves a great number of technical details, we describe the idea first.
We define a \ppterm\ $\chi = \icomp \chi_i$ enjoying the following properties: $\psi \breq \chi$, and $\chi_n \deneq \phi_n$ for all $n < \omega$. The \eqllim\ rule is used in the last step of the derivation $\psi \breq \chi$, verifying that the corresponding premises are valid \wrt\ $\breqe$. 
In turn, Lem.~\ref{rsl:layerpterm-ppsteps} allows to apply \ih\ on any $\chi_n$, since $\chi \in \layerpterm{\delta}$ implies $\delta = \ppsteps{\chi} = \ppsteps{\psi} = \alpha$ (\confer\ Prop.~\ref{rsl:breq-then-deneq}).
Therefore we obtain $\chi_n \breq \phi_n$ for all $n < \omega$, implying $\chi \breq \phi$. Then \eqltrans\ yields $\psi \breq \phi$.
A very schematic derivation tree follows:

\smallskip
$
\prooftree
  \[
		\ldots
		\ \ 
		\begin{array}{l}
		\psi \breqe \xi_k \comp \psi' \\	
		\chi \breqe \xi_k \comp \chi'
		\end{array} 
		\ \ 
		\ldots
	\justifies
		\psi \breq \chi
	\using
		\eqllim
	\]
	\quad
	\[
	  \ldots \ \  
		\[
			B_n
	  \justifies
		  \chi_n \breq \phi_n
		\]
		\ \ \ldots
	\justifies
		\chi \breq \phi
	\using
		\eqlinfcomp
	\]
\justifies
	\psi \breq \phi
\using
	\eqltrans
\endprooftree
$

\smallskip\noindent
where we can observe the soundness of the derivation, even if \eqllim\ is applied in some of the $B_n$ derivations.

\medskip
We define $\chi_k$, by induction on $k$, for all $k < \omega$.
We observe that $\sum_{i < k} \ppsteps{\phi_i} < \ppsteps{\phi} = \ppsteps{\psi}$.
Then we define, along with $\chi_k$, two values $p_k$ and $\beta_k$ as follows:
$p_0 \eqdef 0$, $\beta_0 \eqdef 0$, and if $k > 0$, then $p_k$ and $\beta_k$ are the unique (\confer\ Lem.~\ref{rsl:ordinal-lt-infAdd-then-unique-representation}) values verifying $\sum_{i < k} \ppsteps{\phi_i} = \sum_{i < p_k} \ppsteps{\psi_i} + \beta_k$ and $\beta_k < \ppsteps{\psi_{p_k}}$. We also define $p' \eqdef p_{k+1} - 1$.
Simultaneously with the definiton of $\chi_k$, we will verify the following auxiliary assertion: \\[2pt]
\begin{tabular}{@{$\ \ \bullet \ \ $}p{.93\textwidth}}
$\chi_0 \comp \ldots \comp \chi_k \breqe \psi_0 \comp \ldots \comp \psi_{p'}$ if $\beta_{k+1} = 0$; and \\
there exist
$\chi', \xi$ such that 
$\psi_{p_{k+1}} \breqe \chi' \comp \xi$, $\ppsteps{\chi'} = \beta_{k+1}$ and $\chi_0 \comp \ldots \comp \chi_k \breqe \psi_0 \comp \ldots \comp \psi_{p'} \comp \chi'$ (or $\chi_0 \comp \ldots \comp \chi_k \breqe \chi'$ if $p_{k+1} = 0$), if $\beta_{k+1} > 0$.
\end{tabular} \\[2pt]
Therefore, when defining $\chi_n$ for a given $n$, we can consider this assertion to be valid for all $n' < n$.

Let $n < \omega$. 	Several cases must be analysed to define $\chi_n$. 
\begin{itemize}
\item 
Assume that either $n = 0$, \ie\ the base case, or $n > 0$ and $\beta_n = 0$. 
  \begin{itemize}
  \item 
  Assume $p_n = p_{n+1}$, implying $\ppsteps{\phi_n} = \beta_{n+1} > 0$, so that $\ppsteps{\phi_n} < \ppsteps{\psi_{p_n}}$.
  In this case we define $\chi_n$ to be some term verifying $\psi_{p_n} \breqe \chi_n \comp \xi$ and $\ppsteps{\chi_n} = \ppsteps{\phi_n}$; \confer\ Lem.~\ref{rsl:ppterm-binC-partition}.

  \item
  Assume $p_n < p_{n+1}$ and $\beta_{n+1} = 0$, so that $\ppsteps{\phi_n} = \ppsteps{\psi_{p_n}} + \ldots + \ppsteps{\psi_{p'}}$.
  In this case we define $\chi_n \eqdef \psi_{p_n} \comp \ldots \comp \psi_{p'}$. 

  \item
  Assume $p_n < p_{n+1}$ and $\beta_{n+1} > 0$, implying $\ppsteps{\phi_n} = \ppsteps{\psi_{p_n}} + \ldots + \ppsteps{\psi_{p'}} + \beta_{n+1}$.
  We consider some $\chi', \xi$ verifying $\psi_{p_{n+1}} \breqe \chi' \comp \xi$ and $\ppsteps{\chi'} = \beta_{n+1}$; \confer\ Lem.~\ref{rsl:ppterm-binC-partition}.
  Then we define $\chi_n \eqdef \psi_{p_n} \comp \ldots \comp \psi_{p'} \comp \chi'$.
  \end{itemize}
In any case, if $n = 0$ then the auxiliary assertion holds immediately; otherwise, it suffices to apply the same assertion on $n-1$ obtaining $\chi_0 \comp \ldots \comp \chi_{n-1} \breqe \psi_0 \comp \ldots \comp \psi_{p_n-1}$, and then \eqlrefl\ and \eqlcomp.

\item
Assume $\beta_n > 0$.
In this case $n > 0$, then the auxiliary assertion on $n-1$ implies the existence of $\chi'$, $\xi$ verifying $\psi_{p_n} \breqe \chi' \comp \xi$, $\ppsteps{\chi'} = \beta_n$ and $\chi_0 \comp \ldots \comp \chi_{n-1} \breqe \psi_0 \comp \ldots \comp \psi_{p_n-1} \comp \chi'$ (or $\chi_0 \comp \ldots \comp \chi_{n-1} \breqe \chi'$ if $p_n = 0$).
  \begin{itemize}
  \item 
  Assume $p_{n+1} = p_n$, implying $\beta_{n+1} = \beta_n + \ppsteps{\phi_n} < \ppsteps{\psi_{p_n}} = \beta_n + \ppsteps{\xi}$, implying $\ppsteps{\phi_n} < \ppsteps{\xi}$.
  In this case we define $\chi_n$ bo te some term verifying $\xi \breqe \chi_n \comp \xi'$ and $\ppsteps{\chi_n} = \ppsteps{\phi_n}$; \confer\ Lem.~\ref{rsl:ppterm-binC-partition}.
  Observe $\psi_{p_n} \breqe (\chi' \comp \chi_n) \comp \xi'$, $\ppsteps{\chi' \comp \chi_n} = \beta_n + \ppsteps{\phi_n} = \beta_{n+1}$ and $\chi_0 \comp \ldots \comp \chi_n \breqe \psi_0 \comp \ldots \comp \psi_{p_n-1} \comp (\chi' \comp \chi_n)$, then the auxiliary statement holds for $n$; recall $\beta_{n+1} > \beta_n \geq 0$.

  \item
  Assume $p_{n+1} = p_n+1$ and $\beta_{n+1} = 0$, implying $\ppsteps{\psi_{p_n}} = \beta_n + \ppsteps{\phi_n}$.
  Observe $\ppsteps{\xi} = \ppsteps{\phi_n}$.
  We define $\chi_n \eqdef \xi$.
  Then $\chi_0 \comp \ldots \chi_n \breqe \psi_0 \comp \ldots \comp \psi_{p_n-1} \comp \chi' \comp \xi \breqe \psi_0 \comp \ldots \comp \psi_{p_n-1} \comp \psi_{p_n}$, then the auxiliary statement holds for $n$.

  \item
  Assume $p_{n+1} > p_n+1$ and $\beta_{n+1} = 0$, implying $\ppsteps{\phi_n} = \beta' + \ppsteps{\psi_{p_n+1}} + \ldots + \ppsteps{\psi_{p'}}$, where $\ppsteps{\psi_{p_n}} = \beta_n + \beta'$.
  Observe $\ppsteps{\xi} = \beta'$.
  We define $\chi_n \eqdef \xi \comp \psi_{p_n+1} \comp \ldots \comp \psi_{p'}$.
  We verify the auxiliary statement for $n$ similarly to the previous case.

  \item
  Assume $p_{n+1} > p_n$ and $\beta_{n+1} > 0$, implying $\ppsteps{\phi_n} = \beta' + \ppsteps{\psi_{p_n+1}} + \ldots + \ppsteps{\psi_{p'}} + \beta_{n+1}$ (or just $\beta' + \beta_{n+1}$ if $p_{n+1} = p_n+1$), where $\ppsteps{\psi_{p_n}} = \beta_n + \beta'$.
  Observe $\ppsteps{\xi} = \beta'$.
  Let $\chi'', \xi'$ such that $\psi_{p_{n+1}} \breqe \chi'' \comp \xi'$ and $\ppsteps{\chi''} = \beta_{n+1}$.
  We define $\chi_n \eqdef \xi \comp \psi_{p_n+1} \comp \ldots \comp \psi_{p'} \comp \chi''$ (or just $\xi \comp \chi''$ if $p_{n+1} = p_n+1$).
  We verify the auxiliary statement for $n$ similarly to the previous cases.  

  \end{itemize}
\end{itemize}

In turn, a simple analysis of each case yields $\ppsteps{\chi_n} = \ppsteps{\phi_n}$  for each $n < \omega$.
 
\medskip
We verify $\psi \deneq \chi$, since this assertion is used when obtaining $\psi \breq \chi$.
Given $\ppsteps{\chi_n} = \ppsteps{\phi_n}$ for all $n < \omega$, we obtain immediately $\ppsteps{\chi} = \ppsteps{\phi} = \ppsteps{\psi}$ (recall the hypothesis $\psi \deneq \phi$).
Let $\beta < \ppsteps{\chi}$, let $n$ be a natural number verifying $\beta < \sum_{i \leq n} \ppsteps{\chi_i}$ (\confer\ Lem~\ref{rsl:ordinal-lt-infAdd-then-unique-representation}).
Then $\redel{\chi}{\beta} = \redel{(\chi_0 \comp \ldots \comp \chi_n)}{\beta}$.
Observe $\chi_0 \comp \ldots \comp \chi_n \breqe \psi'$ for some $\psi'$ verifying $\psi \breqe \psi' \comp \psi''$, \confer\ the auxiliary assertion in the definition of $\chi_n$, so that $\ppsteps{\psi'} = \sum_{i \leq n} \ppsteps{\chi_i} > \beta$.
Therefore $\redel{\chi}{\beta} = \redel{(\chi_0 \comp \ldots \comp \chi_n)}{\beta} = \redel{\psi'}{\beta} = \redel{\psi}{\beta}$, \confer\ Prop.~\ref{rsl:breq-then-deneq}. Hence $\psi \deneq \chi$.

\smallskip
We verify $\psi \breq \chi$.
Let $k < \omega$, let $p$ such that $p > 0$ and $\mind{\psi_i} > k$ if $i > p$.
Let $n$ be a natural number verifying $\sum_{i \leq n} \ppsteps{\phi_i} > \sum_{i \leq p} \ppsteps{\psi_i}$.
Observe that $p_{n+1} > p$. We analyse the two possible cases of the auxiliary statement in the definition of $\chi_n$; again, $p' \eqdef p_{n+1}-1$.

If $\beta_{n+1} = 0$, then $\chi_0 \comp \ldots \comp \chi_n \breqe \psi_0 \comp \ldots \comp \psi_{p'}$; observe that also $\chi \breqe \chi_0 \comp \ldots \comp \chi_n \comp (\icomp \chi_{n+1+i})$ and $\psi \breqe \psi_0 \comp \ldots \comp \psi_{p'} \comp (\icomp \psi_{p_{n+1}+i})$.
We obtain immediately $\chi \breqe \psi_0 \comp \ldots \comp \psi_{p'} \comp (\icomp \chi_{n+1+i})$ and $\mind{\icomp \psi_{p_{n+1}+i}} > k$, since $p_{n+1} > p$.
Prop.~\ref{rsl:breq-then-deneq} yields $\chi_0 \comp \ldots \comp \chi_n \deneq \psi_0 \comp \ldots \comp \psi_{p'}$, so that Lem.~\ref{rsl:deneq-binC-right} can be applied to obtain $\icomp \chi_{n+1+i} \deneq \icomp \psi_{p_{n+1}+i}$, and therefore $\mind{\icomp \chi_{n+1+i}} > k$, \confer\ Lem.~\ref{rsl:ppterm-seq-mind}.

Otherwise, there exist some $\chi', \xi$ such that $\chi_0 \comp \ldots \comp \chi_n \breqe \psi_0 \comp \ldots \comp \psi_{p'} \comp \chi'$ and $\psi_{p_{n+1}} \breqe \chi' \comp \xi$.
By an argument analogous to that of the previous case, we obtain $\psi \breqe \psi_0 \comp \ldots \comp \psi_{p'} \comp \chi' \comp (\xi \comp \icomp \psi_{{p_{n+1}}+1+i})$, $\chi \breqe \psi_0 \comp \ldots \comp \psi_{p'} \comp \chi' \comp (\icomp \chi_{n+1+i})$, and $\mind{\xi \comp \icomp \psi_{p_{n+1}+1+i}} = \mind{\icomp \chi_{n+1+i}} > k$.

Consequently we can apply \eqllim\ to obtain $\psi \breq \chi$. Observe that the premises of the \eqllim\ application correspond to the $\breqe$ relation, so that the derivation is sound.

\medskip
The only element needed to complete the idea described earlier, and then to conclude the proof, is to obtain $\chi_n \deneq \phi_n$ for all $n$.
We have already obtained $\psi \deneq \chi$, so that the hypothesis $\psi \deneq \phi$ implies $\chi \deneq \phi$.
On the other hand, we have also obtained $\ppsteps{\chi_n} = \ppsteps{\phi_n}$ for all $n$. Then a simple induction on $n$ yields $\chi_n \deneq \phi_n$ for all $n$. Thus we conclude.
\end{proof}

\begin{theorem}
\label{rsl:breq-iff-deneq}
Let $\psi$, $\phi$ be \pnpterms. Then $\psi \breq \phi$ iff $\psi \deneq \phi$.
\end{theorem}

\begin{proof}
Immediate corollary of Prop.~\ref{rsl:breq-then-deneq} and Prop.~\ref{rsl:deneq-then-breq}.
\end{proof}

\newpage
\section{Compression}
\label{sec:compression}
The compression lemma, \cite{orthogonal-itrs-90, orthogonal-itrs-95, terese, KetemaRTA12} established that the full power of strongly convergent reduction can be achieved considering only reductions having length at most $\omega$, \ie\ the first infinite ordinal.
Formally, the lemma states that for any strongly convergent \redseq\ $t \infredx{\reda} u$, there exists another strongly convergent \redseq\ $t \infredx{\redb} u$ and $\redln{\redb} \leq \omega$.
In \cite{orthogonal-itrs-95} a more precise statement is given: for orthogonal \TRSs, then $\redb$ can be chosen such that it is L\'evy-equivalent (\confer\ \cite{huetLevy91}) to $\reda$.

\doNotIncludeStandardisation{
The aim of this section is to present a novel proof of the property of compression for convergent first-order rewriting, based on the characterisation of \peqence\ given in Section~\ref{sec:peqence}.
}
Given that any convergent \redseq\ can be described by means of a proof term, \confer\ Prop.~\ref{rsl:denotation-existence}, compression can be studied within the framework given by proof terms. 
In this setting, the compression result can be stated as follows: for any convergent proof term (\confer\ Dfn.~\ref{dfn:ppterm}) $\psi$, there exists a \pnpterm\ $\phi$ such that $\psi \peq \phi$ and $\ppsteps{\phi} \leq \omega$.
Observe that the obtained result is more general than the statements present in the referenced literature, in two ways. 
Firstly, the result applies to orthogonal \emph{\redseqs}, even for non-orthogonal \TRSs.
Secondly, the result applies to (the description of) arbitrary contraction activity, independently of whether it is described sequentially. Put in this way, the compression result indicates that any orthogonal contraction activity can be sequentiated in at most $\omega$ steps.

\includeStandardisation{
The compression result is a straightforward consequence of the existence standardisation result, namely Thm.~\ref{rsl:dl-std-peq-existence}.
In this section, an independent proof of compression is given, which does not depend on standardisation. This proof resorts to the \emph{factorisation} results given in Sec.~\ref{sec:factorisation}.
}
\doNotIncludeStandardisation{
This proof resorts to a key technical result, namely the ability of \textbf{factorising} (more precisely, obtaining a factorised version of) any proof term, in a leading part denoting \emph{finite} contraction activity, followed by a tail denoting activity at \emph{arbitrarily big depths}.
The characterisation of \peqence\ shows that the original proof term and its factorised version denote the same contraction activity, while the concatenation symbol included in the signature of proof terms allows to denote the sequential organisation of contraction activity in the factorised version.
%
Therefore, the main auxiliary result for the compression proof is the existence, for any proof term $\psi$ and $n < \omega$, of two proof terms $\chi$ and $\phi$, such that $\psi \peqe \chi \comp \phi$, $\chi$ is a finite \pnpterm, and $\mind{\phi} > n$.

In the following, we will develop the technical work aiming to obtain the factorisation result, and subsequently we will give a statement of the compression lemma based in proof terms and \peqence, and prove it by resorting to factorisation.
}
\newcommand{\collseq}{collapsing sequence}
\newcommand{\collseqs}{collapsing sequences}
\newcommand{\Collseq}{Collapsing sequence}
\newcommand{\Collseqs}{Collapsing sequences}
\newcommand{\posseq}[2]{\langle #1 \rangle_{#2}}

\includeStandardisation{\subsubsection{Factorisation for \imsteps}}
\doNotIncludeStandardisation{\subsection{Factorisation for \imsteps}}
\label{scn:factorisation-imsteps}

In this section, a factorisation result for the particular case of \imsteps\ is stated an proved. The proof is based on the concept of \textbf{collapsing sequence of positions} for an \imstep.
Such a sequence indicates that the contraction activity denoted by the \imstep\ includes a series of reduction steps which can be performed consecutively and at the same position, so that all of these steps, except possibly the last one,  correspond to collapsing rules.

\Ie, considering the rules $\mu: f(x) \to g(x)$, $\rho: i(x) \to x$ and $\rho': j(x) \to x$, the proof term $h(\rho(\rho'(\mu(a))),\mu(b))$ includes a finite collapsing sequence formed by the occurrences of $\rho$ and $\rho'$ plus the leftmost occurrence of $\mu$. This collapsing sequence indicates that a sequentialisation of the activity denoted by this proof term can include up to three consecutive collapsing steps at the same position.

On the other hand, the proof term $\rho\om$ includes an \emph{infinite} collapsing sequence. Observe that this proof term is \emph{not convergent}. 
In the following, a relation between infinite collapsing sequences and non-convergence is shown%
\footnote{We conjecture that, in fact, non-convergence of \imsteps, and therefore non-termination of developments of orthogonal sets of redex occurrences in first-order rewriting, can be fully characterised by means of \collseqs. 
This observation suggests that \imsteps\ could be used as a technical tool to study termination of developments in infinitary rewriting, leading to an approach being alternative to \eg\ the one described in \cite{terese}, Sec.~12.5.
In this work, only the material needed for the factorisation result is developed. Some conjectures follow; further investigation about this subject is left as future work.

Observe that \imsteps\ exist being $\tgtt$-$WN^\infty$ and including infinite \collseqs. \Eg, if we add the rule $\tau: h(x,y) \to y$, then $\tau(\rho\om, a)$ has $a$ as $\tgtt$-normal form.
Intuitively, \collseqs\ prevent an \imstep\ to be $\tgtt$-$WN^\infty$ are those which cannot be erased. Then we state the following conjecture: an \imstep\ is $\tgtt$-$WN^\infty$ iff it does not include any infinite \collseq\ at a \emph{non-erasable} position, where a position $p$ is erasable for $\psi$ iff $p = p_1 i p_2$, $\psi(p_1) = \mu$, and the $i$-th variable in the left-hand side of $\mu$ does not occur in the corresponding right-hand side.}%
, and later exploited in the proof of the factorisation result for \imsteps.

\begin{definition}
\label{dfn:collseq}
Let $\psi$ be an \imstep. A sequence $\posseq{p_i}{i \leq n}$ (resp. $\posseq{p_i}{i < \omega}$) is a \emph{finite} (resp. \emph{infinite}) \emph{\collseq} for $\psi$ iff for all $i < n$ (resp. $i < \omega$), $\psi(p_i) = \mu$ where $\mu: l[x_1, \ldots, x_m] \to x_j$ and $p_{i+1} = p_i \, j$.
\end{definition}

Observe that the length of $\posseq{p_i}{i \leq n}$ is $n + 1$.
Moreover, for any $\posseq{p_i}{i \leq n}$ or $\posseq{p_i}{i < \omega}$, an easy induction (on $k-j$) yields that $j < k < \omega$ implies $p_j < p_k$.

\begin{lemma}
\label{rsl:sub-collseq}
Let $\psi$ be a proof term, $\posseq{p_i}{i \leq n}$ (resp. $\posseq{p_i}{i < \omega}$) a \collseq\ for $\psi$, and $j, k$ such that $j + k \leq n$ (resp $j,k < \omega$).
Then $\posseq{p_{j+i}}{i \leq k}$ is a \collseq\ for $\psi$.
\end{lemma}

\begin{proof}
Easy consequence of Dfn.~\ref{dfn:collseq}.
\end{proof}

Notice that Lem.~\ref{rsl:sub-collseq} implies particularly that $\posseq{p_i}{i \leq k}$ is a \collseq\ if $k \leq n$ (resp. $k < \omega$).

For any $\psi$ \imstep\ and $p \in \Pos{\psi}$, we observe that $\langle p \rangle$ is a \collseq\ for $\psi$ whose length is 1. This is an easy \emph{existence} result. A \emph{uniqueness} result for \collseqs\ holds as well, namely:

\begin{lemma}
\label{rsl:collseq-uniqueness}
Let $\psi$ be an \imstep, $p \in \Pos{\psi}$, and $n$ such that $0 < n < \omega$.
Then there is at most one \collseq\ for $\psi$ starting at $p$ and having length $n$.
\end{lemma}

\begin{proof}
We proceed by induction on $n$.
If $n = 1$ then the result holds immediately since the only suitable sequence is $\langle p \rangle$.

Let $n = n'+1$. Let $\posseq{p_i}{i \leq n'}$ and $\posseq{q_i}{i \leq n'}$ two \collseqs\ for $\psi$, both starting with $p$.
Lem.~\ref{rsl:sub-collseq} implies that both $\posseq{p_i}{i \leq (n'-1)}$ and $\posseq{q_i}{i \leq (n'-1)}$ are \collseqs\ for $\psi$. Then \ih\ on $n'$ implies $p_i = q_i$ if $i < n'$, so that particularly $p_{n'-1} = q_{n'-1}$. Applying Dfn.~\ref{dfn:collseq} on $\posseq{p_i}{i \leq n'}$ and $\posseq{q_i}{i \leq n'}$ yields $\psi(p_{n'-1}) = \psi(q_{n'-1}) = \mu$ such that $\mu: l[x_1, \ldots, x_m] \to x_j$ and $p_{n'} = q_{n'} = p_{n'-1} \, j$.
Thus we conclude.
\end{proof}

\begin{lemma}
\label{rsl:collseq-prefix}
Let $\psi$ be an \imstep, $p \in \Pos{\psi}$, and $n, k < \omega$ (resp $n < \omega$), such that both $\posseq{p_i}{i \leq n}$ and $\posseq{q_i}{i \leq n+k}$ (resp., and $\posseq{q_i}{i < \omega}$) are \collseqs\ for $\psi$ starting with $p$.
Then $i \leq n$ implies $q_i = p_i$.
\end{lemma}

\begin{proof}
Easy consequence of Lem.~\ref{rsl:sub-collseq} and Lem.~\ref{rsl:collseq-uniqueness}.
\end{proof}

We already remarked that any prefix of an infinite \collseq\ is a \collseq\ as well. Conversely, a sequence of growing \collseqs\ starting at the same position indicates the presence of an infinite \collseq. The following lemma formalises this idea.

\begin{lemma}
\label{rsl:collseq-infinite}
Let $\psi$ be an \imstep\ and $p \in \Pos{\psi}$, such that for any $n < \omega$, there is a \collseq\ for $\psi$ starting at $p$ and having length $n$.
Then there is an infinite \collseq\ for $\psi$ starting at $p$.
\end{lemma}

\begin{proof}
We define the sequence $\posseq{p_i}{i < \omega}$ as follows: for all $k < \omega$, $p_k \eqdef q_k$ where $\posseq{q_i}{i \leq k}$ is the only (\confer\ Lem.~\ref{rsl:collseq-uniqueness}) \collseq\ for $\psi$ starting at $p$ and having length $k+1$.
Let $j < \omega$, and $\posseq{q_i}{i \leq j}$ and $\posseq{q'_i}{i \leq (j+1)}$ the \collseqs\ for $\psi$ starting at $p$ and having lengths $j+1$ and $j+2$ respectively. Observe that Lem.~\ref{rsl:collseq-prefix} implies $p_j = q_j = q'_j$; on the other hand, $p_{j+1} = q'_{j+1}$. Then $\posseq{q'_i}{i \leq (j+1)}$ being a \collseq\ implies that $\psi(p_j) = \psi(q'_j) = \mu$ where $\mu : l[x_1, \ldots, x_m] \to x_i$ and $p_{j+1} = q'_{j+1} = q'_j \, i = p_j \, i$. 
Consequently, $\posseq{p_i}{i < \omega}$ is a \collseq. Thus we conclude.
\end{proof}

\medskip
After this general presentation of \collseqs, we will focus on \collseqs\ starting with $\epsilon$. 
The existence of an infinite \collseq\ starting with $\epsilon$ is invariant \wrt\ partial computation of the target of an \imstep. This implies that an \imstep\ including such a sequence is non-convergent, \ie\ its target cannot be computed, \confer\ Dfn.~\ref{dfn:src-tgt-imstep} and Dfn.~\ref{dfn:imstep-convergence}.

\begin{lemma}
\label{rsl:collseq-infinite-invariant-tgtt}
Let $\psi$ be an \imstep, $\posseq{p_i}{i < \omega}$ a \collseq\ for $\psi$ starting at $\epsilon$, and $\psi \infredxtrs{\reda}{\tgtt} \phi$.
Then there exists some $\posseq{q_i}{i < \omega}$ being a \collseq\ for $\phi$ starting at $\epsilon$.
\end{lemma}

\begin{proof}
We proceed by transfinite induction on $\redln{\reda}$.
If $\redln{\reda} = 0$, so that $\phi = \psi$, then we conclude immediately.

Assume $\redln{\reda} = \alpha + 1$, so that $\psi \infredxtrs{\reda'}{\tgtt} \chi \sstepxtrs{\stepa}{\tgtt} \phi$ where $\redln{\reda'} = \alpha$; let us say $\stepa = \langle \chi, r, \uln{\mu}, \sigma \rangle$, and define $d \eqdef \sdepth{\stepa} = \posln{r}$, where $\uln{\mu}$ is the rule in $\tgtt$ corresponding to a rule $\mu$ in the object \TRS.
\Ih\ can be applied on $\reda'$, obtaining the existence of $\posseq{p'_i}{i < \omega}$, a \collseq\ for $\chi$ starting at $\epsilon$.
Observe that $\phi = \repl{\chi}{\sigma h}{r}$ where $\mu: l[x_1, \ldots, x_m] \to h$, so that $\uln{\mu} : \mu(x_1, \ldots, x_m) \to h$, implying $\sigma = \set{x_i \eqdef \subtat{\chi}{r \,i}}$.
Notice also that $\posln{p'_n} = n$ for all $n < \omega$, implying $\posln{p'_d} = \posln{r}$.
We consider two cases.
\begin{itemize}
\item 
Assume $p'_d \disj r$. 
Let $n < \omega$. Observe that $n < d$, resp. $n > d$, implies $p'_n < p'_d$, resp. $p'_d < p'_n$. In either case, $r \leq p'_n$ would contradict $p'_d \disj r$, in the former case by transitivity of $<$, in the latter since all prefixes of $p'_n$ form a total order in a tree domain. Hence $r \not\leq p'_n$.
Consequently, for all $n < \omega$, $p'_n \in \Pos{\phi}$ and $\phi(p'_n) = \chi(p'_n)$. 
Thus $\posseq{p'_n}{n < \omega}$ is a \collseq\ for $\phi$.

\item
Assume $p'_d = r$.
In this case, $\mu: l[x_1, \ldots, x_m] \to x_j$ and $p'_{d+1} = p'_d \, j$, so that $\phi = \repl{\chi}{\subtat{\chi}{p'_{d+1}}}{p'_d}$. Observe that for any position $p''$, $\subtat{\phi}{p'_d \, p''} = \subtat{\chi}{p'_{d+1} \, p''}$.

Let $\posseq{q_i}{i < \omega}$ be the sequence defined as follows: \\
$q_n \eqdef \left\{
	\begin{array}{cl}
		p'_n & \textif n \leq d \\
		p'_d p'' \textnormal{ where } p'_{n+1} = p'_{d+1} p'' & \textif n > d
	\end{array}
	\right.$ 

Let $n < \omega$. 
If $n < d$, then $q_n = p'_n < p'_d$, so that $\phi(q_n) = \phi(p'_n) = \chi(p'_n) = \nu$ where $\nu: l[y_1, \ldots, y_m] \to y_i$ and $q_{n+1} = p'_{n+1} = p'_n \,i = q_n \, i$.
Now assume $n \geq d$. Let $p''$ such that $p'_{n+1} = p'_{d+1} p''$, observe that $n = d$ implies $p'' = \epsilon$.
Observe $\chi(p'_{n+1}) = \nu$, $\nu: l[y_1, \ldots, y_m] \to y_i$ and $p'_{n+2} = p'_{n+1} \, i = p'_{d+1} p'' \, i$.
On the other hand, $q_n = p'_d p''$ (if $n = d$, then $q_n = p'_d = p'_d p''$ since in this case $p'' = \epsilon$), $q_{n+1} = p'_d p'' \, i = q_n \, i$, and in turn $\phi(q_n) = \phi(p'_d p'') = \chi(p'_{d+1} p'') = \chi(p'_{n+1}) = \nu$.

Hence $\posseq{q_i}{i < \omega}$ is a \collseq\ for $\phi$. Thus we conclude by observing that $q_0 = p'_0 = \epsilon$.
\end{itemize}

Assume that $\redln{\reda}$ is a limit ordinal.
For any $n < \omega$, we define $\beta_n$, $\chi_n$, $\langle p^n_i \rangle_{i < \omega}$ and $q_n$ as follows: $\beta_n$ is an ordinal such that $\beta_n < \redln{\reda}$ and $\sdepth{\redel{\reda}{\gamma}} > n$ if $\beta_n \leq \gamma < \redln{\reda}$; and $\chi_n$ is the \imstep\ verifying $\psi \infredxtrs{\redupto{\reda}{\beta_n}}{\tgtt} \chi_n \infredxtrs{\redsublt{\reda}{\beta_n}{\redln{\reda}}}{\tgtt} \phi$. Observe that we can assume wlog that $\beta_n \leq \beta_{n+1}$.
In turn, \ih\ on $\redupto{\reda}{\beta_n}$ and Lem.~\ref{rsl:collseq-uniqueness} imply the existence of a unique \collseq\ for $\chi_n$ starting at $\epsilon$; we define $\langle p^n_i \rangle_{i < \omega}$ to be that sequence, and $q_n \eqdef p^n_n$.

Let $n < \omega$.
Then Lem.~\ref{rsl:sub-collseq} implies that $\langle p^n_i \rangle_{i \leq n}$ is a \collseq\ for $\chi_n$.
Moreover, $\beta_n = \beta_{n+1}$ implies $\chi_n = \chi_{n+1}$, and otherwise $\beta_n < \beta_{n+1}$, so that $\psi \infredxtrs{\redupto{\reda}{\beta_n}}{\tgtt} \chi_n \infredxtrs{\redsublt{\reda}{\beta_n}{\beta_{n+1}}}{\tgtt} \chi_{n+1}$ where $\mind{\redsublt{\reda}{\beta_n}{\beta_{n+1}}} > n$. 
Furthermore, $\chi_{n+1} \infredxtrs{\redsublt{\reda}{\beta_{n+1}}{\redln{\reda}}}{\tgtt} \phi$ and $\mind{\redsublt{\reda}{\beta_{n+1}}{\redln{\reda}}} > n$.
Therefore $\tdist{\chi_n}{\chi_{n+1}} < 2^{-n}$ and \\ $\tdist{\chi_{n+1}}{\phi} < 2^{-(n+1)}$ by Lem.~\ref{rsl:redseq-mind-big-src-tgt}; in turn Lem.~\ref{rsl:tdist-is-ultrametric} implies $\tdist{\chi_n}{\phi} < 2^{-n}$.
Then for any $j \leq n$, $\chi_n(p^n_j) = \chi_{n+1}(p^n_j) = \phi(p^n_j)$ since $\posln{p^n_j} = j$. 
Therefore $\langle p^n_i \rangle_{i \leq n}$ is a \collseq\ for $\chi_{n+1}$, so that Lem.~\ref{rsl:collseq-uniqueness} implies $p^n_j = p^{n+1}_j$ if $j \leq n$.
Hence $q_n = p^{n+1}_n$, so that $\phi(q_n) = \chi_{n+1}(q_n) = \nu$ where $\nu: l[x_1, \ldots, x_m] \to x_i$ and $q_{n+1} = p^{n+1}_{n+1} = p^{n+1}_n \, i = q_n \, i$.
Consequently, $\posseq{q_i}{i < \omega}$ is a \collseq\ for $\phi$. Thus we conclude by observing $q_0 = \epsilon$.
\end{proof}

\begin{lemma}
\label{rsl:infinite-collseq-then-tgtt-non-wn}
Let $\psi$ be an \imstep\ such an infinite \collseq\ for $\psi$ starting at $\epsilon$ exists. Then $\psi$ is not $\tgtt$-weakly normalising.
\end{lemma}

\begin{proof}
Let $\psi \infredtrs{\tgtt} \phi$. Then Lem.~\ref{rsl:collseq-infinite-invariant-tgtt} implies that an infinite \collseq\ for $\phi$ starting at $\epsilon$ exists, so that $\phi$ is not a $\tgtt$-normal form. Thus we conclude.
\end{proof}

\medskip
On the other hand, the inexistence of arbitrarily large \collseqs\ starting at $\epsilon$ allows a finite $\tgtt$-reduction sequence ending in a proof term having a function symbol at the root. In turn, for any finite $\tgtt$-reduction sequence there is a corresponding finite \pnpterm.

\begin{lemma}
\label{rsl:finite-collseq-then-finite-tgtt-redseq}
Let $\psi$ be an \imstep\ and $n$ verifying $1 < n < \omega$, such that there is no \collseq\ for $\psi$ starting at $\epsilon$ and having length $n$.
Then there exists a $\tgtt$-reduction sequence $\reda$ verifying $\psi \sredxtrs{\reda}{\tgtt} \phi$, $\redln{\reda} < n$, $\sdepth{\redel{\reda}{i}} = 0$ for all $i < \redln{\reda}$, and $\phi(\epsilon) \in \Sigma$.
\end{lemma}

\begin{proof}
We proceed by induction on $n$.

Assume $n = 2$. If $\psi(\epsilon) \in \Sigma$ then we conclude immediately.
Otherwise $\psi(\epsilon) = \mu$ where $\mu : l \to f(t_1, \ldots, t_k)$, so that the corresponding rule in $\tgtt$ is $\uln{\mu} : \mu(x_1, \ldots, x_m) \to f(t_1, \ldots, t_k)$, and therefore $\psi \sstepxtrs{(\epsilon, \uln{\mu})}{\tgtt} f(t'_1, \ldots, t'_k)$; thus we conclude by taking $\reda \eqdef \langle (\epsilon, \uln{\mu}) \rangle$.

Assume $n = n' + 1$ and $1 < n' < \omega$. 
If $\psi(\epsilon) \in \Sigma$ or $\psi(\epsilon) = \mu$, $\mu: l \to h$ and $h \notin \thevar$, then the argument of the previous case allows to conclude.
Otherwise, \ie\ if $\psi(\epsilon) = \mu$ and $\mu: l[x_1, \ldots, x_m] \to x_k$, then the corresponding rule in $\tgtt$ is $\uln{\mu} : \mu(x_1, \ldots, x_m) \to x_k$, implying that $\psi \sstepxtrs{(\epsilon, \uln{\mu})}{\tgtt} \subtatnarrowleft{\psi}{k}$.
Observe that $\posseq{p_i}{i \leq n'}$ being a \collseq\ for $\subtat{\psi}{k}$ starting at $\epsilon$ would imply $(\langle \epsilon \rangle; \posseq{k \,p_i}{i \leq n'})$ to be a \collseq\ for $\psi$ having length $n$, thus contradicting the lemma hypotheses.
Indeed, if we define $\posseq{q_i}{i \leq n}$ as the given sequence for $\psi$, then $q_0 = \epsilon$ and $q_1 = k$, so that the condition on \collseqs\ holds for $j = 0$. If $0 < j < n$, then $q_j = k \, p_{j-1}$, so that $\psi(q_j) = \subtatnarrow{\psi}{k}(p_{j-1}) = \nu$ where $\nu : l[y_1, \ldots, y_m] \to y_i$ and $p_j = p_{j-1} \, i$, implying $q_{j+1} = k \, p_j = k \, p_{j-1} \, i = q_j \, i$.

Therefore \ih\ can be applied to $\subtatnarrowleft{\psi}{k}$, yielding the existence of a reduction sequence $\reda'$ verifying $\subtatnarrowleft{\psi}{k} \sredxtrs{\reda'}{\tgtt} \phi$, $\redln{\reda'} < n'$, $\sdepth{\redel{\reda'}{i}} = 0$ for all $i < n'$, and $\phi(\epsilon) \in \Sigma$.
Thus we conclude by taking $\reda \eqdef (\epsilon,\uln{\mu}); \reda'$.
\end{proof}

\begin{lemma}
\label{rsl:tgtt-step-to-leading-one-step}
Let $\psi$ be an \imstep, and $\psi \sstepxtrs{\stepa}{\tgtt} \phi$. Then there exists a one-step $\chi$ such that $\psi \peqe \chi \comp \phi$ and $\sdepth{\chi} = \sdepth{\stepa}$.
\end{lemma}

\begin{proof}
We proceed by induction on $\sdepth{\stepa}$.

Assume $\stepa = (\epsilon, \uln{\mu})$, say $\mu: l[x_1, \ldots, x_m] \to h[x_1, \ldots, x_m]$ so that the corresponding rule in $\tgtt$ is $\uln{\mu} : \mu(x_1, \ldots, x_m) \to h[x_1, \ldots, x_m]$. Therefore $\psi = \mu(\psi_1, \ldots, \psi_m)$ and $\phi = h[\psi_1, \ldots, \psi_m]$.
We take $\chi \eqdef \mu(src(\psi_1), \ldots, src(\psi_m))$.
Then \peqoutin\ yields exactly $\psi \peqe \chi \comp \phi$. Thus we conclude.

Assume $\stepa = (ip, \uln{\mu})$.
In this case, $\psi = f(\psi_1, \ldots, \psi_i, \ldots, \psi_m)$, $\phi = f(\psi_1, \ldots, \phi_i, \ldots, \psi_m)$, and $\psi_i \sstepxtrs{(p,\uln{\mu})}{\tgtt} \phi_i$. Then \ih\ on $(p, \uln{\mu})$ implies $\psi_i \peqe \chi_i \comp \phi_i$ where $\chi_i$ is a one-step verifying $\sdepth{\chi_i} = \posln{p}$.
We take $\chi \eqdef f(src(\psi_1), \ldots, \chi_i, \ldots, src(\psi_m))$.
Observe that for any $j \neq i$, \peqidleft\ implies $\psi_j \peqe src(\psi_j) \comp \psi_j$, so that \\
\hspace*{1cm}
$\begin{array}{rcl}
\psi & \peqe & 
f(src(\psi_1) \comp \psi_1, \ldots, \chi_i \comp \phi_i, \ldots, src(\psi_m) \comp \phi_m) \\
& \peqe &
f(src(\psi_1), \ldots, \chi_i, \ldots, src(\psi_m)) \comp
	f(\psi_1, \ldots, \phi_i, \ldots, \psi_m) \\
& = & \chi \comp \phi	
\end{array}
$ \\[2pt]
Thus we conclude by noticing that $\sdepth{\chi} = \posln{p} + 1 = \sdepth{\stepa}$.
\end{proof}

\begin{lemma}
\label{rsl:tgtt-redseq-to-leading-pnpterm}
Let $\psi$ be an \imstep\ and $\psi \sredxtrs{\reda}{\tgtt} \phi$.
Then there exists a finite \pnpterm\ $\chi$ such that $\psi \peqe \chi \comp \phi$, $\ppsteps{\chi} = \redln{\reda}$, and $\sdepth{\redel{\chi}{i}} = \sdepth{\redel{\reda}{i}}$ for all $i < \ppsteps{\chi}$.
\end{lemma}

\begin{proof}
Easy induction on $\redln{\reda}$.
If $\reda$ is an empty reduction sequence, then we conclude just by taking $\chi \eqdef src(\psi)$.

Assume $\reda = \stepa; \reda'$, so that $\psi \sstepxtrs{\stepa}{\tgtt} \psi_0 \sredxtrs{\reda'}{\tgtt} \phi$.
Then Lem.~\ref{rsl:tgtt-step-to-leading-one-step} implies that $\psi \peqe \chi_0 \comp \psi_0$ where $\chi_0$ is a one-step verifying $\sdepth{\chi_0} = \sdepth{\stepa}$, and \ih\ on $\reda'$ yields $\psi_0 \peqe \chi' \comp \phi$ where $\chi'$ is a finite \pnpterm\ verifying $\ppsteps{\chi'} = \redln{\reda'}  = \redln{\reda} - 1$, and $\sdepth{\redel{\chi'}{i}} = \sdepth{\redel{\reda'}{i}} = \sdepth{\redel{\reda}{i+1}}$ if $i < \ppsteps{\chi'}$.

We take $\chi \eqdef \chi_0 \comp \chi'$. It is straightforward to verify that $\chi$ satisfies the conditions about length and step depth. 
Moreover, $\psi_0 \peqe \chi' \comp \phi$ implies $\chi_0 \comp \psi_0 \peqe \chi_0 \comp (\chi' \comp \phi) \peqe \chi \comp \phi$, so that \eqltrans\ yields $\psi \peqe \chi \comp \phi$ (recall $\psi \peqe \chi_0 \comp \psi_0$). Thus we conclude.
\end{proof}

\medskip
The previous auxiliary results allow to prove the main result of this section, \ie\ factorisation for \imsteps.

\begin{lemma}
\label{rsl:factorisation-imstep}
Let $\psi$ be a convergent \imstep.
Then there exist $\chi$, $\phi$ such that 
$\, \psi \peqe \chi \comp \phi \, $, 
$\, \chi$ is a finite \pnpterm\ verifying $\sdepth{\redel{\chi}{i}} = 0$ for all $i < \ppsteps{\chi}$, and 
$\phi$ is a convergent \imstep\ verifying $\mind{\phi} > 0$.
\end{lemma}

\begin{proof}
We define $A \eqdef \{n \setsthat 0 < n < \omega$ and there is no \collseq\ for $\psi$ starting at $\epsilon$ and having length $n \}$.
Dfn.~\ref{dfn:imstep-convergence} implies that $\psi$ is $\tgtt$-weakly normalising. Then Lem.~\ref{rsl:infinite-collseq-then-tgtt-non-wn} implies that there is no infinite \collseq\ for $\psi$ starting at $\epsilon$, so that Lem.~\ref{rsl:collseq-infinite} implies $A \neq \emptyset$.
Let $n \in A$. Then Lem.~\ref{rsl:finite-collseq-then-finite-tgtt-redseq} implies $\psi \sredxtrs{\reda}{\tgtt} \phi$, where $\redln{\reda} < \omega$, $\sdepth{\redel{\reda}{i}} = 0$ for all suitable $i$, and $\phi$ is an \imstep\ (since it is the target of a $\tgtt$-reduction sequence) verifying $\mind{\phi} > 0$ (since $\phi(\epsilon) \in \Sigma)$.
Moreover, $\psi$ being convergent means that $\psi$ is $\tgtt$-$WN^\infty$, and $\tgtt$ is a convergent \iTRS, so that Lem.~\ref{rsl:orthogonal-leading-head-steps-to-nf} implies that $\phi$ is also $\tgtt$-$WN^\infty$, \ie\ convergent.
We conclude by applying Lem.~\ref{rsl:tgtt-redseq-to-leading-pnpterm} on $\psi \sredxtrs{\reda}{\tgtt} \phi$.
\end{proof}

\subsection{Fixed prefix of contraction activity}
\label{sec:cfpc}
\cfpInsideCompression{%
This section introduces a technical tool, in which the extension of the factorisation result from \imsteps\ to arbitrary proof terms is based on.
This tool is a formalisation of a simple observation: 
the
}%
\cfpInsideStandardisation{The }%
contraction activity denoted by a proof term can lie below some \emph{fixed prefix}. \Ie, the contraction activity corresponding to either of the equivalent proof terms
$h(\mu(a) \cdot \nu(a), \pi)$ and $h(\mu(a),a) \cdot h(\nu(a),a) \cdot h(k(a),\pi)$
leaves the context $h(\Box,\Box)$ fixed, so we will say that $h(\Box,\Box)$ is a fixed prefix for these proof terms.
For proof terms involving root activity, the only possible fixed prefix is $\Box$.
In the sequel, we will establish that fixed prefixes are invariant \wrt\ \peqence.

Computing a fixed prefix for a proof term $\psi$ allows to permute it with a one-step (\confer\ Dfn.~\ref{dfn:one-step}) performed on $tgt(\psi)$, whose redex lies in the fixed prefix of $\psi$.
This observation will be crucial in order to%
\cfpInsideStandardisation{prove some of the main results of this study.}
\cfpInsideCompression{prove a general factorisation result, since it allows to obtain a proof term in which the (denotation of the) activity near to the root 
``shifts to the left as much as possible'', \ie\
lies in the lesser possible positions 
\wrt\ the sequentialisation order given by dot occurrences.}

The following definitions and results characterise the common prefix of a proof term in a way allowing to manipulate it.
The %
\cfpInsideStandardisation{sets of }%
positions mentioned in the statements must be understood as%
\cfpInsideStandardisation{sets of \emph{contraction} positions.}
\cfpInsideCompression{being relative to the contraction activity denoted by a proof term, rather than as positions in proof terms themselves.}

\medskip
We formalise the concept of (the activity denoted by) a proof term having a fixed prefix by defining a relation between proof terms and prefix-closed sets of positions, which we will call \emph{respect}.
Therefore, if $\psi$ respects a set of positions $\Spa$, then $\psi$ has a fixed prefix corresponding to the positions in $\Spa$.

\begin{definition}
\label{dfn:spa-proj}
Let $\Spa$ be a set of positions, and $i \in \Nat$.
Then we define the \emph{projection of $\Spa$ on $i$} as 
$\proj{\Spa}{i} \eqdef \set{p \setsthat ip \in \Spa}$.
\end{definition}

\begin{definition}
\label{dfn:term-prefix}
Let $t$ be a term, and $\Spa$ a finite and prefix-closed set of positions such that $\Spa \subseteq \Pos{t}$.
Then we define $\pref{t}{\Spa}$, the \emph{prefix of $t$ \wrt\ $\Spa$},  as follows. \\
If $\Spa = \emptyset$, then $\pref{t}{\Spa} \eqdef \Box$. \\
If $\Spa \neq \emptyset$ and $t \in \thevar$, so that $\Spa = \set{\epsilon}$, then $\pref{t}{\Spa} \eqdef t$. \\
If $\Spa \neq \emptyset$ and $t = f(t_1, \ldots, t_m)$, so that $\Spa = \set{\epsilon} \,\cup\, \bigcup_{1 \leq i \leq m} (i \cdot \proj{\Spa}{i})$, then $\pref{t}{\Spa} \eqdef f(\pref{t_1}{\proj{\Spa}{1}}, \ldots, \pref{t_m}{\proj{\Spa}{m}})$.
%
\end{definition}

Notice that $C = \pref{t}{\Spa}$ iff $t = C[t_1, \ldots, t_k]$ and $\Spa = \set{p \setsthat p \in \Pos{C} \,\land\, C(p) \neq \Box}$, this can be verified by a simple induction on the cardinal of $\Spa$.

\cfpInsideStandardisation{
\begin{definition}
\label{dfn:respects}
Let $\psi$ be a proof term, $\Spa$ a finite and prefix-closed set of positions.
We say that $\psi$ \emph{respects} $\Spa$ iff $\Spa \subseteq \Pos{src(\psi)}$ and $\frso{\psi}{r}$ is undefined for all $r \in \Spa$.
\end{definition}
}

\cfpInsideCompression{
\begin{definition}
\label{dfn:respects}
Let $\psi$ be a proof term, and $\Spa$ a set of positions.
We say that $\psi$ \emph{respects} $\Spa$ iff $\Spa$ is finite and prefix-closed, and any of the following applies: \\
\begin{tabular}{@{$\ \ \bullet\ \ $}p{.9\textwidth}}
$\psi$ is an \imstep, $\Spa \subseteq \Pos{\psi}$ and $\psi(p) \in \Sigma$ for all $p \in \Spa$. \\
$\psi = \psi_1 \comp \psi_2$ and both $\psi_1$ and $\psi_2$ respect $\Spa$. \\
$\psi = \icomp \psi_i$ and all $\psi_i$ respect $\Spa$. \\
$\psi = f(\psi_1, \ldots, \psi_m)$, at least one of the $\psi_i$ is not an \imstep, and either $\Spa = \emptyset$ or $\psi_i$ respects $\proj{\Spa}{i}$ for all $i \leq m$. \\
$\psi = \mu(\psi_1, \ldots, \psi_m)$, at least one of the $\psi_i$ is not an \imstep, and $\Spa = \emptyset$.
\end{tabular}
\end{definition}
}

\medskip
The relation just defined enjoys some simple properties.

\cfpInsideCompression{
\begin{lemma}
\label{rsl:respects-then-src}
Let $\psi$ be a proof term and $\Spa$ such that $\psi$ respects $\Spa$. Then $\Spa \subseteq \Pos{src(\psi)}$.
\end{lemma}

\begin{proof}
An easy induction on $\psi$ suffices; \confer\ Prop.~\ref{rsl:pterm-induction-principle}.
\end{proof}

\begin{lemma}
\label{rsl:respects-then-tgt}
Let $\psi$ be a convergent proof term and $\Spa$ such that $\psi$ respects $\Spa$. Then $\Spa \subseteq \Pos{tgt(\psi)}$.
\end{lemma}

\begin{proof}
An easy induction on $\psi$ suffices; \confer\ Prop.~\ref{rsl:pterm-induction-principle}.
If $\psi = \icomp \psi_i$ and $\Spa \subseteq \Pos{tgt(\psi_i)}$ for all $i < \omega$, given $p \in \Spa$, we consider $n$ such that $\tdist{tgt(\psi_i)}{tgt(\psi)} < 2^{-\posln{p}}$ if $i > n$, so that $p \in \Pos{tgt(\psi_{n+1})}$ implies $p \in \Pos{tgt(\psi)}$.
\end{proof}

\begin{lemma}
\label{rsl:respect-fnsymbol-coherence}
Let $\psi = f(\psi_1, \ldots, \psi_m)$, and $\Spa$ a set of positions.
Then $\psi$ respects $\Spa$ iff either $\Spa = \emptyset$ or $\psi_i$ respects $\proj{\Spa}{i}$ for all $i \leq m$.
\end{lemma}

\begin{proof}
If $\psi$ is an \imstep, then a straightforward analysis yields the desired result. If at least one of the $\psi_i$ is not an \imstep, then we conclude immediately. Any other case in Dfn.~\ref{dfn:respects} contradicts the stated form of $\psi$.
\end{proof}

\begin{lemma}
\label{rsl:respects-emptyset}
Let $\psi$ be a proof term. Then $\psi$ respects $\emptyset$.
\end{lemma}

\begin{proof}
A straightforward induction on $\psi$, \confer\ Prop.~\ref{rsl:pterm-induction-principle}, suffices to conclude.
\end{proof}

\medskip
The \emph{respects} relation can be obtained from conditions on the target and the minimum activity depth of a proof term.

\begin{lemma}
\label{rsl:mind-plus-tgt-then-respects}
Let $\psi$ be a convergent proof term and $\Spa$ a finite, prefix-closed set of positions, such that $\mind{\psi} > n$, $\posln{p} \leq n$ for all $p \in \Spa$, and $\Spa \subseteq \Pos{tgt(\psi)}$. 
Then $\psi$ respects $\Spa$.
\end{lemma}

\begin{proof}
We proceed by induction on $\psi$, \confer\ Prop.~\ref{rsl:pterm-induction-principle}.

Assume that $\psi$ is an \imstep. 
If $\Spa = \emptyset$ then Lem.~\ref{rsl:respects-emptyset} allows to conclude immediately. 
Otherwise, $\epsilon \in \Spa$, implying $\psi = f(\psi_1, \ldots, \psi_m)$.
We proceed by induction on $n$.
If $n = 0$, then the only set of positions compatible with the lemma hypotheses is $\Spa = \set{\epsilon}$, so that we conclude immediately.
Assume $n = n' + 1$, and let $i$ such that $1 \leq i \leq m$.
It is straightforward to verify that $\mind{\psi_i} > n'$, that $\posln{p} \leq n'$ for all $p \in \proj{\Spa}{i}$, and also that $\proj{\Spa}{i} \subseteq \Pos{tgt(\psi_i)}$ (recall $tgt(\psi) = f(tgt(\psi_1), \ldots tgt(\psi_m)) \,$). 
Therefore, we can apply \ih\ on $\psi_i$, obtaining that $\psi_i$ respects $\proj{\Spa}{i}$, so that $\proj{\Spa}{i} \subseteq \Pos{\psi_i}$, and moreover for any $p \in \proj{\Spa}{i}$, $\psi(i p) = \psi_i(p) \in \Sigma$.
Hence the desired result holds immediately.

Assume $\psi = \psi_1 \comp \psi_2$. 
In this case, $\mind{\psi_i} > n$ for $i = 1,2$, and $\Spa \subseteq \Pos{tgt(\psi)} = \Pos{tgt(\psi_2)}$.
Then \ih\ applies to $\psi_2$ yielding that $\psi_2$ respects $\Spa$. In turn, Lem.~\ref{rsl:respects-then-src} implies $\Spa \subseteq \Pos{src(\psi_2)} = \Pos{tgt(\psi_1)}$. Then \ih\ applies to $\psi_1$ as well, implying that $\psi_1$ respects $\Spa$. Thus we conclude.

Assume $\psi = \icomp \psi_i$.
Observe that $\mind{\psi} > n$ implies $\mind{\psi_i} > n$ for all $i < \omega$.
Let $k$ such that $\tdist{tgt(\psi_i)}{tgt(\psi)} < 2^{-k}$ for all $i > k$.
Let $j > k$. Then $\Spa \subseteq \Pos{tgt(\psi)}$ implies $\Spa \subseteq \Pos{tgt(\psi_j)}$. Then \ih\ can be applied to $\psi_j$ obtaining that $\psi_j$ respects $\Spa$.
In turn, $\psi_{k+1}$ respecting $\Spa$ implies that $\Spa \subseteq \Pos{src(\psi_{k+1})} = \Pos{tgt(\psi_k)}$.
Therefore \ih\ applies also to $\psi_k$, yielding that $\psi_k$ respects $\Spa$, and then Lem.~\ref{rsl:respects-then-src} implies $\Spa \subseteq \Pos{src(\psi_k)} = \Pos{tgt(\psi_{k-1})}$.
Successive application of an analogous argument yields that $\psi_i$ respects $\Spa$ for all $i \leq k$.
Thus we conclude.

If $\psi = f(\psi_1, \ldots, \psi_m)$, then an argument analogous to that given for \imsteps\ applies.

Finally, $\psi = \mu(\psi_1, \ldots, \psi_m)$ contradicts $\mind{\psi} > n$ for any $n < \omega$.
\end{proof}
}

\cfpInsideStandardisation{
\begin{lemma}
\label{rsl:respects-invariant-proj}
Let $\psi = f(\psi_1, \ldots, \psi_m)$, $\Spa$ such that $\Spa \neq \emptyset$ and $\psi$ respects $\Spa$. 
Then for each $i$, $\psi_i$ respect $\proj{\Spa}{i}$.
\end{lemma}

\begin{proof}
Immediate.
\end{proof}
}

\medskip
The \emph{respects} relation is invariant \wrt\ base \peqence.

\cfpInsideCompression{
\begin{lemma}
\label{rsl:respects-invariant-peqe}
Let $\psi$, $\phi$ be convergent proof terms and $\Spa$ a set of positions, such that $\psi \peqe \phi$.
Then $\psi$ respects $\Spa$ iff $\phi$ respects $\Spa$.
\end{lemma}

\begin{proof}
We proceed by induction on $\alpha$ where $\psi \layerpeqe{\alpha} \phi$, analysing the rule used in the last step of that judgement.

If the rule is \eqlrefl, then we conclude immediately.

If the rule is \eqleqn, then we analyse the equation used.
\begin{itemize}
\item 
\peqidleft\ or \peqidright, \ie\ $\psi = src(\phi) \comp \phi \ $ or $\ \psi = \phi \comp tgt(\phi)$.
The $\Rightarrow )$ direction is immediate. 
For the $\Leftarrow )$ direction, observe that Lem.~\ref{rsl:respects-then-src} and Lem.~\ref{rsl:respects-then-tgt} imply $\Spa \subseteq \Pos{src(\phi)}$ and $\Spa \subseteq \Pos{tgt(\phi)}$ respectively. 
Then Dfn.~\ref{dfn:respects} for \imsteps\ implies immediately that both $src(\phi)$ and $tgt(\phi)$ respect $\Spa$. Thus we conclude.

\item
\peqassoc, \ie\ $\psi = \chi_1 \comp (\chi_2 \comp \chi_3)$ and $\phi = (\chi_1 \comp \chi_2) \comp \chi_3$. In this case either $\psi$ or $\phi$ respects $\Spa$ iff $\chi_1$, $\chi_2$ and $\chi_3$ do. Thus we conclude.

\item
\peqstruct, \ie\ $\psi = f(\chi_1, \ldots, \chi_m) \comp f(\xi_1, \ldots, \xi_m)$ and $\phi = f(\chi_1 \comp \xi_1, \ldots, \chi_m \comp \xi_m)$.
If $\Spa = \emptyset$, then both $\psi$ and $\phi$ respect $\Spa$; \confer\ Lem.~\ref{rsl:respects-emptyset}.
Otherwise \\
$\psi$ respects $\Spa$ \\
\begin{tabular}{@{\hspace*{1cm}}l}
iff both $f(\chi_1, \ldots, \chi_m)$ and $f(\xi_1, \ldots, \xi_m)$ do \\
iff for all $j$ such that $1 \leq j \leq m$, both $\chi_j$ and $\xi_j$ respect $\proj{\Spa}{j}$ \\
iff for all $j$ such that $1 \leq j \leq m$, $\chi_j \comp \xi_j$ respects $\proj{\Spa}{j}$ \\
iff $\phi$ respects $\Spa$.
\end{tabular} \\
Thus we conclude.

\item
\peqinfstruct. This case admits an argument analogous to the one used for \peqstruct.

\item
\peqoutin\ and \peqinout. In this case, it is immediate that either $\psi$ or $\phi$ respects $\Spa$ iff $\Spa = \emptyset$.
\end{itemize}

If the rule used in the last step of the judgement $\psi \layerpeqe{\alpha} \phi$ is \eqlsymm, \eqltrans, \eqlfun, \eqlcomp\ or \eqlinfcomp, then a straightforward inductive arguments suffices to obtain the desired result.

Finally, if the rule is \eqlrule, then it is immediate to verify that either $\psi$ or $\phi$ respect $\Spa$ iff $\Spa = \emptyset$.
\end{proof}
}

\cfpInsideStandardisation{
\begin{lemma}
\label{rsl:respects-invariant-peqe}
Let $\psi$, $\phi$ be convergent proof terms and $\Spa$ a set of positions, such that $\psi \peqe \phi$ and $\psi$ respects $\Spa$.
Then $\phi$ respects $\Spa$.
\end{lemma}

\begin{proof}
Lem.~\ref{rsl:peq-then-same-src-mind-tgt} implies immediately $\Spa \subseteq \Pos{\phi}$.
Assume for contradiction that $\frso{\phi}{r}$ is defined for some $r \in \Spa$.
Then L.~\ref{rsl:frso-peqe} implies that $\frso{\psi}{r'}$ is defined for some $r' \leq r$, so that $r' \in \Spa$ (recall $\Spa$ is prefix-closed), contradicting lemma hypotheses. Thus we conclude.
\end{proof}
}

\medskip
Observe that proof terms whose minimum activity depth is greater than 0 are exactly those which respect $\set{\epsilon}$. Lem.~\ref{rsl:peq-then-same-src-mind-tgt} implies this condition to be stable by \peqence.
For such proof terms, we define their \emph{condensed-to-fixed-prefix-symbol form}, which is a proof term denoting the same activity as the original proof term, and having a function symbol at the root. 
\Eg\ the condensed-to-fixed-prefix-symbol form of $f(\mu(a)) \comp f(\nu(a))$ is $f(\mu(a) \comp \nu(a))$. The condensed-to-fixed-prefix-symbol form of an already condensed proof term is itself, so that it is idempotent.

\begin{lemma}
\label{rsl:respects-epsilon-then-stable-root}
Let $\psi$ a convergent proof term which respects $\set{\epsilon}$. 
Then $src(\psi)(\epsilon) = tgt(\psi)(\epsilon)$.
\end{lemma}

\begin{proof}
We proceed by induction on $\psi$, \confer\ Prop.~\ref{rsl:pterm-induction-principle}. 
If $\psi = f(\psi_1, \ldots, \psi_m)$ then the result holds immediately, while $\psi = \mu(\psi_1, \ldots, \psi_m)$ contradicts the lemma hypotheses.

If $\psi = \psi_1 \comp \psi_2$ and the result holds for both components, then lemma hypotheses imply that both $\psi_1$ and $\psi_2$ respect $\set{\epsilon}$, so that $src(\psi_j)(\epsilon) = tgt(\psi_j)(\epsilon)$ for $j = 1,2$. Observe $src(\psi) = src(\psi_1)$, $tgt(\psi) = tgt(\psi_2)$, and moreover $tgt(\psi_1) = src(\psi_2)$ (by the coherence condition on the definition of $\psi$). Thus we conclude immediately.

Assume $\psi = \icomp \psi_i$ and the result holds for each $\psi_i$. For any $i < \omega$, lemma hypotheses imply that $\psi_i$ respects $\set{\epsilon}$, and therefore $src(\psi_i)(\epsilon) = tgt(\psi_i)(\epsilon)$.
Given $tgt(\psi_i) = src(\psi_{i+1})$ for all $i < \omega$, an easy inductive argument yields $src(\psi)(\epsilon) = src(\psi_0)(\epsilon) = tgt(\psi_i)(\epsilon)$ for any $i < \omega$. Let $n$ such that $\tdist{tgt(\psi_k)}{tgt(\psi)} < 1$ if $k > n$; recall $tgt(\psi) = \lim_{i \to \omega}(tgt(\psi_i))$.
Then $tgt(\psi)(\epsilon) = tgt(\psi_{n+1})(\epsilon) = src(\psi)(\epsilon)$. Thus we conclude.
\end{proof}

\begin{definition}
\label{dfn:cfps}
Let $\psi$ be a proof term which respects $\set{\epsilon}$. 
We define $\cfps{\psi}$, \ie\ the \emph{condensed to fixed prefix symbol form} of $\psi$, as follows. \\[2pt]
\begin{tabular}{@{}l@{}l@{\ \ }l@{\ \ }l}
$\ \ \bullet \ \ $ & 
if $\psi = f(\psi_1, \ldots, \psi_n)$ & then & 
		$\cfps{\psi} \eqdef \psi$. \\
$\ \ \bullet \ \ $ & 
if $\psi = \psi_1 \comp \psi_2$ & then & 
		$\cfps{\psi} \eqdef f(\psi_{11} \comp \psi_{21}, \ldots, \psi_{1m} \comp \psi_{2m})$ \\
		& & & where $\cfps{\psi_i} = f(\psi_{i1}, \ldots, \psi_{im})$ for $i = 1,2$ \\
$\ \ \bullet \ \ $ & 
if $\psi = \icomp \psi_i$ & then & 
		$\cfps{\psi} \eqdef f(\icomp \psi_{i1}, \ldots, \icomp \psi_{im})$ \\
		& & & where $\cfps{\psi_i} = f(\psi_{i1}, \ldots, \psi_{im})$ for all $i < \omega$. \\
$\ \ \bullet \ \ $ & 
\multicolumn{3}{@{}l}{$\psi = \mu(\psi_1, \ldots, \psi_m)$ contradicts $\psi$ respecting $\set{\epsilon}$.}
\end{tabular} \\[2pt]
Lem.~\ref{rsl:respects-epsilon-then-stable-root} implies the soundness of the clauses corresponding to both binary and infinite concatenation.
\end{definition}

\medskip
Condensed-to-fixed-prefix-symbol forms enjoy some properties related with base \peqence\ and the \emph{respects} relation. In turn, these properties allow a simple proof of the extension of Lem.~\ref{rsl:respects-epsilon-then-stable-root} to arbitrary finite and prefix-closed sets of positions.

\begin{lemma}
\label{rsl:cfps-peqe}
Let $\psi$ be a proof term which respects $\set{\epsilon}$.
Then $\psi \peqe \cfps{\psi}$.
\end{lemma}

\begin{proof}
Easy induction on $\psi$. For the infinitary composition case, resort to the \eqlinfcomp\ rule and the \peqinfstruct\ equation, \confer\ Dfn.~\ref{dfn:peqe}.
\end{proof}

\begin{lemma}
\label{rsl:peqe-then-cfps-proj-peqe}
Let $\psi$, $\phi$ be proof terms such that $\psi \peqe \phi$ and $\psi$, $\phi$ respect $\set{\epsilon}$.
Let $\cfps{\psi} = f(\psi_1, \ldots, \psi_m)$ and $\cfps{\phi} = f'(\phi_1, \ldots, \phi_{m'})$.
Then $f = f' = src(\psi)(\epsilon)$, so that $m = m'$, and $\psi_i \peqe \phi_i$ for each $i$ between 1 and $m$.
\end{lemma}

\begin{proof}
Lem.~\ref{rsl:cfps-peqe} and the hypotheses imply $\psi \peqe \cfps{\psi} \peqe \cfps{\phi}$, then Lem.~\ref{rsl:peq-then-same-src-mind-tgt} yields $f = f' = src(\psi)(\epsilon)$, and therefore $m = m'$.
We prove $\psi_i \peqe \phi_i$ for all $i$ by induction on $\alpha$ where $\psi \layerpeqe{\alpha} \phi$, analysing the rule used in the last step of that judgement.

\begin{itemize}
\item 
\eqlrefl: we conclude immediately.

\item 
\eqleqn: we analyse each of the equations.
	\begin{itemize}
	\item 
	\peqidleft: let $src(\phi) = f(t_1, \ldots, t_m)$ where $t_i = src(\phi_i)$ for all $i$; \confer\ Lem.~\ref{rsl:cfps-peqe} and Lem.~\ref{rsl:peq-then-same-src-mind-tgt}. 
	Then $\psi = f(t_1, \ldots, t_m) \comp \phi$, so that $\cfps{\psi} = f(t_1 \comp \phi_1, \ldots, t_m \comp \phi_m)$. Thus we conclude.

	\item
	\peqidright: an analogous argument applies.

	\item
	\peqassoc: in this case $\psi = \xi \comp (\gamma \comp \chi)$ and $\phi = (\xi \comp \gamma) \comp \chi$.
	Let $\cfps{\xi} = f(\xi_1, \ldots, \xi_m)$, $\cfps{\gamma} = f(\gamma_1, \ldots, \gamma_m)$ and $\cfps{\chi} = f(\chi_1, \ldots, \chi_m)$; \confer\ Lem.~\ref{rsl:cfps-peqe} (implying $f = src(\psi)(\epsilon) = src(\xi)(\epsilon) = src(\cfps{\xi})(\epsilon)$) and Lem.~\ref{rsl:respects-epsilon-then-stable-root}.
	Then for any $i \leq m$, $\psi_i = \xi_i \comp (\gamma_i \comp \xi_i)$ and $\phi_i = (\xi_i \comp \gamma_i) \comp \chi_i$. Thus we conclude immediately.

	\item
	\peqstruct\ and \peqinfstruct: in either of these cases Dfn.~\ref{dfn:cfps} allows to conclude immediately.

	\item
	\peqoutin\ and \peqinout: either of these cases contradict $\psi, \phi$ to respect $\set{\epsilon}$.
	\end{itemize}

\item
\eqlsymm\ or \eqltrans: a simple inductive argument applies.

\item
\eqlfun: the hypotheses of the \eqlfun\ rule are enough to conclude immediately.

\item
\eqlrule: this case would imply that neither $\psi$ nor $\phi$ respect $\set{\epsilon}$, thus contradicting lemma hypotheses.

\item
\eqlcomp: in this case, 
$\psi = \chi \comp \xi$, $\phi = \gamma \comp \delta$, $\chi \layerpeqe{\alpha_1} \gamma$, $\xi \layerpeqe{\alpha_2} \delta$, $\alpha_1 < \alpha$ and $\alpha_2 < \alpha$.
Let $\cfps{\chi} = f(\chi_1, \ldots, \chi_m)$, $\cfps{\xi} = f(\xi_1, \ldots, \xi_m)$, $\cfps{\gamma} = f(\gamma_1, \ldots, \gamma_m)$ and $\cfps{\delta} = f(\delta_1, \ldots, \delta_m)$.
Let $i$ such that $1 \leq i \leq m$.
Observe $\psi_i = \chi_i \comp \xi_i$ and $\phi_i = \gamma_i \comp \delta_i$.
On the other hand, \ih\ implies $\chi_i \peqe \gamma_i$ and $\xi_i \peqe \delta_i$.
Thus we conclude.

\item
\eqlinfcomp: an analogous argument applies. 
In this case, $\psi = \icomp \psi_i$, $\phi = \icomp \phi_i$, and for any $i < \omega$, $\psi_i \layerpeqe{\alpha_i} \phi_i$ where $\alpha_i < \alpha$.
Let $\cfps{\psi_i} = f(\psi^1_i, \ldots, \psi^m_i)$ and $\cfps{\phi_i} = f(\phi^1_i, \ldots, \phi^m_i)$.
Let $j$ such that $1 \leq j \leq m$.
Then $\psi_j = \icomp \psi^j_i$ and $\phi_j = \icomp \phi^j_i$.
\Ih\ on each $\psi_i \layerpeqe{\alpha_i} \phi_i$ yields $\psi^j_i \peqe \phi^j_i$. Thus we conclude.
\end{itemize}
\end{proof}

\begin{lemma}
\label{rsl:cfps-equals-src-tgt}
Let $\psi$ be a proof term such that $\psi$ respects $\set{\epsilon}$.
Then $\cfps{\psi}(\epsilon) = src(\psi)(\epsilon) = tgt(\psi)(\epsilon)$.
\end{lemma}

\begin{proof}
Immediate consequence of Lem.~\ref{rsl:peqe-then-cfps-proj-peqe} and Lem.~\ref{rsl:respects-epsilon-then-stable-root}.
\end{proof}

\begin{lemma}
\label{rsl:cfps-components-respect}
Let $\psi$ be a proof term and $\Spa$ a set of positions such that $\Spa \neq \emptyset$ and $\psi$ respects $\Spa$. Then $\psi_i$ respects $\proj{\Spa}{i}$ for all $i \leq m$, where $\cfps{\psi} = f(\psi_1, \ldots, \psi_m)$.
\end{lemma}

\begin{proof}
Lem.~\ref{rsl:cfps-peqe} implies $\psi \peqe \cfps{\psi}$, then Lem.~\ref{rsl:respects-invariant-peqe} implies $\cfps{\psi}$ respects $\Spa$.
Therefore %
\cfpInsideStandardisation{Lem.~\ref{rsl:respects-invariant-proj} }%
\cfpInsideCompression{Lem.~\ref{rsl:respect-fnsymbol-coherence} }%
allows to conclude.
\end{proof}

\begin{lemma}
\label{rsl:respects-then-invariant-fp}
Let $\psi$ be a convergent proof term and $\Spa$ a set of positions such that $\psi$ respects $\Spa$.
Then $\pref{tgt(\psi)}{\Spa} = \pref{src(\psi)}{\Spa}$.
\end{lemma}

\begin{proof}
We proceed by induction on the cardinal of $\Spa$.
If $\Spa = \emptyset$, then $\pref{tgt(\psi)}{\Spa} = \pref{src(\psi)}{\Spa} = \Box$.
Otherwise, $\Spa = \set{\epsilon} \cup (\bigcup_{1 \leq i \leq m} i \cdot \proj{\Spa}{i})$ where $\cfps{\psi} = f(\psi_1, \ldots, \psi_m)$.
Lem.~\ref{rsl:cfps-peqe} and Lem.~\ref{rsl:peq-then-same-src-mind-tgt} imply $src(\psi) = f(src(\psi_1), \ldots, src(\psi_m))$ and $tgt(\psi) = f(tgt(\psi_1), \ldots, tgt(\psi_m))$, so that $\pref{src(\psi)}{\Spa} = f(\pref{src(\psi_1)}{\proj{\Spa}{1}}, \ldots, \pref{src(\psi_m)}{\proj{\Spa}{m}})$, and $\pref{tgt(\psi)}{\Spa} = f(\pref{tgt(\psi_1)}{\proj{\Spa}{1}}, \ldots, \pref{tgt(\psi_m)}{\proj{\Spa}{m}})$.
On the other hand, Lem.~\ref{rsl:cfps-components-respect} implies that $\psi_i$ respects $\proj{\Spa}{i}$ for all $i$, so that \ih\ can be applied to obtain $\pref{src(\psi_i)}{\proj{\Spa}{i}} = \pref{tgt(\psi_i)}{\proj{\Spa}{i}}$.
Thus we conclude.
\end{proof}

\medskip
Assume that some proof term, say $\psi$, respects not only the root, but a finite, prefix-closed set of positions $\Spa$.
Then we can define the \emph{condensed-to-fixed-prefix-\textbf{context} form} of $\psi$ \wrt\ $\Spa$, analogously as we have just done with the condensed-to-fixed-prefix-symbol form. 
The activity denoted by a condensed-to-fixed-prefix-context form \wrt\ the set of positions $\Spa$ will lie inside a fixed context, \ie\ a context in $\SigmaTerms$, whose set of (non-hole) positions is exactly $\Spa$.
\Eg, the proof term $h(f(g(\mu(a))), \mu(b)) \comp h(f(g(g(\pi))), \nu(b))$ respects $\Spa \eqdef \set{\epsilon,1,11}$. The corresponding condensed-to-fixed-prefix-context is $h(f(g(\mu(a) \comp g(\pi))), \mu(b) \comp \nu(b))$. Observe that the activity of the latter term lies inside the holes of the context $h(f(g(\Box)), \Box)$, whose set of non-hole positions is $\Spa$.

The condensed-to-fixed-prefix-context form of $\psi$ \wrt\ $\Spa$ can be defined in two different ways: either by induction on $\psi$ analogously as the definition of $\cfpsfn$, or by induction on $\Spa$. 
The following definition uses the latter option for a pragmatic reason: it leads to simpler proofs of the properties to be stated about these forms.

\begin{definition}
\label{dfn:cfpc}
Let $\psi$ be a proof term and $\Spa$ a prefix-closed set of positions, such that $\psi$ respects $\Spa$. 
We define $\cfpc{\psi}{\Spa}$, the \emph{condensed to fixed prefix context form} of $\psi$ \wrt\ $\Spa$, as follows. \\
If $\Spa = \emptyset$, then $\cfpc{\psi}{\Spa} \eqdef \psi$. \\
Otherwise, $\Spa = \set{\epsilon} \,\cup\, ( \bigcup_{1 \leq i \leq m} i \cdot \proj{\Spa}{i} )$, where $src(\psi)(\epsilon) = f/m$.
In this case $\cfpc{\psi}{\Spa} \eqdef f(\cfpc{\psi_1}{\proj{\Spa}{1}}, \ldots \cfpc{\psi_m}{\proj{\Spa}{m}})$, where 
$\cfps{\psi} = f(\psi_1, \ldots, \psi_m)$. 
\end{definition}

\begin{lemma}
\label{rsl:cfpc-peqe}
Let $\psi$, $\Spa$ such that $\psi$ respects $\Spa$.
Then $\psi \peqe \cfpc{\psi}{\Spa}$.
\end{lemma}

\begin{proof}
We proceed by induction on the cardinal of $\Spa$. 
If $\Spa = \emptyset$ then we conclude immediately. 
Otherwise, $\Spa = \set{\epsilon} \,\cup\, ( \bigcup_{1 \leq i \leq m} i \cdot \proj{\Spa}{i} )$ where $\cfps{\psi} = f(\psi_1, \ldots, \psi_m)$, and $\cfpc{\psi}{\Spa} = f(\cfpc{\psi_1}{\proj{\Spa}{1}}, \ldots, \cfpc{\psi_m}{\proj{\Spa}{m}})$.
Lem.~\ref{rsl:cfps-components-respect} implies that $\psi_i$ respects $\proj{\Spa}{i}$ for all $i \leq m$.
Therefore \ih\ can be applied on each $\proj{\Spa}{i}$ to obtain $\psi_i \peqe \cfpc{\psi_i}{\proj{\Spa}{i}}$, so that \eqlfun\ rule yields $\cfps{\psi} \peqe \cfpc{\psi}{\Spa}$. On the other hand, Lem.~\ref{rsl:cfps-peqe} implies $\psi \peq \cfps{\psi}$. Thus we conclude by \eqltrans.
\end{proof}

\begin{lemma}
\label{rsl:peqe-then-fp-args-peqe}
Let $\psi$, $\phi$, $\Spa$ such that $\psi$ and $\phi$ are convergent, $\psi \peqe \phi$ and $\psi$, $\phi$ respect $\Spa$.
Then $\cfpc{\psi}{\Spa} = C[\psi_1, \ldots, \psi_k]$, $\cfpc{\phi}{\Spa} = C[\phi_1, \ldots, \phi_k]$ and $\psi_i \peqe \phi_i$ for all $i$, where $C = \pref{src(\psi)}{\Spa}$.
\end{lemma}

\begin{proof}
We proceed by induction on the cardinal of $\Spa$. If $\Spa = \emptyset$ then we conclude immediately. Otherwise $\Spa = \set{\epsilon} \cup (\bigcup_{1 \leq i \leq m} i \cdot \proj{\Spa}{i})$, 
$\cfpc{\psi}{\Spa} = f(\cfpc{\psi'_1}{\proj{\Spa}{1}}, $ \\ $\ldots, \cfpc{\psi'_m}{\proj{\Spa}{m}})$, and 
$\cfpc{\phi}{\Spa} = f(\cfpc{\phi'_1}{\proj{\Spa}{1}}, \ldots, \cfpc{\phi'_m}{\proj{\Spa}{m}})$, where 
$\cfps{\psi} = f(\psi'_1, \ldots, \psi'_m)$ and $\cfps{\phi} = f(\phi'_1, \ldots, \phi'_m)$.
Lem.~\ref{rsl:cfps-peqe} and Lem.~\ref{rsl:peq-then-same-src-mind-tgt} imply that $src(\psi) = f(src(\psi'_1), \ldots, src(\psi'_m))$ and analogously for $\phi$, so that particularly the root symbols of $\cfps{\psi}$ and $\cfps{\phi}$ coincide since $\psi \peqe \phi$.

Let $j$ such that $1 \leq j \leq m$. 
Lem.~\ref{rsl:peqe-then-cfps-proj-peqe} implies that $\psi'_j \peqe \phi'_j$, and Lem.~\ref{rsl:cfps-components-respect} implies that both $\psi'_j$ and $\phi'_j$ respect $\proj{\Spa}{j}$.
Then we can apply \ih\ on $\proj{\Spa}{j}$ obtaining that 
$\cfpc{\psi'_j}{\proj{\Spa}{j}} = C_j[\psi^j_1, \ldots, \psi^j_{q_j}]$, $\cfpc{\phi'_j}{\proj{\Spa}{j}} = C_j[\phi^j_1, \ldots, \phi^j_{q_j}]$ and $\psi^j_i \peqe \phi^j_i$ for all $i$, where $\pref{src(\psi'_j)}{\proj{\Spa}{j}} = C_j$.

We define $C \eqdef f(C_1, \ldots, C_m)$. It is straightforward to verify that $\pref{src(\psi)}{\Spa} = C$. Moreover, $\cfpc{\psi}{\Spa} = C[\psi_1, \ldots, \psi_k]$ and $\cfpc{\phi}{\Spa} = C[\phi_1 \ldots, \phi_k]$, where $k = \sum_{1 \leq i \leq m} q_i$, and for any $i \leq k$, $\psi_i = \psi^j_l$ and $\phi_i = \phi^j_l$ for some $j \leq m$ and $l \leq q_j$, implying $\psi_i \peqe \phi_i$. Thus we conclude.
\end{proof}

\begin{lemma}
\label{rsl:cfpc-equals-src-tgt}
Let $\psi$, $\Spa$ such that $\psi$ respects $\Spa$.
Then $\pref{\cfpc{\psi}{\Spa}}{\Spa} = \pref{src(\psi)}{\Spa} = \pref{tgt(\psi)}{\Spa}$.
\end{lemma}

\begin{proof}
Straightforward corollary of Lem.~\ref{rsl:peqe-then-fp-args-peqe} and Lem.~\ref{rsl:respects-then-invariant-fp}.
\end{proof}

\doNotIncludeStandardisation{\subsection{General factorisation result}}
\label{scn:factorisation-general}

In this section we will extend the factorisation result obtained for \imsteps\ in Sec.~\ref{scn:factorisation-imsteps}, to the set of all proof terms.
\doNotIncludeStandardisation{As we have already mentioned, the }%
\includeStandardisation{The }%
condensed-to-proof-term forms introduced in Sec.~\ref{sec:cfpc} lead to the proof of the main remaining auxiliary result, namely, the ability of obtain proof terms in which activity at lower depths is in low positions \wrt\ the sequentialisation order given by dot occurrences.

\begin{lemma}
\label{rsl:jump-one-step}
Let $\psi$ be a one-step. 
Then there exist two numbers $n, n' < \omega$ such that, for any convergent proof term $\xi$ verifying $tgt(\xi) = src(\psi)$ and $\mind{\xi} \geq n + n'$, a one-step $\psi'$ and a convergent proof term $\xi'$ can be found, which verify all the following: $\xi \comp \psi \peqe \psi' \comp \xi'$, $\sdepth{\psi'} = \sdepth{\psi}$, and $\mind{\xi'} \geq \mind{\xi} - n'$.
\end{lemma}

\begin{proof}
We take $n \eqdef \sdepth{\psi}$ and $n' = \Pdepth{\mu} + 1$ where $\mu \eqdef \psi(\RPos{\psi})$.
We consider a convergent proof term $\xi$ verifying $\mind{\xi} \geq n + n'$ and $tgt(\xi) = src(\psi)$.

\cfpInsideCompression{%
Let $\Spa_0 \eqdef \set{p \setsthat p \in src(\psi) \land \posln{p} < \sdepth{\psi}}$, $\Spa \eqdef \Spa_0 \cup (\RPos{\psi} \cdot \PPos{\mu})$, and $k \eqdef max \set{\posln{p} \setsthat p \in \Spa}$.
Observe that $p \in \Spa$ implies $\posln{p} \leq \sdepth{\psi} + \Pdepth{\mu}$, so that $k \leq \sdepth{\psi} + \Pdepth{\mu} < \mind{\xi}$.
}%
\cfpInsideStandardisation{%
Let $\Spa_0 \eqdef \set{p \setsthat p \in src(\psi) \land \posln{p} < \sdepth{\psi}}$ and $\Spa \eqdef \Spa_0 \cup (\RPos{\psi} \cdot \PPos{\mu})$.
Observe that $p \in \Spa$ implies $\posln{p} \leq \sdepth{\psi} + \Pdepth{\mu} < \mind{\xi}$, so that Lem.~\ref{rsl:frso-defined-then-length-geq-mind} implies $\frso{\xi}{p}$ to be undefined.
}%
Moreover, it is straightforward to verify that $\Spa \subseteq \Pos{src(\psi)} = \Pos{tgt(\xi)}$.
\cfpInsideCompression{%
Therefore Lem.~\ref{rsl:mind-plus-tgt-then-respects} applies \wrt\ $\xi$, $\Spa$ and $k$, implying that $\xi$ respects $\Spa$. 
}%
\cfpInsideStandardisation{%
Therefore Lem.~\ref{rsl:respects-pos-tgt} implies that $\xi$ respects $\Spa$. 
}%
Then $\xi_F \eqdef \cfpc{\xi}{\Spa}$ can be defined. 
In turn, Lem.~\ref{rsl:cfpc-peqe} implies that $\xi \peqe \xi_F$, so that $\xi \comp \psi \peqe \xi_F \comp \psi$, and Lem.~\ref{rsl:cfpc-equals-src-tgt} implies $\pref{\xi_F}{\Spa} = \pref{tgt(\xi)}{\Spa} = \pref{src(\psi)}{\Spa}$.

Let $C \eqdef \pref{src(\psi)}{\Spa_0}$. An easy induction on $\sdepth{\psi}$ yields that $\pref{\psi}{\Spa_0} = C$, so that the comment following Dfn.~\ref{dfn:term-prefix} implies $\psi = C[t_1, \ldots, t_{j-1}, \mu(u_1, \ldots, u_m), t_{j+1}, \ldots, t_k]$ and $\set{p \setsthat p \in \Pos{C} \,\land\, C(p) \neq \Box} = \Spa_0$. Observe that $\posln{\BPos{C}{i}} = \sdepth{\psi}$ for all $i$, and that particularly $\BPos{C}{j} = \RPos{\psi}$ for some $j$.
In turn, the given form of $\psi$ implies that $src(\psi) = tgt(\xi) = C[t_1, \ldots, t_{j-1}, l[u_1, \ldots, u_m], t_{j+1}, \ldots, t_k]$ where $\mu: l \to h$.
Observe that the set of non-hole positions of the context \\ $C[\Box, \ldots, \Box, l[\Box, \ldots, \Box], \Box, \ldots, \Box]$ is exactly $\Spa$, implying that $C = \pref{tgt(\xi)}{\Spa} = \pref{\xi_F}{\Spa}$, and therefore
$\xi_F = C[\xi_1, \ldots, \xi_{j-1}, l[\phi_1, \ldots, \phi_m], \xi_{j+1}, \ldots, \xi_k]$; \confer\ the comment following Dfn.~\ref{dfn:term-prefix}.
Notice that $\xi_F$ is convergent, implying that all the $\xi_i$ and also the $\phi_i$ are; \confer\ Lem.~\ref{rsl:peqe-soundness-convergent} and Lem.~\ref{rsl:ctx-convergence}.
Moreover, $t_i = tgt(\xi_i)$ for any suitable $i$, and also $u_i = tgt(\phi_i)$ for all suitable $i$. Hence \\
$\begin{array}{@{}rl}
\xi_F \comp \psi 
\\ \peqe & 
  C[\xi_1 \comp t_1, \ldots, \xi_{j-1} \comp t_{j-1}, l[\phi_1 \ldots \phi_m] \comp \mu(u_1, \ldots, u_m), \xi_{j+1} \comp t_{j+1}, \ldots, \xi_k \comp t_k]
\\ \peqe & 
  C[\xi_1, \ldots, \xi_{j-1}, \mu(\phi_1, \ldots, \phi_m), \xi_{j+1}, \ldots, \xi_k]
\\ \peqe & 
	C[s_1 \comp \xi_1, \ldots, s_{j-1} \comp \xi_{j-1}, \mu(w_1, \ldots, w_m) \comp h[\phi_1, \ldots, \phi_m], s_{j+1} \comp \xi_{j+1}, \ldots, s_k \comp \xi_k]
\\ \peqe & 
  C[s_1, \ldots, s_{j-1}, \mu(w_1, \ldots, w_m), s_{j+1}, \ldots, s_k] \ \comp \\ &
	C[\xi_1, \ldots, \xi_{j-1}, h[\phi_1, \ldots, \phi_m], \xi_{j+1}, \ldots, \xi_k]
\end{array}$ \\[2pt]
where $s_i \eqdef src(\xi_i)$ and $w_i \eqdef src(\phi_i)$, in both cases for all suitable $i$.
To justify the equivalences; \confer\ Lem.~\ref{rsl:struct-ctx}; \peqidright, \peqinout\ and Lem.~\ref{rsl:peqe-compatible-ctx}; \peqidleft, \peqoutin\ and Lem.~\ref{rsl:peqe-compatible-ctx} again; and finally Lem.~\ref{rsl:struct-ctx} again; respectively.

We take $\psi' \eqdef C[s_1, \ldots, s_{j-1}, \mu(w_1, \ldots, w_m), s_{j+1}, \ldots, s_k]$ and \\ $\xi' \eqdef C[\xi_1, \ldots, \xi_{j-1}, h[\phi_1, \ldots, \phi_m], \xi_{j+1}, \ldots, \xi_k]$.
Observe that convergence of all $\xi_i$ and $\phi_i$ imply convergence of $\xi'$, \confer\ Lem.~\ref{rsl:ctx-convergence}.

In order to conclude, we must verify that $\mind{\xi'} \geq \mind{\xi} - n' = \mind{\xi_F} - (\Pdepth{\mu} + 1)$; \confer\ Lem.~\ref{rsl:peq-then-same-src-mind-tgt}.
Let $a$ such that $\mind{\xi_a} \leq \mind{\xi_i}$ for all $i$ such that $1 \leq i \leq k$ and $i \neq j$, 
$b$ such that $\mind{\phi_b} + \posln{\BPos{l}{b}} \leq \mind{\phi_i} + \posln{\BPos{l}{i}}$ for all $i$ such that $1 \leq i \leq m$, and
$c$, $k$ such that $\mind{\phi_c} + \posln{\BPos{h}{k}} \leq \mind{\phi_i} + \posln{\BPos{h}{j}}$ if $1 \leq i \leq m$ and $h(\BPos{h}{j}) = x_i$.
In these definitions, $l$ and $h$ are considered as contexts as when we write \eg\ $l[\phi_1, \ldots, \phi_m]$. 
Lem.~\ref{rsl:mind-ctx} implies $\mind{\xi_F} = \sdepth{\psi} + min(\mind{\xi_a}, \mind{\phi_b} + \posln{\BPos{l}{b}})$ and $\mind{\xi'} = \sdepth{\psi} + min(\mind{\xi_a}, \mind{\phi_c} + \posln{\BPos{h}{k}})$.
Observe that $\posln{\BPos{l}{i}} \leq \Pdepth{\mu} + 1$ for all $i$.
We show $\mind{\xi_F} - (\Pdepth{\mu} + 1) \leq \mind{\xi'}$.

If $\mind{\xi_a} \leq \mind{\phi_c} + \posln{\BPos{h}{k}}$, then
$\mind{\xi_F} \leq \sdepth{\psi} + \mind{\xi_a} = \mind{\xi'}$ in either case \wrt\ the characterisation of $\mind{\xi_F}$.
Otherwise, \ie\ if $\mind{\phi_c} + \posln{\BPos{h}{k}} < \mind{\xi_a}$, 
observe that $\mind{\xi_F} \leq \sdepth{\psi} + \mind{\phi_b} + \posln{\BPos{l}{b}}$ holds in any case. Therefore \\ 
$\begin{array}{rcl}
	\mind{\xi_F} & \leq &
	\sdepth{\psi} + \mind{\phi_b} + \posln{\BPos{l}{b}} \\ & \leq &
	\sdepth{\psi} + \mind{\phi_c} + \posln{\BPos{l}{c}} \\ & \leq &
	\sdepth{\psi} + \mind{\phi_c} + (\Pdepth{\mu} + 1)
\end{array}$ \\
Therefore $\mind{\xi_F} - (\Pdepth{\mu} + 1) \leq 
	\sdepth{\psi} + \mind{\phi_c} \leq 
	\sdepth{\psi} + \mind{\phi_c} + \posln{\BPos{h}{k}} = \mind{\xi'}$.
\end{proof}

\begin{lemma}
\label{rsl:jump-ppterm}
Let $\psi$ be a finite \pnpterm. 
Then there exist two numbers $n, n' < \omega$ such that, for any convergent proof term $\xi$ verifying $tgt(\xi) = src(\psi)$ and $\mind{\xi} \geq n + n'$, a finite \pnpterm\ $\psi'$ and a convergent proof term $\xi'$ can be found, which verify all the following: $\xi \comp \psi \peqe \psi' \comp \xi'$, $\ppsteps{\psi'} = \ppsteps{\psi}$, $\sdepth{\redel{\psi'}{i}} = \sdepth{\redel{\psi}{i}}$ for all $i$, and $\mind{\xi'} \geq \mind{\xi} - n' \geq n$.
\end{lemma}

\begin{proof}
We proceed by induction on $\ppsteps{\psi}$.
If $\ppsteps{\psi} = 0$, \ie\ $\psi \in \iSigmaTerms$, then $src(\psi) = \psi$.
Therefore we can take $n = n' = 0$, since for any $\xi$ verifying $tgt(\xi) = \psi$, it is straightforward to obtain $\xi \comp \psi \peqe src(\xi) \comp \xi$, and to verify the required properties for $\psi' \eqdef src(\xi)$ and $\xi' \eqdef \xi$ .

Assume $\ppsteps{\psi} = n+1$, \ie\ $\psi = \chi \comp \phi$, where $\chi$ is a one-step and $\phi$ is a \pnpterm\ verifying $\ppsteps{\phi} = n$.
In this case, \ih\ can be applied on $\phi$; let $m$ and $m'$ be the corresponding numbers. Moreover, Lem.~\ref{rsl:jump-one-step} applies to $\chi$; let $p$ and $p'$ be the numbers whose existence is stated by that lemma.
Let $n \eqdef max(m,p)$ and $n' \eqdef m' + p'$.
Let $\xi$ a convergent proof term verifying $\mind{\xi} \geq n + n' = n + m' + p' \geq p + p'$, and $tgt(\xi) = src(\psi) = src(\chi)$.
Then the conclusion of Lem.~\ref{rsl:jump-one-step} implies that $\xi \comp \psi = \xi \comp \chi \comp \phi \peqe \chi' \comp \xi'' \comp \phi$, where $\chi'$ is a one-step verifying $\sdepth{\chi'} = \sdepth{\chi}$ and $\xi''$ is a convergent proof term such that $\mind{\xi''} \geq \mind{\xi} - p' \geq n + m' \geq m + m'$.
In turn, the conclusion of the \ih\ implies that $\chi' \comp \xi'' \comp \phi \peqe \chi' \comp \phi' \comp \xi'$, where $\phi'$ is a \pnpterm\ verifying $\ppsteps{\phi'} = \ppsteps{\phi}$ and $\sdepth{\redel{\phi'}{i}} = \sdepth{\redel{\phi}{i}}$ for all $i$, and $\xi'$ is a convergent proof term such that $\mind{\xi'} \geq \mind{\xi''} - m' \geq n$. We take $\psi' \eqdef \chi' \comp \phi'$, and we conclude by observing that \eqltrans\ implies $\xi \comp \psi \peqe \psi' \comp \xi'$.
\end{proof}

\medskip
The given auxiliary results allow to prove the statement being the aim of this Section.

\begin{proposition}
\label{rsl:factorisation}
Let $\psi$ be a convergent proof term and $n < \omega$.
Then there exist $\chi$ and $\phi$ such that $\psi \peqe \chi \comp \phi$, $\chi$ is a finite \pnpterm, $\phi$ is convergent and $mind(\phi) > n$.
\end{proposition}

\begin{proof}
We proceed by induction on $\alpha$ where $\psi \in \layerpterm{\alpha}$, analysing the cases in the formation of $\psi$ \wrt\ Dfn.~\ref{dfn:layer-pterm}.

\begin{itemize}
\item 
Assume that $\psi$ is an \imstep. In this case we proceed by induction on $n$.
If $n = 0$ then Lem.~\ref{rsl:factorisation-imstep} suffices to conclude.

Assume $n = n' + 1$. 
Lem.~\ref{rsl:factorisation-imstep} implies $\psi \peqe \chi_0 \comp \phi'$ where $\chi_0$ is a finite \pnpterm, $\phi'$ is a convergent \imstep\ and $\mind{\phi'} > 0$, so that $\phi' = f(\phi'_1, \ldots, \phi'_m)$.
Observe that $\phi'$ convergent implies $\phi'_i$ convergent for all $i$, \confer\ Lem.~\ref{rsl:imstep-convergent-args}.
Then \ih\ can be applied on all $\phi'_i$ \wrt\ $n'$, yielding 
$\phi' \peqe f(\chi_1 \comp \phi_1, \ldots, \chi_m \comp \phi_m)$ where for all $i$, $\chi_i$ is a finite \pnpterm, $\phi_i$ is convergent and $\mind{\phi_i} > n'$.
Hence $\psi \peqe \chi_0 \comp f(\chi_1, \ldots, \chi_m) \comp f(\phi_1, \ldots, \phi_m)$.

Assume that $m = 3$; observe that 
$f(\chi_1, \chi_2, \chi_3) 
   \peqe f(\chi_1 \comp t_1, s_2 \comp \chi_2, s_3 \comp \chi_3)
	 \peqe f(\chi_1, s_2, s_3) \comp f(t_1, \chi_2, \chi_3)
	 \peqe f(\chi_1, s_2, s_3) \comp f(t_1 \comp t_1, \chi_2 \comp t_2, s_3 \comp \chi_3)
	 \peqe f(\chi_1, s_2, s_3) \comp f(t_1, \chi_2, s_3) \comp (t_1, t_2, \chi_3)$.
An analogous reasoning for any $m$ yileds 
$f(\chi_1, \chi_2, \ldots, \chi_m)
	 \peqe f(\chi_1, src(\chi_2), \ldots, src(\chi_m)) 
	     \comp f(tgt(\chi_1), \chi_2, \ldots, src(\chi_m))
	     \comp f(tgt(\chi_1), tgt(\chi_2), \ldots, \chi_m)$.
In turn, it is straightforward to obtain a \ppterm\ $\chi'_k \peqe f(tgt(\chi_1), \ldots, \chi_k, \ldots, src(\chi_m))$, so that $\chi' \eqdef \chi'_0 \comp \ldots \comp \chi'_m$ is a \ppterm\ verifying $\chi' \peqe f(\chi_1, \chi_2, \ldots, \chi_m)$.
Thus we conclude by taking $\chi \eqdef \chi_0 \comp \chi'$ and $\phi \eqdef f(\phi_1, \ldots, \phi_m)$.

\item
Assume $\psi = \psi_1 \comp \psi_2$ and $\psi$ is not an infinite composition.
In this case we can apply \ih\ on $\psi_2$, obtaining $\psi_2 \peqe \chi_2 \comp \phi_2$ where $\chi_2$ is a finite \pnpterm, $\phi_2$ is convergent and $\mind{\phi_2} > n$.
Lem.~\ref{rsl:jump-ppterm} applies to $\chi_2$, implying the existence of two numbers, say $m_0$ and $m'$, which enjoy some properties.
Let $m \eqdef max(n,m_0)$. 
Applying \ih\ on $\psi_1$ \uln{\wrt\ $m + m'$}, we obtain $\psi_1 \peqe \chi_1 \comp \phi_1$, where $\chi_1$ is a finite \pnpterm, $\phi_1$ is convergent and $\mind{\phi_1} > m + m' \geq m_0 + m'$. Observe $\psi \peqe \chi_1 \comp \phi_1 \comp \chi_2 \comp \phi_2$, so that $tgt(\phi_1) = src(\chi_2)$.

Therefore, the conclusion of Lem.~\ref{rsl:jump-ppterm} implies $\phi_1 \comp \chi_2 \peqe \chi'_2 \comp \phi'_1$, so that $\psi \peqe \chi_1 \comp \chi'_2 \comp \phi'_1 \comp \phi_2$, where $\chi'_2$ is a finite \pnpterm\ (since $\ppsteps{\chi'_2} = \ppsteps{\chi_2}$), $\phi'_1$ is convergent and $\mind{\phi'_1} \geq \mind{\phi_1} - m' > m \geq n$.
Thus we conclude by taking $\chi \eqdef \chi_1 \comp \chi'_2$ and $\phi \eqdef \phi'_1 \comp \phi_2$.

\item
Assume $\psi = \icomp \psi_i$.
Let $k$ such that $\mind{\psi_i} > n$ if $i > k$; convergence of $\psi$ entails the existence of such $k$.
Then $\psi \peqe \psi_0 \comp \ldots \comp \psi_k \comp (\icomp \psi_{k+1+i})$, and $\mind{\icomp \psi_{k+1+i}} > n$; notice that convergence of $\psi$ implies convergence of $\icomp \psi_{k+1+i}$.
Observe that $\psi_0 \comp \ldots \comp \psi_k \in \layerpterm{\alpha'}$ where $\alpha' < \alpha$. This observation allows to use \ih\ to obtain $\psi_0 \comp \ldots \comp \psi_k \peqe \chi \comp \phi'$ where $\chi$ is a finite \pnpterm, $\phi'$ is convergent and $\mind{\phi'} > n$. Then we conclude by taking $\phi \eqdef \phi' \comp (\icomp \psi_{k+1+i})$.

\item
Assume $\psi = f(\psi_1, \ldots, \psi_m)$ and $\psi$ is not an \imstep.
In this case, we can apply \ih\ on each $\psi_i$ obtaining $\psi_i \peqe \chi_i \comp \phi_i$, where $\chi_i$ is a finite \pnpterm, $\phi_i$ is convergent, and $\mind{\phi_i} > n$.
Then $\psi \peqe f(\chi_1, \ldots, \chi_m) \comp f(\phi_1, \ldots, \phi_m)$.
Hence, an argument about $f(\chi_1, \ldots, \chi_m)$ analogous to that used in the \imstep\ case allows to conclude.

\item
Assume $\psi = \mu(\psi_1, \ldots, \psi_m)$ and $\psi$ is not an \imstep.
Say $\mu: l[x_1, \ldots, x_m] \to h$.

Assume $h = f(h_1, \ldots, h_k)$.
In this case $\psi \peqe \mu(src(\psi_1), \ldots, src(\psi_m)) \comp f(h_1[\psi_1, \ldots, \psi_m], \ldots, h_k[\psi_1, \ldots, \psi_m])$.
Applying \ih\ on each $\psi_i$ yields $\psi_i \peqe \chi_i \comp \phi_i$, where $\chi_i$ is a finite \pnpterm, $\phi_i$ is convergent, and $\mind{\phi_i} > n$.  \\
Therefore
$\psi \peqe 
	\mu(src(\psi_1), \ldots, src(\psi_m)) \comp 
	f(h_1[\chi_1, \ldots, \chi_m], \ldots, h_k[\chi_1, \ldots, \chi_m]) 
	\comp f(h_1[\phi_1, \ldots, \phi_m], \ldots, h_k[\phi_1, \ldots, \phi_m])$; \confer\ Lem~\ref{rsl:struct-ctx}.
Hence, an argument about $f(h_1[\chi_1, \ldots, \chi_m], \ldots, h_k[\chi_1, \ldots, \chi_m])$ analogous to that used in the \imstep\ case for $f(\chi_1, \ldots, \chi_m)$, \confer\ Lem.~\ref{rsl:struct-ctx}, allows to conclude.

The other possible case is $h = x_j$, implying $\psi \peqe \mu(src(\psi_1), \ldots, src(\psi_m)) \comp \psi_j$. \Ih\ can be applied on $\psi_j$ obtaining $\psi_j \peqe \chi' \comp \phi$, where $\chi'$ is a finite \pnpterm, $\phi$ is convergent and $\mind{\phi} > n$.  
Thus we conclude by taking $\chi \eqdef \mu(src(\psi_1), \ldots, src(\psi_m)) \comp \chi'$.
\end{itemize}
\end{proof}

\doNotIncludeStandardisation{
\subsection{Proof of the compression result}
}

\begin{theorem}
\label{rsl:compression}
Let $\psi$ be a convergent proof term. Then there exists some \ppterm\ $\phi$ verifying $\psi \peq \phi$ and $\ppsteps{\phi} \leq \omega$.
\end{theorem}

\begin{proof}
We define the sequences of proof terms $\langle \psi_i \rangle_{i < \omega}$ and $\langle \phi_i \rangle_{i < \omega}$ as follows.
We start defining $\psi_0 \eqdef \psi$.
Then, for each $i < \omega$, we define $\phi_i$ and $\psi_{i+1}$ to be proof terms verifying that $\psi_i \peqe \phi_i \comp \psi_{i+1}$, $\phi_i$ is a finite \pnpterm\ and either $\mind{\psi_{i+1}} > \mind{\psi_i}$ or $\mind{\psi_{i+1}} = \mind{\psi_i} = \omega$; \confer\ Prop.~\ref{rsl:factorisation}.
Observe that $\mind{\psi_i} < \omega$ implies $\mind{\phi_i} = \mind{\psi_i}$ by \ref{rsl:peq-then-same-src-mind-tgt}, so in that case $\phi_i$ is a \ppterm, \ie\ it is not trivial. 
Moreover, an easy induction on $n$ yields $\psi \peqe \phi_0 \comp \ldots \comp \phi_n \comp \psi_{n+1}$ for all $n$.

We define $T \eqdef \set{n \setsthat \psi_n \textnormal{ is a trivial proof term}}$.
There are three cases to consider:
\begin{itemize}
\item 
If $0 \in T$, \ie\ if $\psi$ is a trivial proof term, then it is enough to take $\phi \eqdef src(\psi)$ and refer to Lem.~\ref{rsl:trivial-pterm-peq-src}.

\item
Assume $0 \notin T$ and $T \neq \emptyset$, let $n$ be the minimal element in $T$.
In this case we take $\phi \eqdef \phi_0 \comp \ldots \comp \phi_{n-1}$.
For any $k < \omega$, observe that $\psi \peqe \phi \comp \psi_n$, $\phi \peqe \phi \comp tgt(\phi)$ (\confer\ \peqidright), and $\mind{\psi_n} = \mind{tgt(\phi)} = \omega > k$, \confer\ Lem.~\ref{rsl:trivial-pterm-mind-omega}.
Then Dfn.~\ref{dfn:peqe} allows to assert $\psi \peq \phi$.
Finally, observe that each $\phi_i$ being finite implies that $\phi$ is also a finite \ppterm, \ie\ it verifies $\ppsteps{\phi} < \omega$.

\item
Assume $T = \emptyset$. In this case, for any $i$ Lem.~\ref{rsl:trivial-pterm-mind-omega} implies that $\mind{\psi_i} < \omega$, so that $\phi_i$ is non-trivial.
We take $\phi \eqdef \icomp \phi_i$.
Let $n < \omega$. We have already verified that $\psi \peqe \phi_0 \comp \ldots \comp \phi_n \comp \psi_{n+1}$, and $\phi \peqe \phi_0 \comp \ldots \comp \phi_n \comp \icomp \phi_{n+1+i}$.
On the other hand, an easy induction on $k$ implies $\mind{\psi_k} = \mind{\phi_k} \geq k$ for all $k$, then $\mind{\psi_{n+1}} > n$, and also $\mind{\icomp \phi_{n+1+i}} > n$.
Hence the rule \eqllim\ can be applied to obtain $\psi \peq \phi$.
We conclude by observing that $\ppsteps{\phi_n} < \omega$ for all $n$ implies that $\ppsteps{\phi} \leq \omega$.
\end{itemize}
\end{proof}

\newpage
\bibliography{tesisd}
\bibliographystyle{alpha}

\end{document}